\documentclass[11pt]{article}

\usepackage{graphicx}
\usepackage{multirow}
\usepackage{amsmath,amssymb,amsfonts}
\usepackage{amsthm}
\usepackage{mathrsfs}
\usepackage[title]{appendix}
\usepackage{xcolor}
\usepackage{textcomp}
\usepackage{manyfoot}
\usepackage{booktabs}
\usepackage{algorithm}
\usepackage{algorithmicx}
\usepackage{algpseudocode}
\usepackage{listings}

\usepackage{comment}
\usepackage{tikz}
\usepackage{tcolorbox}
\usepackage{ragged2e}
\usepackage[font=footnotesize,margin=0.38in]{caption}
\usepackage{capt-of}
\usetikzlibrary{positioning, arrows.meta, shapes}
\usepackage{etoolbox}
\usepackage{appendix}
\usepackage{etoc}

\usepackage{bibunits} 
\usepackage{titletoc} 
\usepackage[margin=1.3in]{geometry}
\usepackage{abstract}
\usepackage{authblk} 
\usepackage{fancyhdr}
\usepackage{xurl}
 \usepackage{hyperref}
\hypersetup{
    colorlinks,
    linkcolor=[RGB]{9, 81, 135},
    citecolor=[RGB]{9, 81, 135},
    urlcolor=[RGB]{9, 81, 135}
}

\title{MultistageOT: Multistage optimal transport infers trajectories from a snapshot of single-cell data}  
\author[1*]{Magnus Tronstad}  
\author[2${\dagger}$*]{Johan Karlsson}   
\author[1${\dagger}$*]{Joakim S. Dahlin} 

\affil[1]{\small Department of Medicine Solna, Karolinska Institutet, and Center for Molecular Medicine, Karolinska University Hospital, Stockholm, Sweden}
\affil[2]{\small Department of Mathematics, KTH Royal Institute of Technology, Stockholm, Sweden \vspace{1cm}} 
\affil[$\dagger$]{These authors jointly supervised the work.}
\affil[*]{Corresponding authors. Emails: magnus.tronstad@ki.se, johan.karlsson@math.kth.se, and joakim.dahlin@ki.se}
\date{\today} 
\AtBeginEnvironment{tcolorbox}{\small}

\newcommand{\ones}[1]{\mathbf{1}_{#1}} 

\newcommand{\mR}{\mathbb{R}}
\newcommand{\cX}{\mathcal{X}}

\newtheorem{proposition}{Proposition}[section] 

\newtheorem{remark}{Remark}
\theoremstyle{definition}

\begin{document}

\maketitle   

\begin{center}
    \textbf{Abstract}
\end{center}
\begin{quote}
    \noindent
 \noindent \small Single-cell RNA-sequencing captures a temporal slice, or a snapshot, of a cell differentiation process. A major bioinformatical challenge is the inference of differentiation trajectories from a single snapshot, and methods that account for outlier cells that are unrelated to the differentiation process have yet to be established. We present MultistageOT: a Multistage Optimal Transport-based framework for trajectory inference in a snapshot (\url{https://github.com/dahlinlab/MultistageOT}). Application of optimal transport has proven successful for many single-cell tasks, but classical bimarginal optimal transport for trajectory inference fails to model temporal progression in a snapshot. Representing a novel generalization of optimal transport, MultistageOT addresses this major limitation by introducing a temporal dimension, allowing for high resolution modeling of intermediate differentiation stages. We challenge MultistageOT with snapshot data of cell differentiation, demonstrating effectiveness in pseudotime ordering, detection of outliers, and significantly improved fate prediction accuracy over state-of-the-art. 
 \newline
 
 \textbf{Keywords:} trajectory inference, snapshot, cell fate, pseudotime, single-cell, optimal transport, multiple marginals, multistage, hematopoiesis, cell differentiation,  scRNA-seq
\end{quote}

\clearpage

\captionsetup[table]{width=0.75\textwidth}
\captionsetup[figure]{width=\textwidth, labelfont=bf, labelsep=colon, name=Fig.}
\defaultbibliographystyle{naturemag}
\begin{bibunit}
\section{Background}
Cell differentiation is a fundamental biological process. For example, the continuous differentiation of hematopoietic stem and progenitor cells in bone marrow ensures the replenishment of the blood cells in the adult setting, and deregulation of the same process causes disease. The revolution in single-cell RNA-sequencing now enables the capture of a high-resolution snapshot of the cell differentiation process. Specifically, single-cell RNA-sequencing analysis of a bone marrow sample measures the gene expression profiles of thousands of differentiating cells, forming the basis of a cellular and molecular map of hematopoiesis \cite{watcham2019new, boey2022charting}.

The major challenge for the establishment of a molecular blueprint of cell differentiation is to accurately reconstruct cell differentiation trajectories based on the individual cells' gene expression profiles. The difficulty of tracing cells in time is a consequence of the destructive nature of whole-cell transcriptomics: The measurement of a specific cell's gene expression profile kills the cell, preventing future analysis of the same cell. Computational approaches are therefore necessary to accurately infer a cell's trajectory based on its gene expression measurement. To validate such methods, results from lineage tracing experiments, which allow the tracking of the progeny of individual cell clones, can be used as proxy for each cell's ground truth fate \cite{weinreb2020lineage, vanhorn2021next}.

Optimal transport was recently identified as constituting a unifying mathematical framework for solving many challenges related to the analysis of single-cell omics data, including trajectory inference \cite{bunne2024optimal}. For example, several optimal transport-based methods have successfully been established to identify connections between cells sampled from different stages (consecutive time points) of a developmental process by computing optimal transport maps \cite{schiebinger2019optimal, yang2020, forrow2021lineageot}. Cell differentiation in adult bone marrow is a steady-state-like process, in which new hematopoietic stem and progenitor cells are continuously produced. As a result, the entire differentiation process from hematopoietic stem cells to lineage-committed progenitors is captured in a single snapshot. Inferring trajectories within a single snapshot is based on the view that the cells' gene expression profiles represent feasible states for developing cells to occupy during differentiation. StationaryOT successfully used this concept to predict the fate probabilities of differentiating cells, solving a classical bimarginal (single-stage) optimal transport problem \cite{zhang2021optimal}.

While an optimal transport formulation between two marginals has elegantly been implemented for single-cell trajectory inference, single-stage optimal transport cannot model temporal progression. Generalizations of optimal transport with multiple marginals \cite{haasler2021control} address this limitation, but such  generalizations have yet to be explored in the context of single-cell analyses. Here, we describe MultistageOT (\url{https://github.com/dahlinlab/MultistageOT}), a novel multistage optimal transport framework developed for trajectory inference from individual snapshots of the cell differentiation process. Simulated data is used to build and test the framework and we challenge the framework to predict the cell fates of hematopoietic stem and progenitor cells. To validate the method, we benchmark the prediction in a data set with lineage tracing information. Notably, MultistageOT introduces a pseudotemporal dimension to the analysis and outperforms StationaryOT in predicting the fates of the differentiating cells in a large-scale data set of hematopoiesis. Finally, we demonstrate how the MultistageOT framework can be extended for outlier detection---using it to identify cell states that do not belong to the main cell differentiation process. Taken together, MultistageOT is a highly flexible framework for analyzing single-cell RNA-sequencing data, inferring cell fate potential, pseudotime, and outliers.

\section{Results}
Here, we develop MultistageOT for the inference of cell differentiation trajectories in a snapshot of single-cell data. We give a brief explanation of the MultistageOT framework, test it using synthetic data, validate it using results from lineage tracing experiments of \textit{in vitro} hematopoiesis, and apply it to two different \textit{in vivo} snapshots of hematopoiesis.

\subsection{Trajectory inference MultistageOT}
Single-cell RNA-sequencing measures the gene expression profiles of individual cells. In a snapshot of single-cell data, we view the measured gene expression profiles as representing possible cell states during a differentiation process. Let the squared Euclidean distance between any two cell states correspond to the cost of transitioning between them. Note that this cost corresponds to the negative log-likelihood of transitioning under a Gaussian probability kernel. By minimizing a total transition cost, discrete bimarginal optimal mass transport can be used to match single cells sampled at two different stages of a differentiation process (\hyperref[box:ot]{Box 1}). However, cell differentiation involves multiple transitions. MultistageOT addresses this by extending bimarginal optimal transport (\hyperref[box:ot]{Box 1}) to transport across multiple (more than two) marginals (\hyperref[box:msot]{Box 2}), allowing cell differentiation in a snapshot to be modeled as a stepwise transportation process over multiple differentiation stages. MultistageOT leverages prior knowledge of the earliest and latest stages of differentiation to establish directionality in the modeled process, discouraging spurious backward transitions. Notably, MultistageOT links all differentiation stages in a non-greedy manner through a global entropy-regularized mass transport formulation, and we derive an efficient algorithm for finding the globally optimal transport scheme (\hyperref[alg:baseblockascent]{Algorithm \ref{alg:baseblockascent}} in \hyperref[sec:supplementary_note]{Supplementary Note}).

MultistageOT leverages mass transport to infer the optimal transitions between cell states. The transport of mass in each differentiation stage can be interpreted in terms of transition likelihoods: the more mass a cell sends in a particular differentiation stage, the more likely it is to transition in that stage. This induces a probability distribution that models how probable a cell is to transition to a nearby cell state in each differentiation stage (\hyperref[sec:methods]{Methods}). MultistageOT thus provides a natural basis for pseudotemporal ordering of the cells and for inferring cell fate potential (\hyperref[sec:methods]{Methods}).  

\begin{figure}[H]
\begin{tcolorbox}[colback=white,colframe=black,rounded corners]
\textbf{Box 1: Optimal Transport}
\newline

 In the classical discrete optimal transport setting, two discrete mass distributions are given:
\[
 \sum_{i=1}^{n_1} \mu^{(i)}_1\delta_{x^{(i)}_1}, \quad  \sum_{j=1}^{n_2} \mu^{(j)}_2\delta_{x^{(j)}_2}, 
\]
 where the ordered sets $\mathcal{X}_1=(x_1^{(i)})_{i=1}^{n_1}$ and $\mathcal{X}_2=(x_2^{(j)})_{j=1}^{n_2}$ correspond to the points of support of the respective distributions, and where $\mu_1^{(i)}$ and $\mu_2^{(j)}$ represent the mass in the points $x_1^{(i)}$ and  $x_2^{(j)}\!,$ respectively.

\centering
\includegraphics[scale=0.33]{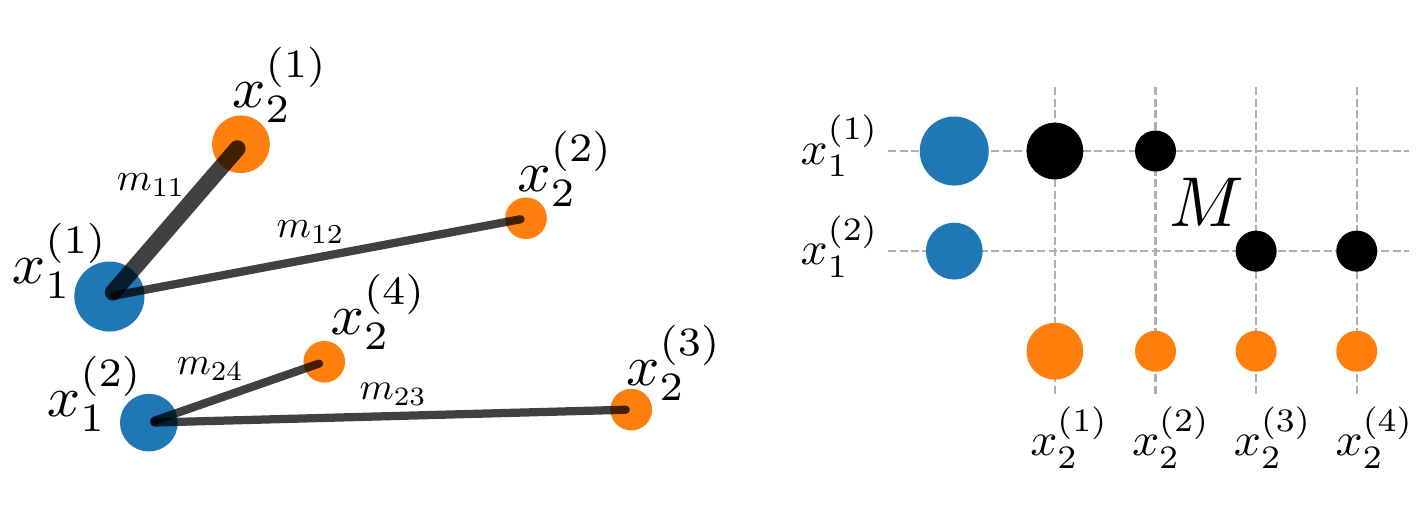} 
\captionsetup{width=\linewidth}
\captionof{figure}{Optimal transport finds an optimal assignment between points in two point cloud distributions $\mu_1 \in \mathbb{R}_+^2$ (blue dots) and $\mu_2 \in \mathbb{R}_+^4$ (orange dots). The two distributions $\mu_1$ and $\mu_2$ could for example represent two cell populations. The solution is a matrix, $M\in \mathbb{R}_+^{2\times 4}$, encoding the optimal way to redistribute the mass in $\mu_1$ to form $\mu_2$. The thickness of the black lines in the left plot represents the amount of mass sent between the points in the optimal solution under the squared Euclidean cost (\ref{eq:box_sqeuclidean_cost}). The corresponding transport plan's matrix elements in the right plot can be interpreted as the strength of the association between any pair of points $x_1^{(i)} \in \mathcal{X}_1, x_2^{(j)} \in \mathcal{X}_2$. }

\hfill 
\justify{
A transport plan is a matrix, $M=[m_{ij}]_{i=1, j=1}^{n_1, n_2} \in \mathbb{R}_+^{n_1\times n_2}$, whose component $m_{ij}$ denotes the amount of mass transported from  $x^{(i)}_1 \in \mathcal{X}_1$ to $x^{(j)}_2 \in \mathcal{X}_2$. We say that the transport plan is feasible if the total amounts transported are consistent with the initial and final distributions, i.e., if $M\ones{n_2}=\mu_1$ and $M^T\ones{n_1}=\mu_2$, where $\mu_1 = [\mu_1^{(i)}]_{i=1}^{n_1}$ and $\mu_2 = [\mu_2^{(j)}]_{j=1}^{n_2}$ are the vectors representing the mass distributions.
Let  
\begin{equation} 
\label{eq:box_sqeuclidean_cost}
c_{ij} :=  \lVert x_1^{(i)} - x_2^{(j)} \rVert_2^2
\end{equation}
denote the cost of moving a unit of mass from $x^{(i)}_1$ to $ x^{(j)}_2$\!, and let $C=[c_{ij}]_{i=1, j=1}^{n_1, n_2}$ be the corresponding cost matrix.  }
The optimal transport problem is to find a feasible transport plan that solves
\justify{\begin{subequations}
\begin{align}
  \mathcal{T}(\mu_1,\mu_2) := \underset{M \in \mathbb{R}_+^{n_1\times n_2}}{\text{minimize}}  \quad &\langle C, M  \rangle 
  \\
   \text{subject to} \,\, \quad  &M\ones{n_2} \;\;=\mu_1 \\
   &M^T \ones{n_1} = \mu_2,
   \end{align}
\end{subequations}
where $\langle \cdot, \cdot \rangle$ is the Frobenius inner product. An example of such a problem, together with a corresponding optimal solution, is given in Fig. 1. Note that when $\mu_1$ and $\mu_2$ are probability measures, then $\mathcal{T}(\mu_1,\mu_2)$ defines the Wasserstein-2 metric between $\mu_1$ and $\mu_2$. In practice, it is common to solve an entropy-regularized version of the problem using Sinkhorn algorithms \cite{cuturi2013sinkhorn}, corresponding to coordinate ascent on the dual function \cite{karlsson2017generalized}.}
\end{tcolorbox}
\label{box:ot}
\end{figure}

\begin{figure}[H]
\begin{tcolorbox}[colback=white,colframe=black,rounded corners]
\textbf{Box 2: The multistage optimal transport problem}
\newline

\begin{minipage}{0.49\textwidth} 
\justify{We assume a single-cell RNA sequencing snapshot of a developmental process is given, with corresponding sets of initial states (e.g., stem cells), denoted $\mathcal{X}_0$, and terminal states (e.g., lineage committed cell types), denoted $\mathcal{X}_F$. We refer to all other cells as intermediate states, denoted $\mathcal{X}$ (see Fig. 2). Let $n_0 := |\mathcal{X}_0|$, $n := |\mathcal{X}|$ and $n_F := |\mathcal{X}_F|$ denote the cardinalities of each set. }
\end{minipage}
\hfill
\begin{minipage}{0.45\textwidth} 
\centering
\includegraphics[scale=0.4]{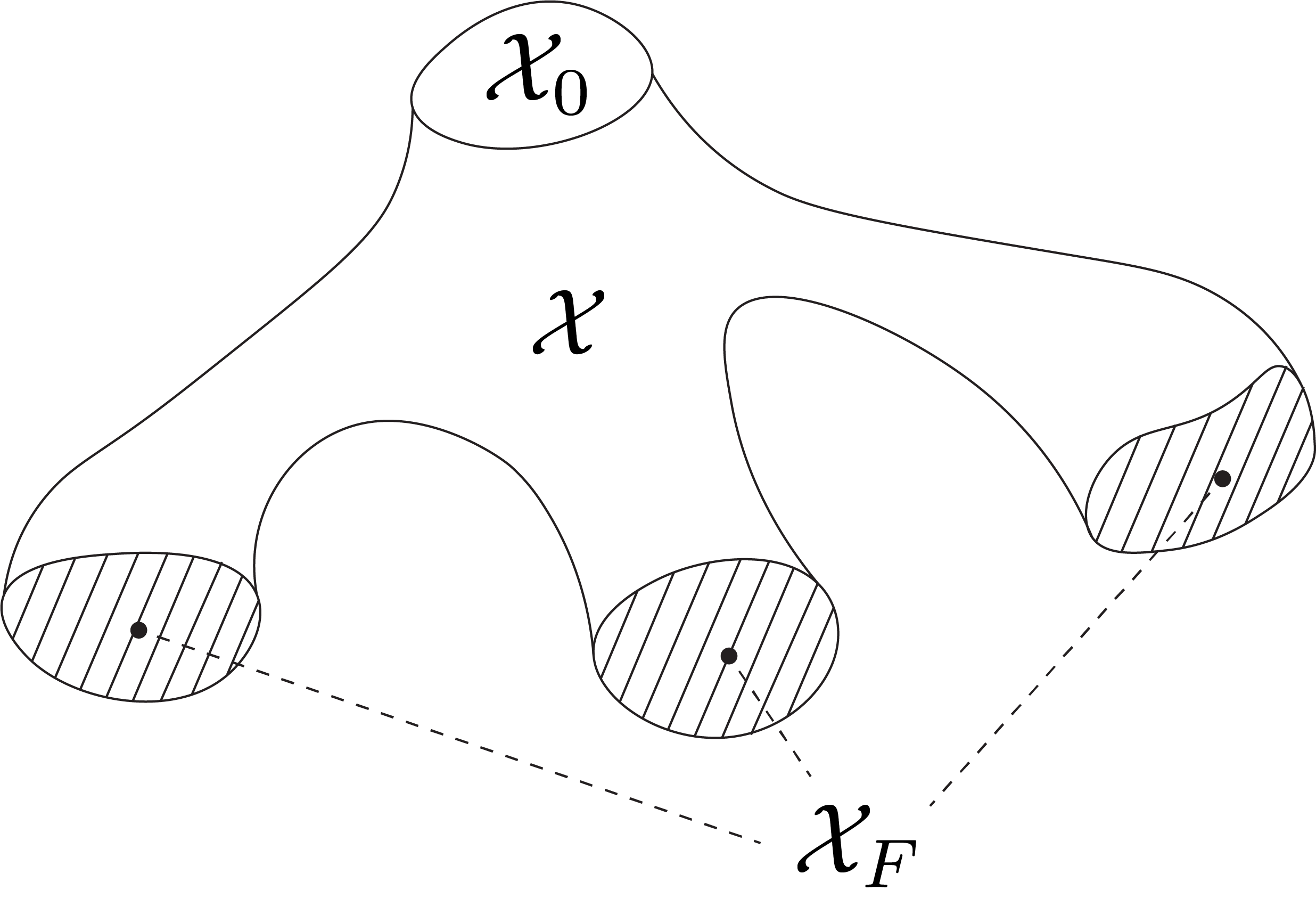}
\captionsetup{width=1\linewidth}
\captionof{figure}{Conceptual figure of a scRNA-seq data set partitioned into three subsets: $\mathcal{X}_0$ (initial states), $\mathcal{X}$ (intermediate states), and $\mathcal{X}_F$ (terminal states)}
\end{minipage}
\vspace{0.25cm}
\begin{flushleft}
    \justify{We model the transitions of cells from $\mathcal{X}_0$ to $\mathcal{X}_F$ as a multistage mass transport process, represented below in terms of marginals $\tilde{\mu}_t$ (mass that continues in the system) and $\hat{\mu}_t$ (mass that exits the system) in the $t$:th stage of transport (Fig. 3).}
\end{flushleft}
\vspace{0.25cm}
\centering

\begin{tikzpicture}[>=latex, node distance=2cm, circleNode/.style={draw, circle, inner sep=2pt}]

    \node[circleNode] (X0) at (0, 0) {$\mathcal{X}_0$};
    \node[circleNode] (X1) [right=1.25 cm of X0] {$\mathcal{X}$};
    \node[circleNode, fill={rgb,255:red,235;green,235;blue,235}] (XT1) [below=1.25 cm of X0] {$\mathcal{X}_F$};
    \node[circleNode] (X2) [right=1.25 cm of X1] {$\mathcal{X}$};
    \node[circleNode, fill={rgb,255:red,235;green,235;blue,235}] (XT2) [below=1.25 cm of X1] {$\mathcal{X}_F$};
    \node[circleNode] (X3) [right=2.25 cm of X2] {$\mathcal{X}$};
    \node[circleNode, fill={rgb,255:red,235;green,235;blue,235}] (XT3) [below=1.25 cm of X2] {$\mathcal{X}_F$};

    \node[circleNode, fill={rgb,255:red,235;green,235;blue,235}] (XT4) [below=1.25 cm of X3] {$\mathcal{X}_F$};

    \node[] (dotsXT-1) [right=0.75 cm of X2] {$ \dots $};
      \node[] (dotsXTT-1) [below = 1.4 cm of dotsXT-1] {$\quad\dots \quad$};

    \node[circleNode] (XT-1) [right=1.25 cm of X3] {$\mathcal{X}$}; 
    \node[circleNode, fill={rgb,255:red,235;green,235;blue,235}] (XTTT) [below=1.25 cm of XT-1] {$\mathcal{X}_F$};

    \draw[->] (X0) -- node[midway, right, yshift=0.3cm] {$\hat{\mu}_0$} (XT1);
    \draw[->] (X0) -- node[midway, above, xshift=-0.25cm] {$\tilde{\mu}_0$} (X1);
    \draw[->] (X1) -- node[midway, right, yshift=0.3cm] {$\hat{\mu}_1$} (XT2);
    \draw[->] (X1) -- node[midway, above, xshift=-0.25cm] {$\tilde{\mu}_1$} (X2);
    \draw[->] (X2) -- node[midway, right, yshift=0.3cm] {$\hat{\mu}_2$} (XT3);
    \draw[->] (X2) -- node[midway, above] {$\tilde{\mu}_{2}$} (dotsXT-1);
    \draw[->] (dotsXT-1) -- node[midway,above] {} (X3);
    \draw[->] (X3) -- node[midway, right, yshift=0.3cm] {$\hat{\mu}_{T-2}$} (XT4);
    \draw[->] (XT-1) -- node[midway, right, yshift=0.3cm] {$\hat{\mu}_{T-1}$} (XTTT);
       \draw[->] (X3) -- node[midway,above, xshift=-0.1cm] {$\tilde{\mu}_{T-2}$} (XT-1);
\end{tikzpicture}

\captionsetup{width=\linewidth}
\captionof{figure}{Graph representation of our general multistage optimal transport framework to model cell differentiation. In each time step $t = 0,1,\dots, T-1$ mass can be sent either to intermediate states in $\mathcal{X}$ or to terminal states in $\mathcal{X}_F$. Mass sent to $\mathcal{X}$ stays within the system, whereas any mass sent to $\mathcal{X}_F$ permanently leaves the system.}
\phantom{notext}
\vspace{-0.5cm}
\begin{flushleft}
 \justify{We compute the optimal transitions between the cells by solving the following multistage optimal transport problem (MultistageOT):}
\end{flushleft}
\begin{equation*}
 (\text{MultistageOT}) \quad \begin{aligned} 
  \underset{\underset{t=0,\dots,T-1 }{\tilde{\mu}_t,\hat{\mu}_t\hat{\nu}_t}}{\text{minimize}} \quad  &\sum_{t=0}^{T-2}\tilde{\mathcal{T}}_t (\tilde{\mu}_t,\tilde{\mu}_{t+1}+\hat{\mu}_{t+1}) + \sum_{t=0}^{T-1}\hat{\mathcal{T}}_t (\hat{\mu}_t,\hat{\nu}_{t})  \\
	\text{subject to}     \quad   &\tilde{\mu}_{0} + \hat{\mu}_{0}  \geq \ones{n_0} \\
&\sum_{t=1}^{T-2} \tilde{\mu}_{t} + \sum_{t=1}^{T-1}\hat{\mu}_{t}    \geq \ones{n} \\
 &\sum_{t=0}^{T-1} \hat{\nu}_t   \geq \ones{n_F},
	\end{aligned}
\end{equation*}

\justify{where each terms in the objective,  i.e., $\tilde{\mathcal{T}}_t(\cdot, \cdot)$ or $\hat{\mathcal{T}}_t(\cdot, \cdot )$, is the optimal transport cost in a bimarginal optimal transport problem using the squared Euclidean distances between cell's gene expression states as ground metric (see \hyperref[box:ot]{Box 1}), and $\hat{\nu}_t = [(\hat{\nu}_t)_i]_{i=1,\dots,n_F}$ is the mass received by cells in $\mathcal{X}_F$ in stage $t$. In practice, we solve an entropy-regularized version of MultistageOT nested in a proximal point scheme (see \hyperref[sec:supplementary_note]{Supplementary Note} for full details and derivation).}
\end{tcolorbox}
\label{box:msot}
\end{figure}

\subsection{MultistageOT establishes differentiation trajectories in unstructured snapshot data}

As a proof-of-concept, MultistageOT was first evaluated on synthetic data. A set of data points was generated in the 2D unit square, representing cell states from a differentiation process with known initial and terminal states (\autoref{fig:synthetic-data-results}a). We assumed 10 intermediate differentiation stages and employed MultistageOT for finding the optimal transitions from the initial states to the terminal states (using \hyperref[alg:baseblockascent]{Algorithm \ref{alg:baseblockascent}}; see \hyperref[sec:supplementary_note]{Supplementary Note})---hypothesizing that MultistageOT would find a pseudotemporal progression through the states consistent with the shape of the data.

Visualizing the mass transport in the optimal solution revealed a gradual progression through the intermediate cell states (\autoref{fig:synthetic-data-results}b). As expected, cells near the initial states transitioned (transported mass) in earlier differentiation stages compared to cells near the terminal states (\autoref{fig:synthetic-data-results}b). The varied distribution over the differentiation stages (\autoref{fig:synthetic-data-results}c-d) enables a ranking of the cell states based on their mean differentiation stage: Each cell is assigned a normalized pseudotime index between 0 and 1, with 0 corresponding to the earliest mean differentiation stage and 1 to the latest mean differentiation stage (\hyperref[sec:methods]{Methods}). In the 2D data set, the pseudotime values produced a smooth heatmap gradient from the initial to the terminal states (\autoref{fig:synthetic-data-results}e), consistent with the shape of the data.

To infer cell fates, we assume a Markov chain model. First, the general affinity between any two cells is quantified by summing, over all differentiation stages, the mass transported between the cells (\autoref{fig:synthetic-data-results}f). The normalized aggregated mass transport between each cell pair is then interpreted as a transition probability in a stationary absorbing Markov chain (\autoref{fig:synthetic-data-results}g). To infer likely cell fates, we let different groups of terminal states define different classes of absorbing states in the Markov chain. Each class of absorbing states represents a fate (\autoref{fig:synthetic-data-results}h). For each cell, we compute the probability of it being absorbed in the different fates under the Markov chain model, and refer to this as cell fate probabilities (see \hyperref[sec:methods]{Methods} for additional details). 

To assess the utility of the MultistageOT framework in predicting likely cell fates, we considered the four groups of terminal states in each of the diverging arms of the data (\autoref{fig:synthetic-data-results}a) to represent cells committed to four different cell lineages. We computed cell fate probabilities and graphed each cell's predicted fate probabilities as a pie chart on top of their state coordinate. This enabled joint visualization of how lineage commitment changed as a function of cellular state for each lineage (\autoref{fig:synthetic-data-results}i). The estimated cell fate probabilities were consistent with the shape of the data; cells near the initial states reflected multi-potency, whereas cell states located in the diverging arms of the data reflected full commitment to a particular fate (\autoref{fig:synthetic-data-results}i). The MultistageOT algorithm (\hyperref[alg:baseblockascent]{Algorithm \ref{alg:baseblockascent}} in \hyperref[sec:supplementary_note]{Supplementary Note}) depends on a regularization parameter controlling the level of diffusion in the optimal transitions. Shannon entropy quantifies the degree of multipotency reflected in the predicted cell fate probabilities. The degree to which cells exhibit multi-potency, as measured by mean Shannon entropy over all cells, increased with increasing regularization parameter values (\autoref{fig:entropy_regularization_effects}), as expected.

\begin{figure}[H]
    \centering
\includegraphics[width=\linewidth]{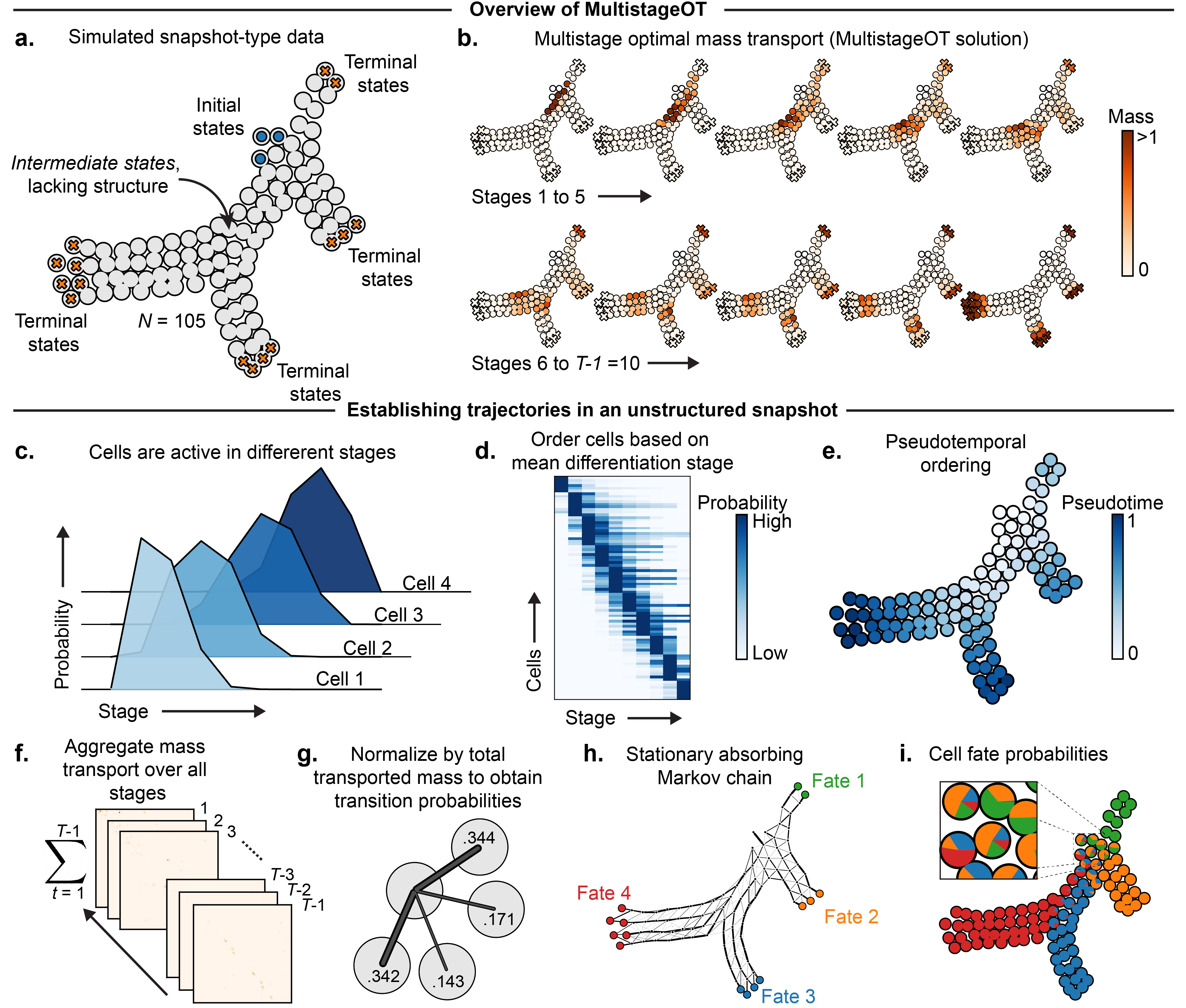}
    \caption{MultistageOT orders simulated snapshot data along a
pseudotemporal axis and induces state transition
probabilities between data points (\hyperref[sec:methods]{Methods}). (a) Two-dimensional data ($N = 105$) was generated in the unit square to mimic samples from an developmental process with known initial and terminal  states. (b) MultistageOT solves an entropy-regularized multistage optimal mass transport problem over $T$ stages, modeling transitions between cell states in a discrete-time differentiation process, where each transport stage models a differentiation step (see \hyperref[sec:supplementary_note]{Supplementary Note} for details). Here, $T = 11$. The plot visualizes the optimal way for cells to transition from initial into terminal states, achieved by plotting the total transported mass out from each cell in each stage $t=1,\dots, 10$. (c) Mass transport varies across the differentiation stages, inducing a probability distribution for each cell over the differentiation stages. The plot shows the probability profiles of four example cells. (d) The varying distributions of cells across the differentiation stages establishes a pseudotemporal order, from early- to late-stage transporters. (e) Cells ordered along a pseudotemporal axis based on mean transport stage. (f) The transport maps are aggregated over all stages to obtain the total likelihood of transition between any cell pair. (g) The aggregated likelihood matrix is normalized such that each matrix element represents a transition probability between a cell pair. (h) The MultistageOT-based transition probabilities define a stationary absorbing Markov chain in which the terminal states represent different classes of absorbing states. (i) Under the MultistageOT-induced Markov chain, each cell has certain probability of being absorbed in each fate. We refer to these as cell fate probabilities, and represent them as a pie chart for each cell. }
    \label{fig:synthetic-data-results}
\end{figure}

The results based on synthetic data encouraged the testing of the MultistageOT framework on a single-cell RNA-sequencing snapshot of a cell differentiation process. The complex process of hematopoiesis was chosen to challenge the framework (\autoref{fig:paul15-results}a). We applied MultistageOT to the data set of Paul et al. \cite{paul2015transcriptional}, selecting multipotent progenitors as initial states and progenitors at the entry points to each of 8 distinct cell lineages as terminal states (\autoref{fig:paul15-results}a-b).  UMAP visualization of the MultistageOT-based results revealed progressively increasing pseudotime values and lineage commitment from the multipotent progenitor region to the cell lineage entry points (\autoref{fig:paul15-results}c-d), results consistent with how hematopoietic stem and progenitor cells differentiate into the various cell types.

\begin{figure}[H]
    \centering
\includegraphics[width=\linewidth]{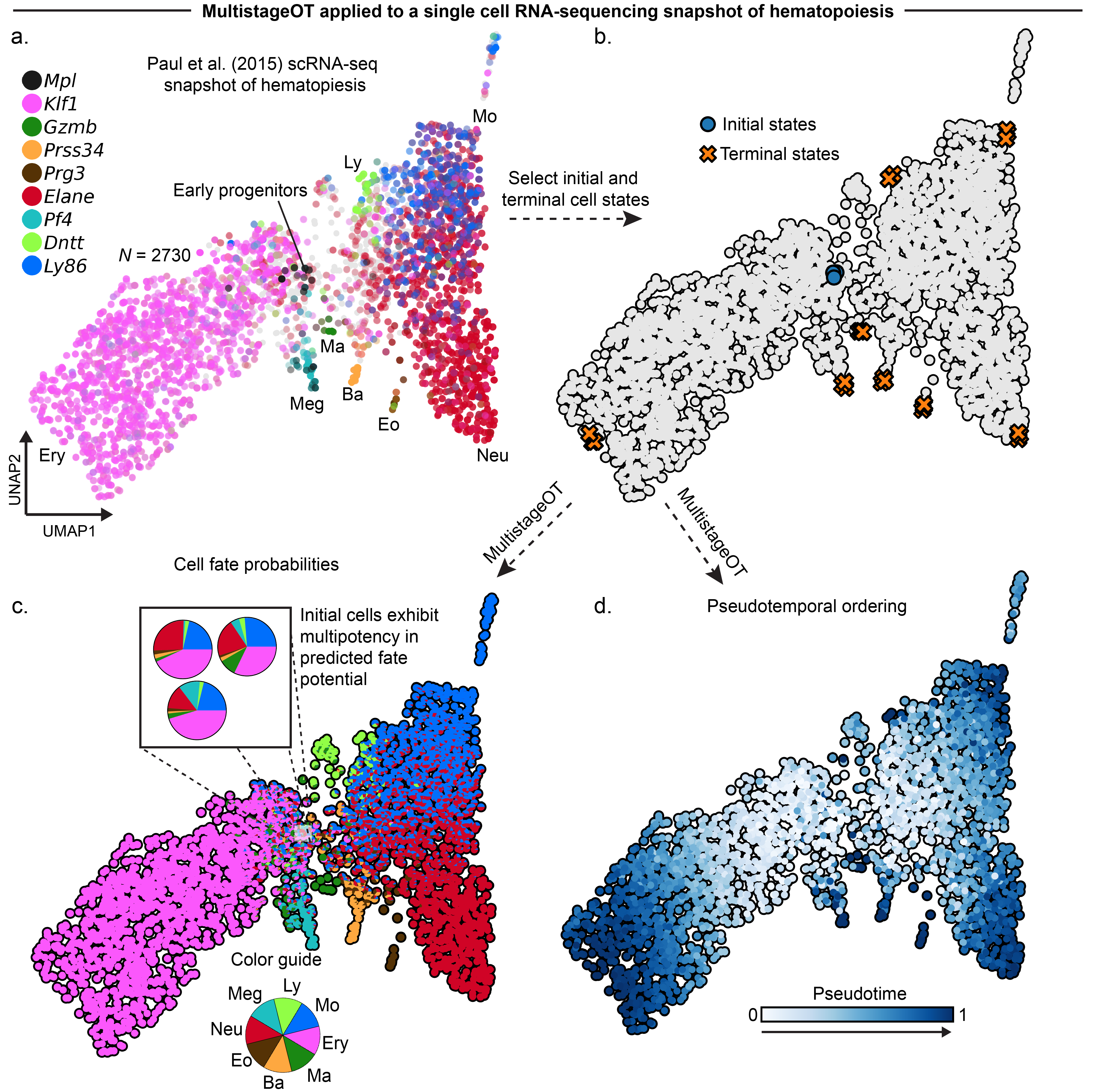}
    \caption{MultistageOT-based inference of fate potential and pseudotemporal ordering of single cells in a publicly available single-cell RNA-sequencing snapshot of hematopoiesis. (a) Layered gene expression level (log($x$+1)-transformed) plot of known marker genes for early progenitors and eight mature blood cell lineages graphed on top of a UMAP embedding of the data from Paul et al. \cite{paul2015transcriptional}. Abbreviations: Meg}
    \label{fig:paul15-results}
\end{figure}

\begin{figure}
    \ContinuedFloat 
    \caption*{(megakaryocyte), Ery (erythroid), Ma (mast cell), Ba (basophil), Eo (eosinophil), Neu (neutrophil), Mo (monocyte), Ly (lymphoid). (b) Marker genes were used to orient the UMAP and choose initial and terminal cell states as input to our MultistageOT model. (c) MultistageOT-based inference of fate potential. Each cell's UMAP coordinate is graphed as pie chart, representing the MultistageOT-based fate probabilities (d) Pseudotemporal ordering of the cells based on the optimal mass transport.} 
\end{figure}

Taken together, MultistageOT infers trajectories in unstructured snapshot data using multistage optimal mass transport. The mass transport induces state transition probabilities between data points, representing an affinity between cell states, as well as a pseudotemporal axis, representing the degree of differentiation.

\subsection{Cell fate predictions based on MultistageOT matches clonal sister fates observed \textit{in vitro} in mouse hematopoiesis}

Thus far, we have shown that MultistageOT can produce reasonable downstream results in a single-cell snapshot. However, the data set of Paul et al. \cite{paul2015transcriptional} (\autoref{fig:paul15-results}) does not include a reference or ground truth to quantitatively assess MultistageOT's performance in predicting cell differentiation. To evaluate the accuracy of MultistageOT, we utilized the data set of Weinreb et al. \cite{weinreb2020lineage}, which provides longitudinal single-cell RNA-sequencing data of \textit{in vitro} hematopoiesis, integrated into a unified snapshot-like landscape of cell states. Notably, the data set includes lineage tracing information of cell clones, which can serve as proxy for the ground truth trajectories of individual cells. 

The Weinreb et al. \cite{weinreb2020lineage} data set consists of 130 887 cells. To satisfy the memory requirements of a standard laptop computer, we uniformly subsampled cells into 12 disjoint subsets of size $N=$ 10 907 cells and MultistageOT was applied to each subset. Initial and terminal cell states were selected (\autoref{fig:weinreb-results}a, and \hyperref[sec:methods]{Methods}), corresponding to multipotent progenitors and the most differentiated progenitors in the landscapes, respectively. Pooling the MultistageOT results from each subset generated a unified cell landscape, in which the cells were pseudotemporally ordered and the fate potential of each cell was predicted  (\autoref{fig:weinreb-results}b-c).

Visualization of the MultistageOT-derived pseudotime values showed progressively increasing pseudotime from undifferentiated progenitors to the mature cells (\autoref{fig:weinreb-results}b,e). Overall, these MultistageOT-derived  pseudotime values were consistent with actual time, showing a significant tendency to increase with the cells' experimental time stamps (\autoref{fig:weinreb-results}g).

Visualization of the MultistageOT-based inference of fate potential of the cells displayed a gradual increase in lineage commitment as the cells approached the terminal states (\autoref{fig:weinreb-results}c), which is expected from differentiating cells. StationaryOT \cite{zhang2021optimal} was applied to the same data set to allow a side-by-side comparison with the state-of-the-art optimal-transport-based trajectory inference method \autoref{fig:stationaryot_and_idw_optimization}a. As reference, we also developed a naive predictor of cell fate referred to as Inverse Distance Weighted (IDW). The IDW method is not based on optimal tranport and instead uses the relative Euclidean distances to the mature cell types to infer the cell fate potential (\autoref{fig:stationaryot_and_idw_optimization}b-c and \hyperref[sec:methods]{Methods}). Similar

\begin{figure}[H]
    \centering 
\includegraphics[width=\linewidth]
{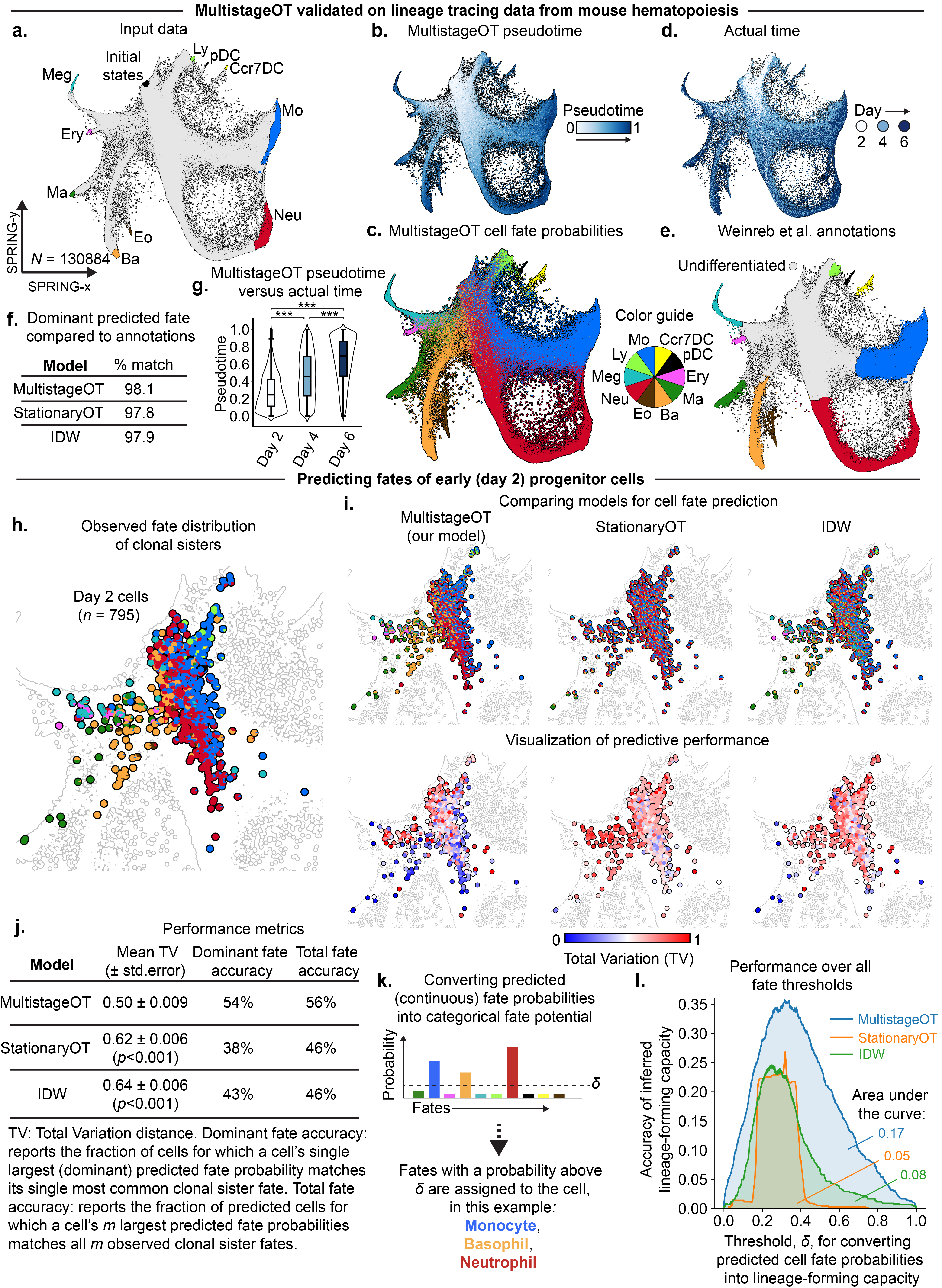}
    \caption{Validation of the MultistageOT framework with \textit{in vitro} lineage tracing data from mouse \makebox[\linewidth][s]{hematopoiesis. (a) SPRING-plot of the data from Weinreb et al. \cite{weinreb2020lineage} highlighted with initial and terminal cell}}
    \label{fig:weinreb-results}
\end{figure}

\begin{figure}
    \ContinuedFloat 
    \caption*{states in different lineages. (b) MultistageOT-based pseudotemporal ordering. (c) MultistageOT-based inference of cell fate potential. Each cell is represented by a pie chart with wedge sizes proportional to the predicted probability of the cell ending up in each fate. (d) Experimental time point for each cell (e) Cell type annotation given by Weinreb et al. \cite{weinreb2020lineage} (f) Accuracy when comparing annotated cell label against dominant predicted fate potential. (g) Box plots of MultistageOT pseudotime. Each box extends from the first to the third quartile of the data. The black horizontal line indicates the median.  Whiskers extend to the farthest pseudotime still within 1.5 times the inter quartile range. Mann-Whitney U-test, ***$p<$0.001. (h) ``Ground truth'' fate distribution of clonal sisters for day 2 cells. (i) Benchmarking fate predictions of MultistageOT (top left), StationaryOT (top middle) and an Inverse Distance Weighted (IDW) model (top right).  Bottom row shows corresponding errors (total variation distance, TV) in the predictions when compared to the distribution of clonal sister fates. (j) Performance metrics for each method. The $p$-values correspond to comparisons to MultistageOT in a Mann-Whitney U-test. (k) Assigning a cell to fates with predicted fate probability above a threshold, denoted $\delta$. (l) Predictive performance for each of the three models over all possible thresholds for $\delta$ (horizontal axis). Vertical axis reports the fraction of cells whose model-based assignment of fate potential matches the observed clonal sister fates.  Abbreviations:  Meg (megakaryocyte), Ery (erythroid), Ma (mast cell), Ba (basophil), Eo (eosinophil), Neu (neutrophil), Mo (monocyte) pDC (plasmocytoid dendritic cell), Ccr7DC (Ccr7$^{+}$ DC) and Ly (lymphoid).} 
\end{figure}
\noindent to MultistageOT, the StationaryOT- and IDW-based predictions both displayed increasing commitment as cells approach the terminal cells (\autoref{fig:weinreb-results}c and \autoref{fig:stationaryot_and_idw_optimization}a,c).

To evaluate the performance of MultistageOT in further detail, we first inspected the cells that Weinreb et al. \cite{weinreb2020lineage} had annotated as mature (\autoref{fig:weinreb-results}e). MultistageOT predicted the dominant cell fate in agreement with the author annotation in 98.1 \% of these cells. StationaryOT and IDW showed similar results, reaching a 97.8\% and 97.9\% match respectively with the author annotations (\autoref{fig:weinreb-results}f). 

We next used the lineage tracing data of Weinreb et al. \cite{weinreb2020lineage} as proxy for the ground truth fate potential to evaluate the performance of MultistageOT in predicting the fates of the most immature progenitors of the data set. We specifically analyzed the progenitors belonging to the earliest time point (day 2) that had clonal sister cells in any of the subsequent time points (days 4 or 6) (\autoref{fig:quantifying_ground_truth}, and \hyperref[sec:methods]{Methods}). Visual comparison between the clonal fates and the three models' predictions---using colors to represent the unique cell fates---revealed striking similarities in the color patterns of MultistageOT compared to the observed fate distribution, with seemingly complete matches in some cells (\autoref{fig:weinreb-results}h-i), whereas StationaryOT showed a tendency to predict substantial multipotency across all day 2 progenitors, with less obvious color matching compared with the observed fates. The color pattern of IDW appeared as a hybrid between the MultistageOT and StationaryOT results (\autoref{fig:weinreb-results}i).

We evaluated several metrics to quantify the degree of similarity between the observed fates and the model predictions  (\autoref{fig:weinreb-results}j). The output of MultistageOT, StationaryOT, and the IDW predictor are probability distributions that for each cell predict the likelihood of each fate. We therefore assessed Total Variation (TV) distance (\autoref{fig:weinreb-results}i, bottom), which measures the total absolute deviations between two probability distributions. In doing so, we took each cell's clonal sister fate distribution in days 4 or 6 as the reference, ``ground truth'', distribution (\hyperref[sec:methods]{Methods}). The mean TV was found to be significantly lower in MultistageOT compared to StationaryOT and IDW (\autoref{fig:weinreb-results}j), indicating that MultistageOT's fate predictions best matched the observed clonal sister fate distributions. We then turned to viewing fate prediction in terms of a classification problem, and benchmarked the accuracy of a model's dominant fate prediction compared to the single most common fate among the cell's clonal sisters in days 4 or 6. We call this the dominant fate accuracy. Notably, MultistageOT achieved higher dominant fate accuracy than StationaryOT and IDW (\autoref{fig:weinreb-results}j).

We next evaluated a metric that accounts for the observation that some progenitors give rise to multiple cell lineages. This metric, which we termed total fate accuracy, quantifies how often a cell's observed clonal sister fates correspond to a model's dominant fate predictions. For example, if a cell's clonal sisters in days 4 or 6 are observed to be neutrophils and monocytes, then the model's prediction for that cell is considered a correct match only if the neutrophil and monocyte fates correspond to the model's two highest predicted fate probabilities. MultistageOT achieved the highest total fate accuracy score compared with StationaryOT and IDW (\autoref{fig:weinreb-results}j). 

As a final evaluation, we assigned lineage-forming capacity if the model-inferred fate probability for a lineage was higher than a threshold $\delta$, which converted the progenitors' model-inferred fate probabilities into categorical data (\autoref{fig:weinreb-results}k). We then compared the inferred lineage-forming capacity to the observed clonal fates of each progenitor for different values of $\delta$. An assignment was considered correct only if the model-inferred lineage-forming capacity and observations matched completely. Notably, MultistageOT performed better than StationaryOT and IDW across all values of $\delta$ (\autoref{fig:weinreb-results}l). 

Taken together, MultistageOT predicts the fate of differentiating progenitors in an integrated snapshot-like data set more accurately than StationaryOT and IDW.

\subsection{MultistageOT identifies outliers and predicts bipotent progenitors in a snapshot of bone marrow hematopoiesis}

Modeling cell differentiation in an \textit{in vivo} snapshot, i.e. a single time point measurement of the primary progenitors' transcriptome, can be associated with challenges. For example, technical limitations in cell isolation can result in the unintentional capture of cells not belonging to the differentiation process, leading to a data set containing outliers. To explore the \textit{in vivo} snapshot setting further, we applied MultistageOT to a publicly available single-cell RNA-sequencing data set of hematopoiesis. This data set comprises the gene expression profiles of more than 44 000 primary hematopoietic stem and progenitor cells isolated from mouse bone marrow \cite{dahlin2018single}. Notably, the data set includes hematopoietic stem cells and 8 author-annotated entry points to distinct hematopoietic cell lineages, which were used to define the initial and terminal states in MultistageOT (\autoref{fig:dahlin2018_annotations} and \hyperref[sec:methods]{Methods}). In computing the transition costs between cells, we discovered a small number of cells at extreme distances from other cells (\autoref{fig:cost_distributions}). Large transition costs cause numerical instability if the distances are too extreme, causing the computer to interpret every possible transition as infeasible. Such instability issues may be combated by increasing the level of regularization in the optimal transport maps. However, this may lead to undesirable levels of diffusion in the optimal transitions if the parameter necessary for stability grows large enough (\autoref{fig:entropy_regularization_effects}). We hypothesized the presence of outlier cells, not representing intermediate states between the selected initial and the terminal states, as the cause of these large transition costs. 

\begin{figure}[H]
    \centering 
\includegraphics[width=\linewidth]{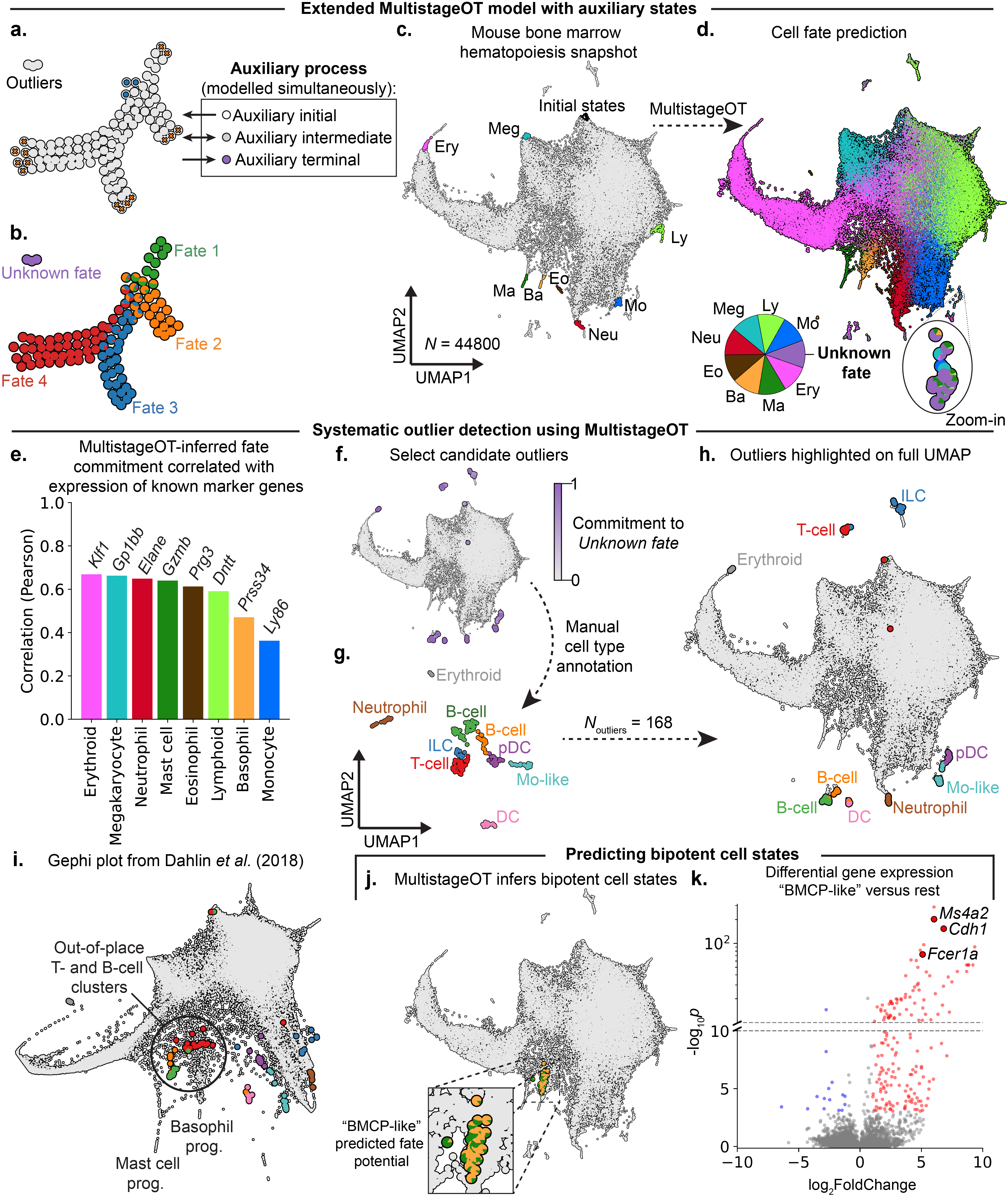}
    \caption{MultistageOT identifies outliers and predicts bipotent progenitors in a snapshot of bone marrow hematopoiesis. (a) The 2D synthetic data set with the addition of a small population of outliers; auxiliary states enables MultistageOT to model the development of outliers separately. (b) MultistageOT-based cell fate probabilities, represented by pie charts. The ``unknown fate'' (purple) represents absorption in the auxiliary terminal state. (c) UMAP of Dahlin et al. \cite{dahlin2018single} single-cell RNA sequencing data from mouse hematopoiesis, highlighted with initial and terminal cell states in different lineages. (d) MultistageOT-based cell fate probabilities. (e) log($x +1$)-expression of each lineage associated gene marker (graphed in \makebox[\linewidth][s]{\autoref{fig:dahlin2018_annotations}) correlated (Pearson) against MultistageOT's predicted commitment to the associated lineage.}}
    \label{fig:outlier-results}
\end{figure}

\begin{figure}
    \ContinuedFloat 
    \caption*{ (f) Predicted commitment to an ``unknown fate'' was used as a means to identify outliers. Cells committed to ``unknown fate'' with a probability above 0.5 are highlighted. (g) $N_\text{outlier} = 168$ candidate outlier cells with ``unknown fate'' probability above 0.5 were selected. The candidate outliers were clustered, represented in a new UMAP embedding and annotated. (h) Annotated outlier cells highlighted on the original UMAP. (i) Annotated outlier cells highlighted on the gephi embedding used in Dahlin et al. (j) MultistageOT predicts mast cell/basophil common progenitors (above $10\%$ in mast cell and basophil commitment, and below $1\%$ in all other fates). (k) Volcano plot showing log$_2$-fold change of gene expression levels against the Benjamini-Hochberg adjusted $p$-value in a Mann-Whitney U-test of log($x$+1)-transformed gene expression in the BMCP-like subset versus the rest of the cells in the snapshot. Positive fold change indicates higher expression in the BMCP-like subset. Genes with absolute fold change above 1.0 and adjusted $p$-value below 0.001 is highlighted (blue and red indicates negative and positive fold change respectively). Note the broken $y$-axis. Abbreviations: Meg (megakaryocyte), Ery (erythroid), Ma (mast cell), Ba (basophil), Eo (eosinophil), Neu (neutrophil), Mo (monocyte), Ly (lymphoid), DC (dendritic cell), pDC (plasmocytoid DC), Mo-like (monocyte-like), ILC (innate lymphoid cell), prog. (progenitor). } 
\end{figure}

To be able to identify and account for potential outliers, we developed an extension of the MultistageOT model. The extension works by introducing three auxiliary states (corresponding to an initial, intermediate, and terminal cell state) to which cells will transition if the main transportation process is associated with too high cost. This allows outliers to be modeled in a separate auxiliary process (\autoref{fig:outlier-results}a), and stabilizes the algorithm with respect to extreme transition costs, without necessitating undesirably large levels of regularization. To test this extension, we added outliers to the simulated data in \autoref{fig:synthetic-data-results}a. To mark cells as outliers, we let the auxiliary terminal state represent an ``unknown fate'', expecting MultistageOT to highlight outliers with strong commitment to this ``unknown fate'' (\autoref{fig:extension_synthetic_data}). Notably, the extension of MultistageOT was able to identify the outlier cells in the simulated data (\autoref{fig:outlier-results}a-b).

The \textit{in vivo} hematopoiesis data set \cite{dahlin2018single} (\autoref{fig:outlier-results}c) was subsampled into four disjoint subsets of size $N =$ 11 200 cells and MultistageOT was applied to each subset. Pooling the results and plotting the fate probabilities for each cell highlighted the progression from multipotent progenitors to lineage-committed cells and revealed potential outlier cells that were marked by MultistageOT as committed to the ``unknown fate'' (\autoref{fig:outlier-results}d). The hematopoietic progenitors' ground truth fate potential is unavailable in a single snapshot of hematopoiesis. To assess the validity of the inference, we therefore correlated the MultistageOT-derived fate potential with established lineage marker genes. The results revealed positive correlation for each of the eight mature cell lineages and the corresponding marker gene (\autoref{fig:outlier-results}e).

A small number of cells were predicted to have $50\%$ or higher probability to be absorbed in an unknown fate (\autoref{fig:outlier-results}f). We considered such cells candidate outliers. To inspect these cells further, we visualized them in a new UMAP embedding and performed Leiden clustering analysis (\autoref{fig:outlier-results}g). The candidate outlier clusters were annotated based on the expression levels of marker genes (\autoref{fig:dahlin2018_gene_annotations_1}-\autoref{fig:dahlin2018_gene_annotations_full_umap_2}). Two of the clusters expressed genes representative of cells in late stages of neutrophil and erythrocyte differentiation, respectively  (\autoref{fig:dahlin2018_gene_annotations_1}d-e). As these neutrophil and erythroid cells are formed beyond the entry points to each respective lineage identified by Dahlin et al. \cite{dahlin2018single}, they do not represent intermediate differentiation stages in the snapshot and MultistageOT appropriately flagged them as potential outliers. The annotation also revealed the presence of cells with lineage programs representative B cells, T cells, and innate lymphoid cells (ILCs) (\autoref{fig:outlier-results}g). These lineages emerge past the lymphoid lineage entry point that was identified in Dahlin et al. \cite{dahlin2018single}, further evidence that the candidate cells constituted outliers in the analyzed differentiation process. We observed that the B-, T-, and ILC-like cells occupied distinct regions of the UMAP, separate from the main single-cell landscape and the lymphoid entry point. This served as an additional indicator that the identified cells did not belong to the differentiation process. However, we also note that MultistageOT is independent of data visualization approach and can therefore reveal outliers in situations where the data visualization fails to clearly distinguish them. The relevance of this becomes clear when the data is visualized with the original gephi embedding used in Dahlin et al. \cite{dahlin2018single}. In the gephi-plot, the B- and T-like cells are found out-of-place as they position near the entry point to the basophil and mast cell differentiation trajectories (\autoref{fig:outlier-results}i).

To chart lineage diversification in a differentiation process, the identification of bipotent progenitor stages is crucial. Dahlin et al. \cite{dahlin2018single} identified, isolated, and characterized bipotent basophil-mast cell progenitors (BMCPs) in bone marrow through an educated guess of the progenitors' phenotype. Here, we investigated whether MultistageOT's data-driven inference of fate potential predicted the existence of bipotent BMCPs in this snapshot. Visualization of cells predicted by MultistageOT to possess specific basophil- and mast cell-forming capacity (\hyperref[sec:methods]{Methods}) revealed a population of BMCP-like cells located near the basophil and mast cell entry points on the UMAP (\autoref{fig:outlier-results}j). Differential gene expression analysis between the MultistageOT-inferred BMCP-like cells and all other cells (\autoref{fig:outlier-results}k) revealed high expression of the genes coding for the $\alpha$ and $\beta$ chains of Fc$\epsilon$RI (\textit{Fcer1a}, $\textit{Ms4a2}$) and E-cadherin (\textit{Cdh1}), genes shown to mark progenitors with basophil- and mast cell-forming capacity \cite{doi:10.1126/sciimmunol.aba0178, dahlin2018single}.

Taken together, MultistageOT infers cell fate potential consistent with known gene markers, detects outlier populations, and predicts bipotent progenitors in a single-cell RNA-sequencing snapshot of cell differentiation.

\section{Discussion}

Data driven computational methods that reliably infer hidden developmental trajectories from a single-cell RNA-sequencing snapshot are necessary for charting the cell differentiation process at the molecular level. In this work, we presented a multistage optimal transport model of cell differentiation trajectories.

Following established methods in trajectory inference (e.g., \cite{zhang2021optimal,trapnell2014dynamics, street2018slingshot,wolf2019paga, setty2019characterization, lange2022cellrank}), MultistageOT takes the view of a single-cell snapshot as a graph, where nodes represent feasible differentiation states to which cells can transition, and edge weights represent affinities between states. Since cell differentiation underlies a progression from undifferentiated states to more differentiated states, a central property of computational models that seek to infer biological relationships between cell states in a snapshot is directional bias. For example, naively building affinities between cells through a $k$-nearest neighbour graph based on Euclidean distances alone would fail to capture this directional bias as symmetries in neighbour proximity generally lead to undirected graphs. Directionality can be imprinted on a cell-cell affinity graph in different ways: by explicitly incorporating directional data in the affinity calculations, e.g. based on the RNA velocity \cite{la2018rna} model, as in Lange et al. \cite{lange2022cellrank}; or by applying \textit{post hoc} modifications on the graph, pruning edges that do not conform with a secondary pseudotime model \cite{setty2019characterization}; or by leveraging prior knowledge of initial and terminal cell states \cite{zhang2021optimal}. Building on this latter strategy, MultistageOT represents a novel approach to computing the cell-cell affinities. A major departure from previous methods is that, in MultistageOT, cell differentiation is modeled through multiple, consecutive, mass transport steps, establishing a pseudotemporal axis here demonstrated to be consistent with real time course measurements. 

We derived an algorithm for applying MultistageOT to a single-cell RNA-sequencing snapshot. The discretization of the differentiation process into multiple differentiation stages in MultistageOT enables a large degree of control and modeling flexibility, as constraints can be imposed on each differentiation stage. Notably, our MultistageOT implementation showed improved predictive performance over the state-of-the-art StationaryOT \cite{zhang2021optimal} algorithm  when inferring fate potential of early hematopoieitic progenitor cells. One limitation in our implementation of MultistageOT is the memory requirement, leading us to subsample the larger data sets to around $10^4$ cells. However, the subsampling approach had the advantage of allowing us to affirm that MultistageOT produced consistent annotations across multiple subsamples. For increased robustness, the data could in principle be repartitioned multiple times and estimates could be averaged over the resamplings for each cell \cite{zhang2021optimal}. Note that the memory bottle neck remains an issue in standard optimal transport formulations as well, as the number of decision variables generally grows quadratically with the number of points in each support, whereas it grows linearly with the number of differentiation stages in MultistageOT for a fixed number of cells. However, recent work \cite{klein2025mapping} has proposed leveraging online computation of the optimal transport costs to reduce memory cost, combined with GPU acceleration to combat the accompanying increase in compute time. Although demonstrated in Klein et al. \cite{klein2025mapping} for a different setting than ours, we note that the benefits might extend to broader classes of optimal transport formulations, and we thus see potential in future implementations of MultistageOT to experiment with similar techniques.

A single-cell RNA-sequencing snapshot may feature disconnected clusters of cells, or even individual cells, that lie far away from other cells in gene expression space. In previous work, the issue of disconnectedness in a snapshot has been addressed by leveraging clustering algorithms and identifying disconnected clusters on the basis of a ``connectivity measure'' \cite{wolf2019paga}. However, methods have yet to be established that can automatically detect and highlight single cells that do not belong to the differentiation process of interest. Such outliers can create spurious trajectories through the data, which in turn may lead to incorrect conclusions about the underlying cell differentiation process.  MultistageOT takes a novel approach to address this issue:  Auxiliary states act as layer of stem, intermediate, and mature cells that bridges the disconnected cells to a hidden process, allowing for identification of outliers at single-cell resolution on the basis of inferred cell fate potential. Using our algorithm, MultistageOT was employed for cell fate inference in a snapshot of \textit{in vivo} bone marrow hematopoiesis, allowing us to identify true outliers (in the sense of lineage commitment), demonstrating the effectiveness of this approach. Moreover, the results also highlight  limitations of low-dimensional representations of high-dimensional data, as the B- and T-like cells detected by MultistageOT are poorly resolved in the original force-directed graph embedding published in Dahlin et al. \cite{dahlin2018single}.

Finally, applying MultistageOT to predict cells belonging to a previously described, but rare, bipotent progenitor population allowed us to establish the transcriptional composition of the predicted cells at the single-cell level. Recent advances in methods enabling paired transcriptome and surface protein measurements for individual cells (e.g., CITE-seq \cite{stoeckius2017simultaneous}) in principle enables analogous characterization of differentially expressed surface markers of the predicted cells. In future work, we thus hope to see utilization of multi-omics data for model-guided inference of fluorescence-activated cell sorting strategies to isolate the progenitors of interest. Cell fate assays are then used to investigate and validate the progenitors' predicted fate potential. 

Taken together, we believe that MultistageOT has the potential to be utilized in the discovery, characterization, and validation of uncharted differentiation stages of any developmental process marked by well-defined beginning and end states.

\section{Conclusions}
MultistageOT models cell differentiation in a snapshot as a stepwise process over multiple differentiation stages. MultistageOT is a flexible modeling framework that orders cells in pseudotime and infers differentiation fates. It outperforms a previously published bimarginal optimal transport model when predicting fates of early progenitor cells in a lineage tracing data set from mouse hematopoeisis. MultistageOT can also identify rare cell populations not representing intermediate differentiation states to identified mature cell lineages, and predicts the transcriptional states of a previously described bipotent progenitor population.

\section{Acknowledgements}
The Swedish Research Council (2022-00558), the Swedish Cancer Society, and Karolinska Institutet funded the research in JSD's team. The Swedish Research Council (2020-
03454) and KTH Digital Futures funded the research for JK. We thank Lovisa Julin for conceptualizing an early version of the method during her master thesis work.

\section{Methods}
 \label{sec:methods}

\subsection{Data preprocessing}
All single-cell RNA-sequencing data were preprocessed in Python 3.9 using the Scanpy module (version 1.9.3).

\subsubsection{Generation of synthetic data} 
For the simulated data in \autoref{fig:synthetic-data-results}, we manually generated $N=105$ data points in the unit square to reflect a simple differentiation process with known initial and terminal cell states, utilizing the online tool: \url{https://guoguibing.github.io/librec/datagen.html}.

\subsubsection{Paul et al. (2015) \textit{in vivo} snapshot of mouse hematopoiesis } 
The Paul et al. \cite{paul2015transcriptional} data is integrated in the Scanpy python module, and was loaded using the function \texttt{scanpy.datasets.paul15}. We performed the preprocessing steps outlined in the Scanpy tutorial in \url{https://scanpy-tutorials.readthedocs.io/en/latest/paga-paul15.html}. In particular, the data was converted to float64 (as per the recommendations in the tutorial) and finally preprocessed using the function \texttt{scanpy.pp.recipe\_zheng17}. This procedure retained 1000 highly expressed genes and then normalized cells by total counts, after which log($x$+1)-transformation and standardization to unit variance and zero mean was performed.

\subsubsection{Weinreb et al. (2020) \textit{in vitro} lineage tracing in mouse hematopoiesis} 

The lineage tracing data from Weinreb et al. \cite{weinreb2020lineage} was downloaded as a normalized count table (L$_1$-normalized counts) from the authors' Github repository: \url{https://github.com/AllonKleinLab/paper-data/tree/master/Lineage\_tracing\_on\_transcriptional\_landscapes\_links\_state\_to\_fate\_during\_differentiation}. To start, the data frame contained 25289 genes and 130887 cells.  Highly variable genes were selected using Scanpy's \texttt{highly\_variable\_genes} function (\texttt{flavor='seurat', min\_mean=0.001, max\_mean=5, min\_disp=0.05}) leaving us with a total of 10220 genes. The data was then log($x$+1)-transformed and scaled to unit variance and zero mean. PCA was performed on this preprocessed data keeping the first 50 principal components. Finally, we randomly omitted three data points from the data set in order to have the total number of cells divisible by 12, and then created a partition of 12 disjoint subsets of the data, each of size 10907. MultistageOT was applied to each partition independently, and the results were then pooled when creating the 2D visualizations of pseudotime and cell fate probabilities in \autoref{fig:weinreb-results}.

\subsubsection{Dahlin et al. (2018) \textit{in vivo} snapshot of mouse hematopoiesis }
Data from Dahlin et al.  \cite{dahlin2018single} was downloaded from \url{http://gottgens-lab.stemcells.cam.ac.uk/adultHSPC10X/}. Doublets were removed by referencing the doublet score file in the publicly available data. Cells were removed based on the following critera: having more than 10\% mitochondrial reads, having the total number of reads $\pm$ 3 standard deviations from the mean (in log-scale) and having fewer than 500 genes expressed. Genes expressed in fewer than 3 cells were then removed. The resulting gene expression data frame featured 17633 genes in 44802 cells. The data was normalized by total counts so that each cell's expression levels summed to $10^{4}$. Subsequently, we identified 5033 highly variable genes based on using Scanpy's \texttt{filter\_genes\_dispersion}-function (with parameters \texttt{log=True, flavor='seurat', 
min\_mean=0.001, max\_mean=5, min\_disp=0.05, \\ copy=True}).
The remaining genes were then log($x$+1)-transformed and scaled to unit variance and zero mean. Finally, genes were correlated against a reference of cell cycle genes provided by Dahlin et al. \cite{dahlin2018single}, keeping only genes with with 0.2 or lower Pearson's correlation coefficient. This removed a total of 368 genes. At this stage, the data featured a total of 4665 genes in 44802 cells. PCA was performed on this processed data, keeping the first 50 principal components. Finally, a UMAP was created by first computing a neighbour graph with \texttt{scanpy.pp.neighbors}, using 9 nearest neighbors and the 50 first principal components,  and then calling \texttt{scanpy.tl.umap} with \texttt{min\_dist}-parameter to 0.4. We randomly removed two data points to make the total number of points divisible by four, and split the data set into four disjoint partitions, each of size 11200. MultistageOT was then applied to each partition independently.

The $p$-values in the differential expression analysis (\autoref{fig:outlier-results}k) were computed with the Mann-Whitney U-test function in Scipy (version 1.11.0).

\subsection{Implementation}
\label{sec:implementation}
MultistageOT solves a multistage optimal transport problem on cellular state space to minimize a total transition cost (see \hyperref[sec:supplementary_note]{Supplementary Note} for details). Transition costs between cell states were computed 
as the squared Euclidean distances between data points.  In applying MultistageOT to scRNA-seq data, we applied the preprocessing steps outlined in the previous section and then computed the squared Euclidean distances on the PCA transformed data, keeping the first 50 principal components. 

All results were computed using a MacBook Pro 14 with an M1 Pro chip and 32 GBs of memory. This allowed us to compute the MultistageOT solution on data sets of around $N \approx 11000$ cells using \hyperref[alg:baseblockascent]{Algorithm \ref{alg:baseblockascent}} (\hyperref[sec:supplementary_note]{Supplementary Note}). Hence, for larger the data sets from Dahlin et al. \cite{dahlin2018single} and Weinreb et al. \cite{weinreb2020lineage} we partitioned the data into uniformly subsampled disjoint subsets of roughly that size and applied the method independently on each subset.

\subsubsection{Number of transport stages}
For all scRNA-seq data sets, we solved the MultistageOT problem using $T-1 = 20$ intermediate transport steps.

\subsubsection{Picking initial and terminal states}
For the Paul et al. \cite{paul2015transcriptional} data, in the absence of prior knowledge about the number of cells represented in each mature blood cell lineage, we picked $n_0 = 3$ initial cell states and $n_F = 24$ terminal cell states (3 in each of the eight lineages) in regions of the UMAP expressing lineage associated markers (\autoref{fig:paul15-results}a). The initial states were chosen to lie in a region of the UMAP marked by high \textit{Mpl} expression as well as low \textit{Pf4} expression. 

For the Weinreb et al. \cite{weinreb2020lineage} lineage tracing data, we selected a total of $n_0 = 120$ initial states, and $n_F = 6972$ terminal states (\autoref{fig:weinreb-results}a). Since the data was partitioned into 12 subsets, solved independently, we selected exactly 10 initial states and 581 terminal states in each subset. The initial states comprise a subset of an HSC population found by only considering cells annotated by Weinreb et al. \cite{weinreb2020lineage} as ``Undifferentiated'', in a region of the SPRING-plot embedding also marked by high expression of \textit{Procr} and \textit{Cd34}. Terminal states were chosen to lie on the ``tips'' of the spring-graph embedding, representing each of the annotated mature blood cell lineage. We based the number of terminal states represented in each fate on the relative frequency of annotated cells in each fate, and then normalized this number of terminal states so that the smallest fate had one cell represented.

For the Dahlin et al. \cite{dahlin2018single} data, we picked a total of $n_F = 252$ terminal cell states, in regions of the UMAP expressing lineage associated markers (\autoref{fig:dahlin2018_annotations}) (40 in each lineage except for Eosinophils (12), Mast cell (20) and Baso (20)). We ensured that any terminal cell states in $\mathcal{X}_F$ had nearby intermediate cell states in $\mathcal{X}$ so as to not cause numerical instability. We picked $n_0 = 40$ initial states in a region of the UMAP marked by high expression of the gene \textit{Procr}.  

\subsubsection{Cost of transporting mass to auxiliary cell states}
When applying the model extension with auxiliary cell states on the scRNA-seq data in Dahlin et al. \cite{dahlin2018single}, we used a fixed cost $Q = 3.25$ of transporting mass to any of the auxiliary states. This cost was chosen as conservatively as possible so as to not needlessly affect the overall solution (see \autoref{fig:cost_distributions}c for a representative plot of the distributions of the costs on which $Q=3.25$ is marked). As such, it corresponds to the highest value that we found prevented the iterates from generating division-by-zero errors. 

To gauge the effects of the auxiliary cell state extension, we applied it to the synthetic data set and found it to produce robust results over a range of $Q$-values (\autoref{fig:extension_synthetic_data_outliers_different_q_values}). We also tested it on a partition of the lineage tracing data set \cite{weinreb2020lineage} with $Q = 2$ and found that too produced similar results (see \autoref{tab:q-values-weinreb}).

\subsubsection{Convergence}
We monitored the convergence of the Sinkhorn algorithm (\hyperref[sec:supplementary_note]{Supplementary Note}) and terminated the iterations when the maximum over all dual variable updates and absolute constraint deviations were less than the chosen tolerance level $\tau =10^{-4}$ for the data from Paul et al. \cite{paul2015transcriptional} and Dahlin et al. \cite{dahlin2018single} and $\tau =5\cdot 10^{-4}$ for the data from Weinreb et al. \cite{weinreb2020lineage} (tolerance level was increased slightly to lower compute time on this bigger data set).

\subsubsection{Regularization parameter}
\label{sec:methods_regularization_parameter}
MultistageOT solves an entropy-regularized multistage optimal transport problem. The level of regularization is set by the user, and is therefore a hyperparameter in MultistageOT, affecting the level of diffusion in the optimal transitions (\hyperref[sec:supplementary_note]{Supplementary Note}). It is important to note that the relative influence of the regularization parameter, denoted $\epsilon$, will in general be affected by the scale of the data. Therefore, in our implementations, following Schiebinger et al. \cite{schiebinger2019optimal}, we normalized the cost matrices by dividing the elements by the median of all costs. This made $\epsilon$ less dependent of the original scale of the input data and it was thus on a comparable scale between different data sets. See \autoref{tab:regularization-parameter-values} for values used in our implementations.

To achieve a regularization parameter in the desired range, we nested our algorithm in a proximal point scheme (see \autoref{sec:proximalpoint} in \hyperref[sec:supplementary_note]{Supplementary Note}). This required us to set a starting value for the regularization parameter. This starting value was chosen as low as possible to limit the number of outer proximal point iterations needed to reach a sufficiently low effective regularization parameter. Note that, even when normalizing the cost matrices, these starting values differ slightly depending on specific data distribution within each data set  (e.g., due to the presence of outliers that correspond to many large distances to other cells). For example, the median transition cost in the synthetic data set with outliers is $\sim$1.1 times larger than the median in the original data without outliers. To ensure that the optimal transitions between the non-outlier cells would be comparable in the two settings assessed in \autoref{fig:extension_synthetic_data}e, the regularization parameter was scaled accordingly. We refer to \autoref{tab:regularization-parameter-values} for a summary of the initial and final regularization parameter values for each data set.

\subsection{Downstream analyses}
\label{sec:downstream}
Let $\mathcal{D}=\{x^{(i)}\,|\, i=1,\ldots, N\}$, represent a snapshot of gene expression states. Specifically, $x^{(i)}\in \mR^m$ represents the vector of $m$ gene expression levels of cell $i$. We represent the sets of initial, intermediate and terminal states respectively with the ordered sets   $\mathcal{X}_0=(x^{(i)}\,|\, i=1,\ldots, n_0)$, $\mathcal{X}=(x^{(i)}\,|\, i=n_0+1,\ldots, N-n_F) 
$ and $\mathcal{X}_F=(x^{(i)}\,|\, i=N-n_F+1,\ldots, N)$,  

In MultistageOT, cell differentiation is modeled as a multistage optimal mass transport problem, whereby initial cells in $\mathcal{X}_0$ are assumed to transition into terminal cells in $\mathcal{X}_F$ over $T$ stages of differentiation in a way that minimizes a total transition cost. Due to the squared Euclidean distance cost, transitions corresponding to large jumps between cell states are associated with high cost. Thus, to achieve a minimum transition cost, intermediate cells in $\mathcal{X}$ must be recruited as feasible intermediate differentiation states (see \hyperref[sec:supplementary_note]{Supplementary Note} for details). We here show how this induces a pseudotemporal ordering as well as cell fate probabilities of all intermediate cell states. 

\subsubsection{Pseudotime}
Denote by $\mu_t^{(i)}$ the total mass transported through each intermediate state,  $x^{(n_0 + i)} \in \mathcal{X}$, in stage $t$ for $i = 1,\dots,N - n_0 - n_F$ and $t = 1,\dots, T-1$.
By normalizing by the total transported mass (over all stages), we obtain, for each intermediate state, the corresponding probabilities
\begin{align}
     \mathbb{P}\left (x^{(n_0 + i)} \in \mathcal{X}  \text{ transports mass in stage } t \right) := \frac{ \mu_t^{(i)} }{\sum_{k=1}^{T-1}  \mu_k^{(i)} }
\end{align}
(illustrated in \autoref{fig:synthetic-data-results}c). Similarly, let $\nu_t^{(i)}$ be the total mass received by each terminal cell state,  $x^{(N-n_F + i)} \in \mathcal{X}_T$, in stage $t$ for $i = 1,\dots,n_F$ and $t = 1,\dots, T-1$. We obtain the probabilities
\begin{align}
     \mathbb{P}\left(x^{(N-n_F + i)}  \in \mathcal{X}_F  \text{ receives mass in stage } t\right) := \frac{ \hat{\nu}_t^{(i)} }{\sum_{k=1}^{T-1}  \hat{\nu}_k^{(i)} }.
\end{align}
Based on these probabilities, we compute, for each intermediate state and terminal state, an expected value over the different transport stages. A pseudotemporal ordering of the cells is then obtained by ranking the cells based on this expected transport stage, such that the cell with the earliest expected transport stage is assigned a pseudotime of 0, and the cell with the latest expected transport stage is assigned a pseudotime of 1. Initial states in $\mathcal{X}_0$ by definition have an expected transport stage of 0 and are thus all given a pseudotime of 0.

\subsubsection{Cell fate probabilities}
To quantify the likelihood of a cell ending up in a particular fate, we use results from Markov chain theory. In particular, we estimate cell fate probabilities by computing absorption probabilities in an absorbing Markov chain as follows, similar to Weinreb et al. \cite{weinreb2018fundamental} and
 Lange et al. \cite{lange2022cellrank}.

Let $m_{ij}^{(t)}$ denote the amount of mass sent between cell $i$ and any other cell $j$ in the $t$:th stage, obtained by solving \hyperref[alg:baseblockascent]{Algorithm \ref{alg:baseblockascent}} (see \hyperref[sec:supplementary_note]{Supplementary Note}). To quantify the affinity between any two cell states, we aggregate the transport plans by summing the total mass sent between any two cells over all stages, i.e., we compute
\begin{equation}
    m_{ij} = \sum_{t=1}^{T-1}m_{ij}^{(t)}, \quad \text{for }i = 1,\dots, N, \mbox{ and }  j = 1,\dots, N. 
\end{equation}
For each initial and intermediate cell states (i.e., $i=1, \dots, N-n_F$) the aggregated transport coupling is normalized to represent a probability of transition between cell $i$ and cell $j = 1,\dots, N$:
\begin{equation}
    a_{ij} := \frac{m_{ij}}{\sum_{k=1}^N m_{ik}} \quad \text{(cell $i$ to cell $j$  transition probability model)}.
\end{equation}
For terminal cell states, $i=N-n_F, \dots, N$, we have that $m_{ij}=0$ for all $j = 1,\dots,N$, so we define
\begin{equation}
    a_{ij} = \begin{cases}
        1 \enspace &\text{if $i = j$}\\
        0, &\text{otherwise}.
    \end{cases}
\end{equation}
We let the matrix $A = [a_{ij}]_{i,j=1}^N$ be a transition probability matrix in a stationary absorbing Markov chain, in which initial and intermediate cell states in $\mathcal{X}_0$ and $\mathcal{X}$ define the $n$ transient states and the terminal cell states in $\mathcal{X}_F$ define the $n_F$ absorbing states. We partition $A$ according to
\begin{align}
    A = \begin{bmatrix} Q & R \\ \boldsymbol{0} & I \end{bmatrix},
\end{align}
so that $Q = [q_{ij}]_{i,j=1}^{n}$ encodes transitions between transient states, $R$ encodes transitions between transient and absorbing states, $\boldsymbol{0}$ is the $n_F \times n$ matrix of zeros, and $I$ is the $n_F\times n_F$ identity matrix. It is a well-known result in Markov chain theory that the matrix
\begin{align}
    Z = (I -Q)^{-1}R,
\end{align}
with elements $z_{ij}$, $i = 1,\dots,n$, $j= 1,\dots, n_F$, encodes the probability of transient state $i$ being absorbed in absorbing state $j$. 

Assume $K$ different classes of absorbing states (e.g., representing different lineages). The probability of transient state $i$ ending up in a class $k$ of absorbing states is obtained by summing the absorption probabilities over all states in class $k$. Hence, we compute the cell fate probability of cell $i$ eventually being absorbed in fate $k$ through

\begin{align}
    p_k^{(i)} :=   \sum_{ j \in \mathcal{I}_k} z_{ij},
\end{align}
 where $\mathcal{I}_k = \{ j \in \{1,\dots, n_F \} | \text{cell state } x^{(N-n_F + j)} \text{ belongs to fate } k\}$ denotes the index set corresponding to terminal states in fate class $k$. As a measure of potency of cell $i$, we computed the Shannon entropy according to:
\begin{equation}
    s^{(i)} = -\sum_{k=1}^{K} p_k^{(i)}\log { p_k^{(i)}}, \quad i = 1,\dots, N.
\end{equation}
Higher entropy $s^{(i)}$ reflects a lower degree of commitment, whereas lower values  reflects a higher degree of commitment to a particular fate (the lowest value  corresponds cell fate probabilities $p_k^{(i)}$ which is a singular distribution, supported only on a single fate, whereas the highest value corresponds to a uniform distribution over all fates).

\subsubsection{Quantifying the model's predictive performance}
Lineage tracing data from Weinreb et al. \cite{weinreb2020lineage} was used to benchmark model predictions.

\textit{Predicting annotated mature cells}: When comparing fate prediction with author annotations in \autoref{fig:weinreb-results}f, we excluded cells from the terminal cell subset, $\mathcal{X}_F$, since these were used in fitting the models to the data. The dominant predicted fate probability (the fate predicted as most probable) in each cell was then compared to its annotation. The accuracy was taken as the ratio of correct matches to the total number of considered cells (a total of 51720 annotated cells were considered).

\textit{Predicting day 2 cells}:  When benchmarking the performance of predicting fates of day 2 cells, we chose day 2 cells satisfying each of the following criteria. In addition to being a cell from day 2, it needed to:
\begin{enumerate}
    \item Belong to a clone with a non-empty set of clonal sisters in days 4 or 6.
    \item Not be part of the set of predefined terminal states used in optimizing our model.
    \item Have at least 7 clonal sisters in days 4 or 6.
    \item Have at least 1 clonal sister in days 4 or 6 annotated as committed to any of the 10 identified cell lineages.
    \item Not belong to a clone represented in more than one library type (which would indicate a spurious clonal relationship based on a repeated lineage barcode).
\end{enumerate} 
This left us with $N_{d_2} = 795$ day 2 cells. For each such cell, we computed ``ground truth'' cell fate probabilities by counting the number of its clonal sisters in days 4 or 6 that were annotated as a particular mature cell type, ignoring the counts of clonal sisters named ``Undifferentiated''. The counts were then normalized to sum to one: inducing, for each day 2 cell, a probability distribution over the different mature cell types. More specifically, by letting $n_k^{(i)}$ denote the number of clonal sisters found in fate $k$ at days 4 or 6 within the clone of a given day 2 cell $i$ ($i=1,\dots, N_{d_2}$), the empirical cell fate probability, $q_k^{(i)}$, of cell $i$ ending up in fate $k$ was computed according to 
\begin{align}
    q_k^{(i)} = \frac{n_k^{(i)}}{\sum_{k=1}^{K} n_k^{(i)}}.
\end{align}
Based on using these empirical cell fate probabilities as a ground truth, we could compute  different performance metrics (see table in \autoref{fig:weinreb-results}g) including total variation distance (TV), defined via
\begin{align}
    \text{TV}(p,q) := \frac{1}{2}\sum_{k=1}^{K} | p_k - q_k |,
\end{align}
for two probably vectors, $p := (p_1,\dots, p_K)$, $q := (q_1,\dots, q_K)$, as well as dominant fate accuracy when predicting the most likely cell fate:

\begin{equation}
\label{eq:accuracy}
  \text{dominant fate accuracy} := \frac{\text{correct classifications of most likely fate}}{N_{d_2}}.
\end{equation}

We also developed two novel metrics to gauge cell fate predictive performance. The first, which we call total accuracy, is defined as the fraction of cells for which the $m$ largest predicted fate probabilities matches the $m$ observed clonal sister fates. More specifically, let $\mathcal{F}^{(i)}$ represent the set of observed clonal sister fates for cell $i$, i.e., $\mathcal{F}^{(i)} =  \{ k | q_k^{(i)} > 0\}$, and again let $p_k^{(i)}$ denote the predicted probability of fate $k$ in cell $i$. Total accuracy is defined as the fraction of day 2 cells for which it holds that
\begin{equation}
    \min\left \{p_k^{(i)}\right\}_{k\in \mathcal{F}^{(i)}} > \max\left\{p_k^{(i)}\right\}_{k\notin \mathcal{F}^{(i)}}.
\end{equation} 
In the second measure, presented in \autoref{fig:weinreb-results}l, we turned the predicted (continuous) fate probabilities into categorical fate assignments by assigning to each cell any fate with a predicted probability above a threshold $\delta$  (\autoref{fig:weinreb-results}k). This can be seen as a generalization of the dominant fate accuracy which corresponds to a single assignment (the dominant fate probability). We achieved a new type of performance metric by calculating the fraction of all cells with perfect matches after performing fate assignments at a particular threshold. Note that this performance metric depends on the threshold $\delta$. To get a robust estimate for this generalized accuracy, we computed this value for all possible thresholds $\delta \in [0,1]$ and computed an area under the curve (\autoref{fig:weinreb-results}l).
 
 \subsubsection{StationaryOT} 
The StationaryOT \cite{zhang2021optimal} method is also based on optimal transport. It requires the user to specify sets of initial and terminal cell states (referred to as source and sink nodes respectively in Zhang et al. \cite{zhang2021optimal}). It represents a fundamentally different approach from MultistageOT in that StationaryOT solves a single-step transport problem to obtain couplings between cells, whereas MultistageOT solves a global transport problem over multiple transport steps (\hyperref[sec:supplementary_note]{Supplementary Note}). 

Following \cite{zhang2021optimal}, we used the quadratic cost function $C(x,y) = \frac{1}{2}\lVert x - y \rVert^2_2$ between cell states $x$ and $y$. In benchmarking predictive performance in \autoref{fig:weinreb-results}, we applied StationaryOT for computing cell fates within each of the 12 partitions of the data set, using the same set of initial and terminal states as in MultistageOT, denoted $\mathcal{X}_0$ and $\mathcal{X}_F$ respectively. In addition to the entropy-regularization parameter, $\epsilon$, StationaryOT requires the user to specify relative growth rates for the cells, $g_i = \exp ( R_i) $, $i=1,\dots,N$ ($N$ being the number of cells and $R_i$ being a ``flux-rate'' \cite{zhang2021optimal}) as well as a time step parameter $\Delta t$. In the absence of growth rate estimates in the data set, we used the following heuristic: we specified $g_i \equiv 0$ for any sink node $ x^{(i)} \in \mathcal{X}_F$ and 
\begin{align}
  R_i \equiv  \frac{1}{\Delta t}\ln \left(\frac{n_F}{n_0} + 1 \right)  \iff g_i \equiv \exp \left\{ \frac{1}{\Delta t}\ln \left(\frac{n_F}{n_0} + 1 \right)\right\} = \left(\frac{n_F}{n_0} + 1 \right)^{\frac{1}{\Delta t}},
\end{align}
for any source node $ x^{(i)} \in \mathcal{X}_0$, and $g_i = 1$ for any intermediate cell $x^{(i)} \in \mathcal{X}$. Note that this implies that
\begin{align}
    \sum_{i=1}^{N} g_i^{\Delta t} = n_F + n_0 
 + n = N,
\end{align}
so that, if we take the first marginal at $t = 0$ to be $\mu = g^{\Delta t}$ and we define the second marginal at $t = \Delta t$ as $\nu = \left(\frac{\sum_{i=1}^{N} \mu_i}{N}\right) \ones{N}$, we obtain $\nu = \ones{N}$. 

The time step parameter was taken as $\Delta t = 0.25$, which was used by Zhang et al. \cite{zhang2021optimal} on the \textit{Arabidopsis thaliana} root tip data set. The regularization parameter $\epsilon = 0.12$ was chosen small enough such that we did not experience underflow in the computation for the optimal couplings, and thus allowed us to obtain numerically stable results for the cell fate probabilities (both of these numerical issues are discussed in Zhang et al. \cite{zhang2021optimal}).

\subsubsection{Inverse Distance Weighed (IDW) model}
To gauge the performance of our modeling framework in relation to other approaches, we established a naive cell fate model, based only on the relative proximity to the different terminal fates and a prior fate bias. Let $\mathcal{X}_k \subset \mathcal{X}_F$ be the subset terminal states in $\mathcal{X}_F$ belonging to cell fate $k$ and let $x^{(i)}$ denote the cell state of cell $i$ in day 2. Specifically, the IDW model estimates the probability for a cell $i$ to end up in fate $k$  according to:

\begin{align}
 r_k^{(i)} = \frac{b_k\left(\frac{\lVert x^{(i)} -  \bar{x}_k \rVert} { \lVert \bar{x}_{0} -  \bar{x}_k \rVert} \right)^{-\gamma}}{\sum_k b_k\left(\frac{\lVert x^{(i)} -   \bar{x}_k \rVert} { \lVert \bar{x}_{0} -  \bar{x}_k \rVert} \right)^{-\gamma}},
\end{align}
where $b_k$ is a weight representing a prior fate bias for fate $k$ such that $\sum_k b_k = 1$, 
\begin{align}
    \bar{x}_0 = \frac{1}{|\mathcal{X}_0 |} \sum_{ x^{(i)} \in \mathcal{X}_0 }x^{(i)}
\end{align}
is the sample mean vector of all initial states, and 
\begin{align}
    \bar{x}_k = \frac{1}{ |\mathcal{X}_k |} \sum_{ x^{(i)} \in \mathcal{X}_k }x^{(i)}
\end{align} 
is the sample mean of the terminal states in fate $k$. The parameter $\gamma \in [0, \infty)$ controls how much the distances to the fates should be weighted. As $\gamma \rightarrow \infty$, the cell fate probabilities tend to the deterministic distributions corresponding to a single dominant fate, whereas $\gamma = 0$ corresponds to $r_k^{(i)} \equiv b_k$, for $i=1,\dots,N_{d_2}$. We base the prior weights, $b_k$, for each fate $k$, on the number of mature cell type annotations found in the corresponding fate (as annotated by Weinreb et al. \cite{weinreb2020lineage}). We scaled these to sum to unity and obtained the weights shown in \autoref{tab:prior_weights_for_idw_model}. All distances were computed on the space spanned by the 50 first PCA components.

The IDW model depends on $\gamma$. We chose $\gamma \in [1,50]$ so as to maximize the accuracy \eqref{eq:accuracy}. This yielded the optimal $\gamma^* = 20$ (see \autoref{fig:stationaryot_and_idw_optimization}b). 

\subsubsection{Inference of bipotent basophil and mast cell progenitors}

For the analysis of cells possessing both basophil and mast cell lineage-forming capacity (shown in  \autoref{fig:outlier-results}j-k), we selected cells with MultistageOT-inferred fate probabilities above 10\% for both the mast cell and basophil fates, and less than 1\% for all other fates.

\end{bibunit}

\clearpage

\begin{bibunit}
\begin{appendix}

\renewcommand{\thesection}{} 
\counterwithin{table}{section}
\renewcommand{\thetable}{S.\arabic{table}}

\section{Supplementary tables}

\begin{table}[ht]
\centering
\caption{Regularization parameter values. Note that $\epsilon_0$ is the initial value in the proximal point scheme, and $\epsilon'$ is the effective regularization parameter value, corresponding to having obtained the optimal transport plans with an uninformed prior matrix of all-ones in the entropy-regularization (see section \ref{sec:proximalpoint} in \hyperref[sec:supplementary_note]{Supplementary Note}). The parameter values for the synthetic data with and without outliers correspond to the results in \autoref{fig:synthetic-data-results}i and \autoref{fig:outlier-results}b, respectively.}
\label{tab:regularization-parameter-values}
\begin{tabular}{@{\hskip 0.5in}llll@{\hskip 0.5in}}
\toprule  
 \textbf{Data set} & \textbf{Initial} $\epsilon_0$ & \textbf{Effective}, $\epsilon'$\\ 
\midrule
Synthetic data: & &  \\ 
Without outliers & 0.015 & 0.0042 \\
With outliers & 0.038 & 0.0038 \\
\midrule
  Paul et al. (2015):  &  0.055& 0.0055 \\ 
\midrule
  Dahlin et al. (2018):  &  &  \\ 
  \hline
    Partition 1  &  0.013 &  0.0045 \\ 
  \hline
      Partition 2  & 0.013 & 0.0043 \\ 
  \hline
      Partition 3  & 0.013 &  0.0043\\ 
  \hline
      Partition 4  & 0.013 & 0.0045 \\ 
\midrule
  Weinreb et al. (2020):  &  &  \\ 
  \hline
      Partition 1  & 0.03 & 0.0057 \\ 
  \hline
      Partition 2  &  0.03& 0.0060 \\ 
  \hline
      Partition 3  & 0.03 &  0.0055\\ 
  \hline
      Partition 4  & 0.03 &  0.0060\\ 
  \hline
      Partition 5  & 0.03 &  0.0055\\ 
  \hline
      Partition 6  &0.03  &  0.0059\\ 
  \hline
      Partition 7  &0.03  & 0.0060 \\ 
  \hline
      Partition 8  &  0.03&  0.0057\\ 
  \hline
      Partition 9  &  0.03&  0.0060 \\ 
  \hline
      Partition 10  & 0.03 &  0.0060 \\ 
  \hline
      Partition 11  & 0.03 & 0.0066 \\ 
  \hline
      Partition 12  & 0.03 & 0.0058 \\ 
\bottomrule
\end{tabular}
\end{table}

\begin{table}[ht]
\centering
\caption{Weights, $b_k$ used in the Inverse Distance Weighed (IDW) model for estimating cell fate probabilities}.
\label{tab:prior_weights_for_idw_model}
\begin{tabular}{@{\hskip 0.5in}ll@{\hskip 0.5in}}
\toprule  
 Fate, $k$: & Weight, $b_k$: \\ 
\midrule
Neutrophil  \quad &  \quad0.3789 \\
Monocyte   \quad  &  \quad0.3268 \\
Baso       \quad  & \quad 0.1730 \\
Mast       \quad  & \quad 0.0440 \\
Meg        \quad  & \quad 0.0357 \\
Lymphoid   \quad  & \quad 0.0146 \\
Erythroid  \quad  & \quad 0.0117 \\
Eos        \quad  & \quad 0.0095 \\
Ccr7\_DC   \quad   &\quad  0.0041 \\
pDC        \quad  & \quad 0.0017
\end{tabular}
\end{table}

\begin{table}[ht]
\centering
\caption{Change (measured by total variation distance (TV)) when recomputing the cell fate probabilities with the extended MultistageOT model with auxiliary cell states on day 2 cells in a partition of the Weinreb et al. \cite{weinreb2020lineage} data set , for three different values for the fixed transport cost, $Q$.}
\label{tab:q-values-weinreb}
\begin{tabular}{@{\hskip 0.5in}cc@{\hskip 0.5in}}
Cost, $Q$ & Mean(max) TV \\ 
\bottomrule
1 & 0.047(0.25)  \\ 
2 & 0.0092(0.054)  \\ 
3.25 & 0.0091(0.054)  \\
\bottomrule
\end{tabular}
\end{table}

\end{appendix}

\clearpage

\begin{appendix}

\renewcommand{\thesection}{} 
\counterwithin{figure}{section}
\renewcommand{\thefigure}{S.\arabic{figure}}

\section{Supplementary figures}

\begin{figure}[H]
    \centering
    \includegraphics[width=\linewidth]{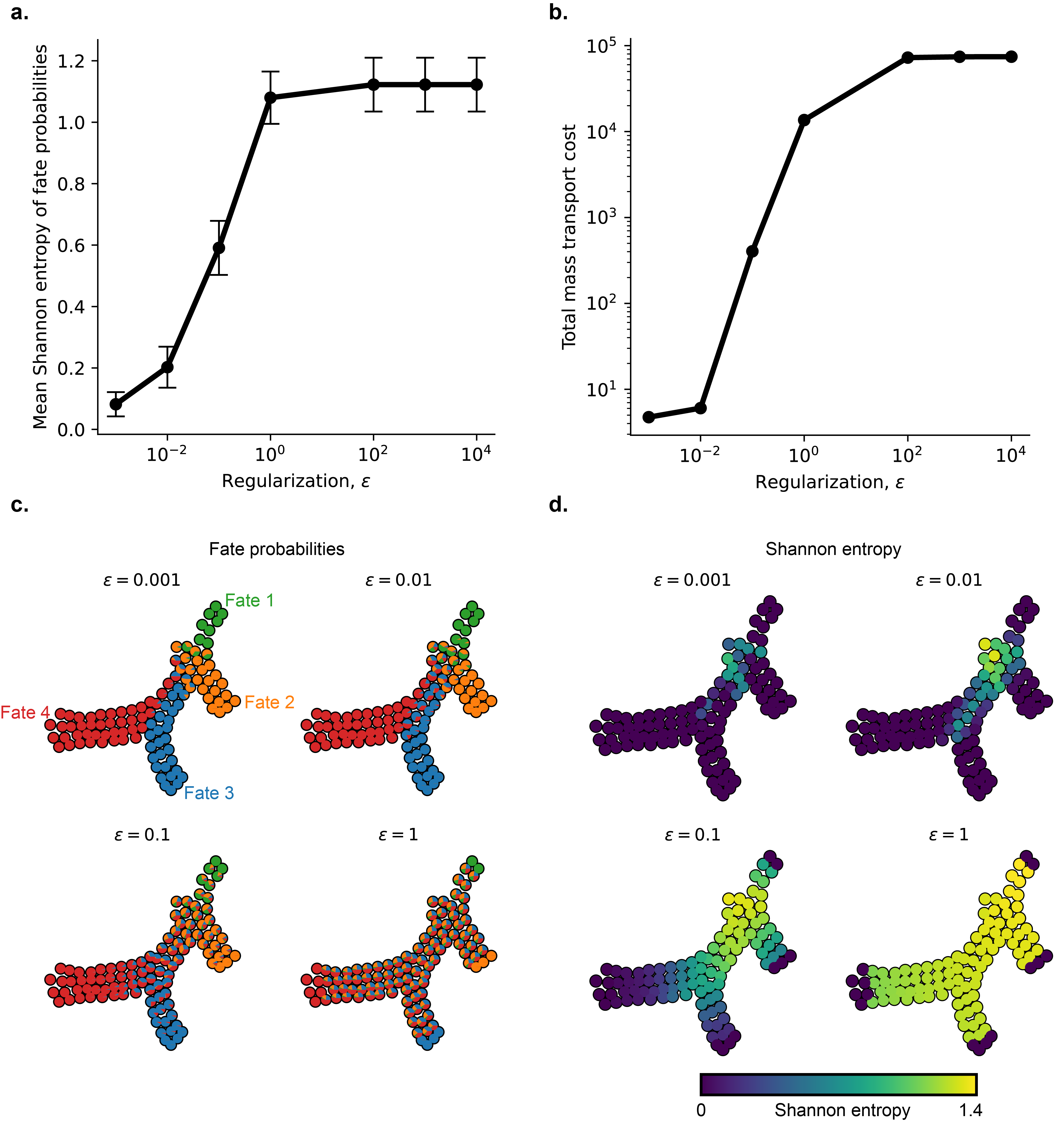}
    \caption{Overall Shannon entropy increases monotonically with increasing regularization parameter values. (a) Mean Shannon entropy of fate probabilities, averaged over all cells in the MultistageOT solution corresponding to the data set used in \autoref{fig:synthetic-data-results}, plotted against increasing values of the regularization parameter, $\epsilon$. Errorbars correspond to approximate 95\% confidence intervals. (b) Total cost of transports in the MultistageOT solution plotted against increasing values of the regularization parameter, $\epsilon$. (c) Cell fate probabilities, represented by pie charts, for different values of $\epsilon$. (d) Shannon entropy of the cell fate probabilities for each cell shown in (c), for different values of $\epsilon$.}
    \label{fig:entropy_regularization_effects}
\end{figure}

\begin{figure}[H]
\centering\includegraphics[width=\linewidth]{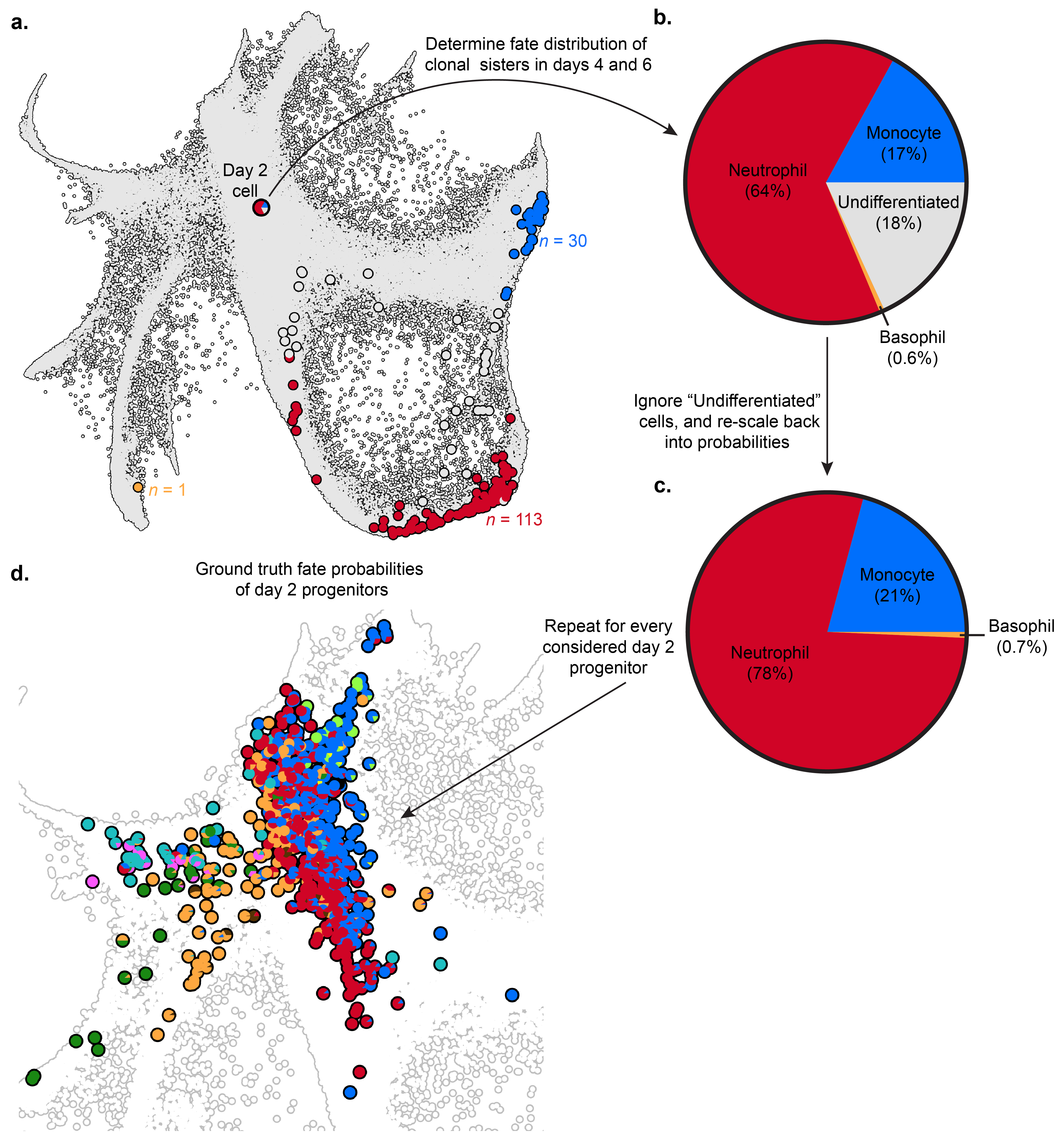}
   \caption{Establishing ``ground truth'' fate probabilities of day 2 progenitor cells using fate distribution of their clonal sisters in later time points (days 4 and 6). (a) Example of a day 2 progenitor cell and its clonal sisters in days 4 and 6. (b) The distribution of clonal sisters (basophil: 1, neutrophil: 113, monocyte: 30, undifferentiated: 32) can be normalized to induce a probability distribution over the different fates. (c) We ignore the ``Undifferentiated'' cells, and then rescale the relative frequencies of clonal fates so that they sum to unity. (d) We repeat this for every selected day 2 progenitor that matched our inclusion criteria (see \hyperref[sec:methods]{Methods}). Performance metrics were then obtained by comparing our model's estimated cell fate probabilities predictions to these empirical fate probabilities.}
\label{fig:quantifying_ground_truth}
\end{figure}

\begin{figure}[H]
    \centering
    \includegraphics[width=\linewidth]{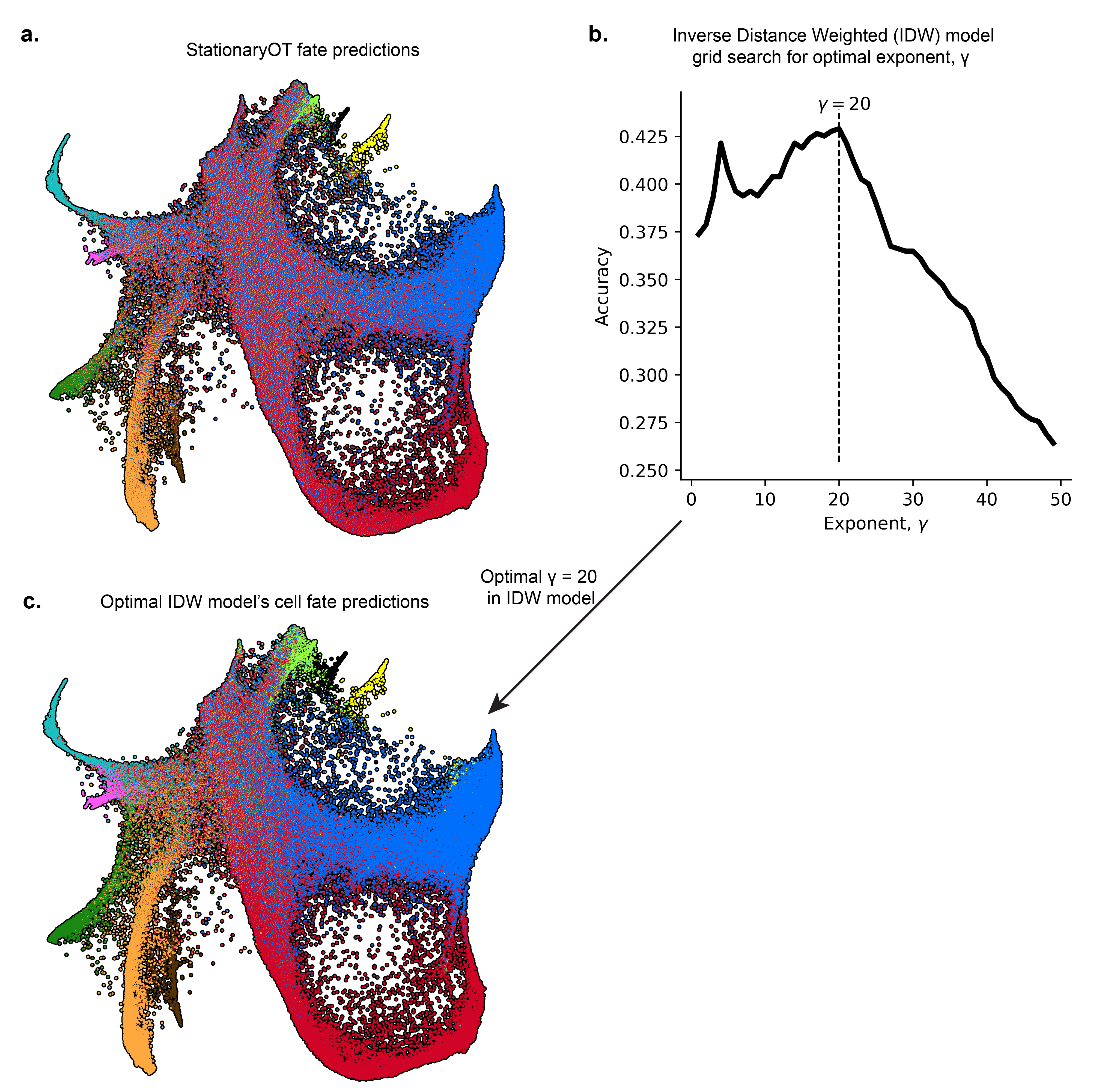}
   \caption{(a) StationaryOT's predictions of cell fate probabilities for 130884 cells in the Weinreb et al. \cite{weinreb2020lineage} data set. (b) The inverse distance weighted (IDW) model (see \hyperref[sec:methods]{Methods}). We used the ground truth fates to optimize the parameter with respect to ``dominant fate accuracy''. (b) The cell fate probabilities corresponding to IDW's optimal $\gamma$ parameter. (c) IDW model's predictions of cell fate probabilities for 130884 cells in the Weinreb et al. data set. }
\label{fig:stationaryot_and_idw_optimization}
\end{figure}

\begin{figure}[H]
    \centering\includegraphics[width=1\linewidth]{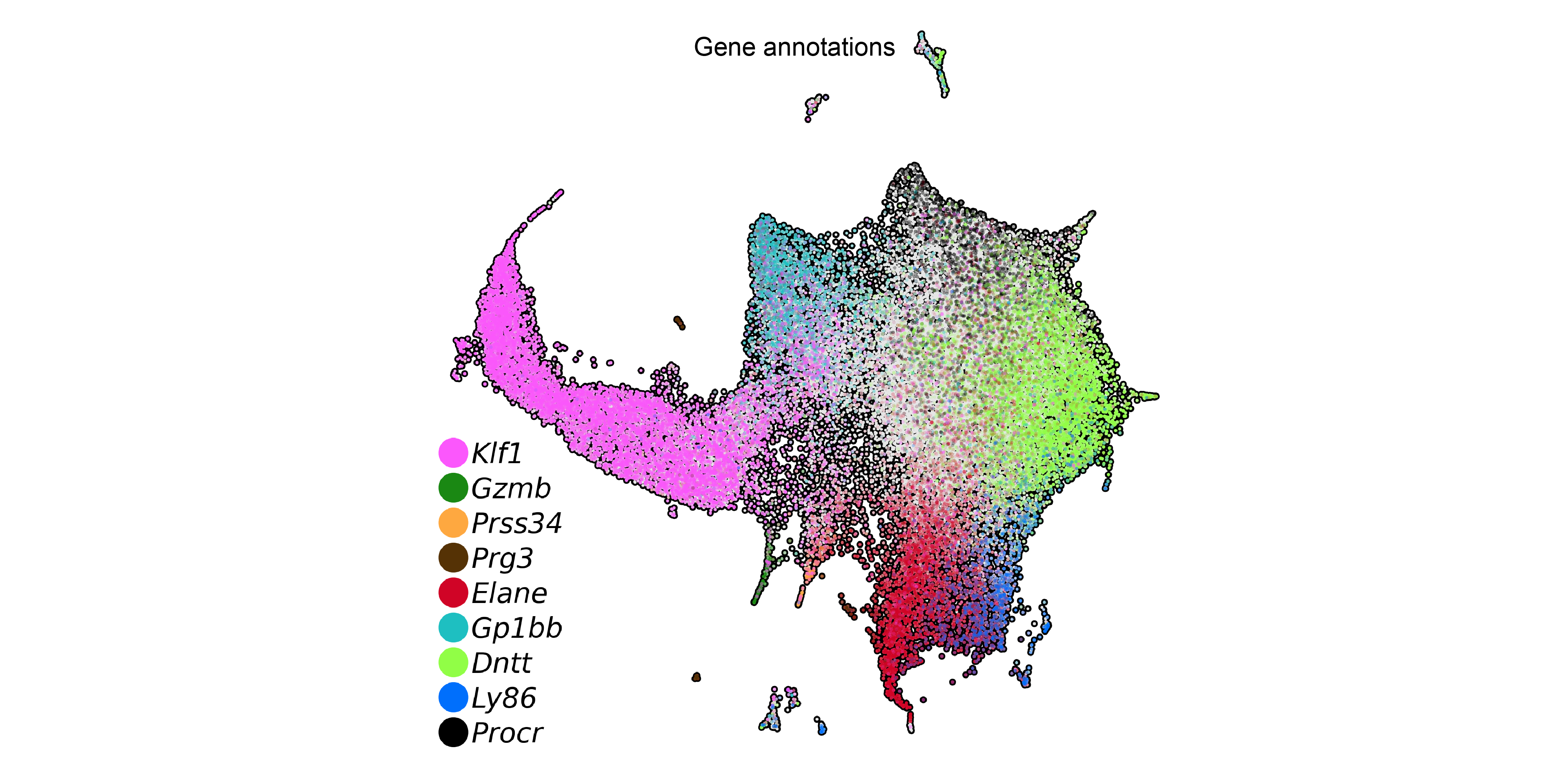} 
   \caption{Layered gene expression (log($x$+1)-transformed) plot showing the orientation of different cell lineages and hematopoietic stem cells: erythroid (\textit{Klf1}), mast cell (\textit{Gzmb}), basophil (\textit{Prss34}), eosinophil (\textit{Prg3}), neutrophil (\textit{Elane}), megakaryocyte (\textit{Gp1bb}), lymphoid (\textit{Dntt}), monocyte (\textit{Ly86}) and hematopoietic stem cells (\textit{Procr}). Gene markers from Dahlin et al. \cite{dahlin2018single}. } 
\label{fig:dahlin2018_annotations}
\end{figure}

\begin{figure}[H]
    \centering
    \includegraphics[width=\linewidth]{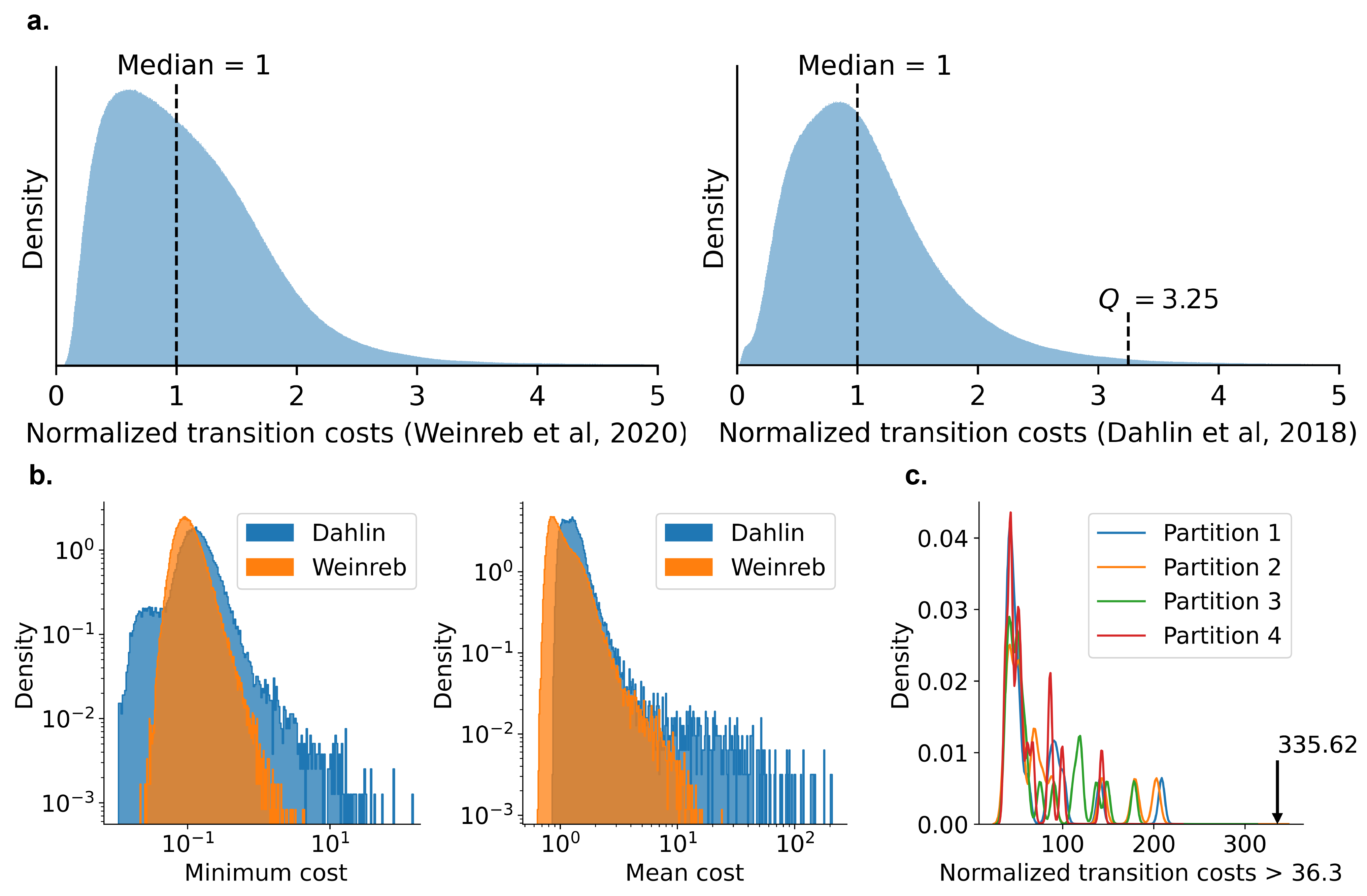}
    \caption{(a) Representative distributions of transition costs (squared Euclidean distance) in the range [0,5] in Weinreb et al. \cite{weinreb2020lineage}  (left) versus Dahlin et al. \cite{dahlin2018single} (right). The cost $Q = 3.25$, of transporting unit mass to the auxiliary states in the extended MultistageOT model (see \hyperref[sec:methods]{Methods}) is marked in the right plot. (b) Minimum (left) and mean (right) transition costs in the Dahlin et al. versus Weinreb et al. data sets. Note the log-scales. (c) Density plots of transition costs for each partition of the Dahlin et al. data (\hyperref[sec:methods]{Methods}), considering only large costs above the maximum cost in the Weinreb et al. data set, which was 36.3. For reference, the arrow marks the highest transition cost in the Dahlin et al. data set. } 
    \label{fig:cost_distributions}
\end{figure}

\begin{figure}[H]
    \centering
    \includegraphics[width=\linewidth]{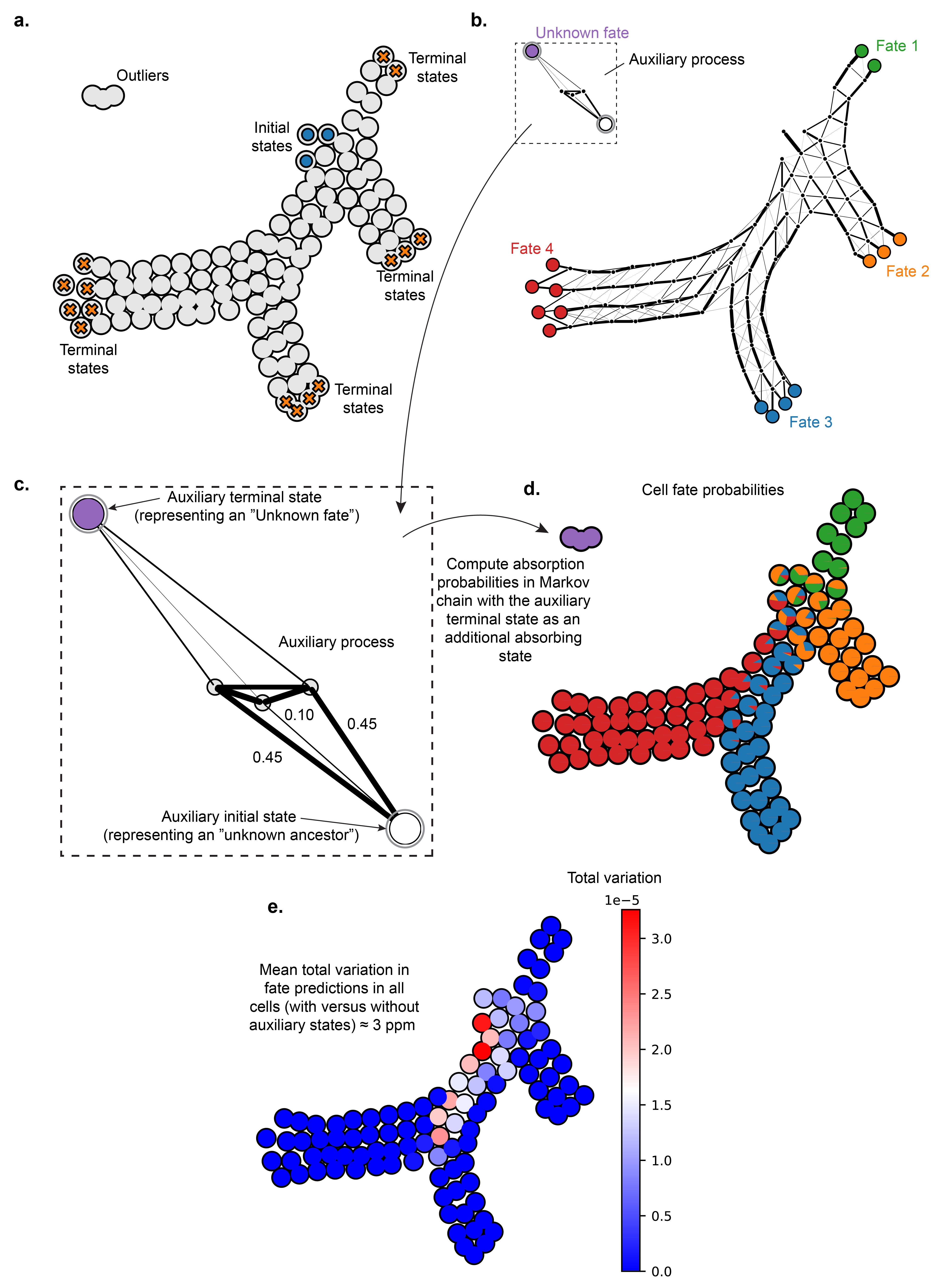}
    \caption{Extended MultistageOT model with auxiliary states applied to the synthetic data set in \autoref{fig:synthetic-data-results}, \makebox[\linewidth][s]{with the addition of outliers. (a) A cluster of three outlier cells was added in the top left corner of} } 
    \label{fig:extension_synthetic_data}
\end{figure}
\begin{figure}
    \ContinuedFloat 
    \caption*{the data  in \autoref{fig:synthetic-data-results}. (b) We solved the MultistageOT problem with the added auxiliary cell states. The plot visualizes the Markov chain transition probabilities between cells, based on the optimal transport  solution (\hyperref[sec:methods]{Methods}). (c) In this scenario, the auxiliary states help model the development of the outlier cells as separate process: No connections were formed between the outliers and the cells in the main process. The auxiliary intermediate cell state was not utilized by any cell in the optimal transport solution and was therefore omitted in the visualization. Note: The auxiliary terminal and initial state (purple and white respectively) are here only included to visualize the transition probabilities; in general they do not possess geometric coordinates.   (d) Including the auxiliary terminal state as an additional absorbing state in the Markov chain allows estimating the probability of each cell ending up in an ``unknown fate'', represented by the auxiliary terminal cell state. The probabilities are represented by pie charts, and purple wedges corresponds to this ``Unknown fate''. (e) The total variation distance for each non-outlier cell, quantifying the difference between the cell fate predictions shown in \autoref{fig:synthetic-data-results} (using the original MultistageOT formulation) and the cell fate predictions made with the extended MultistageOT model shown in subpanel (d).} 
\end{figure}

\begin{figure}[H]
    \centering
    \includegraphics[width=\linewidth]{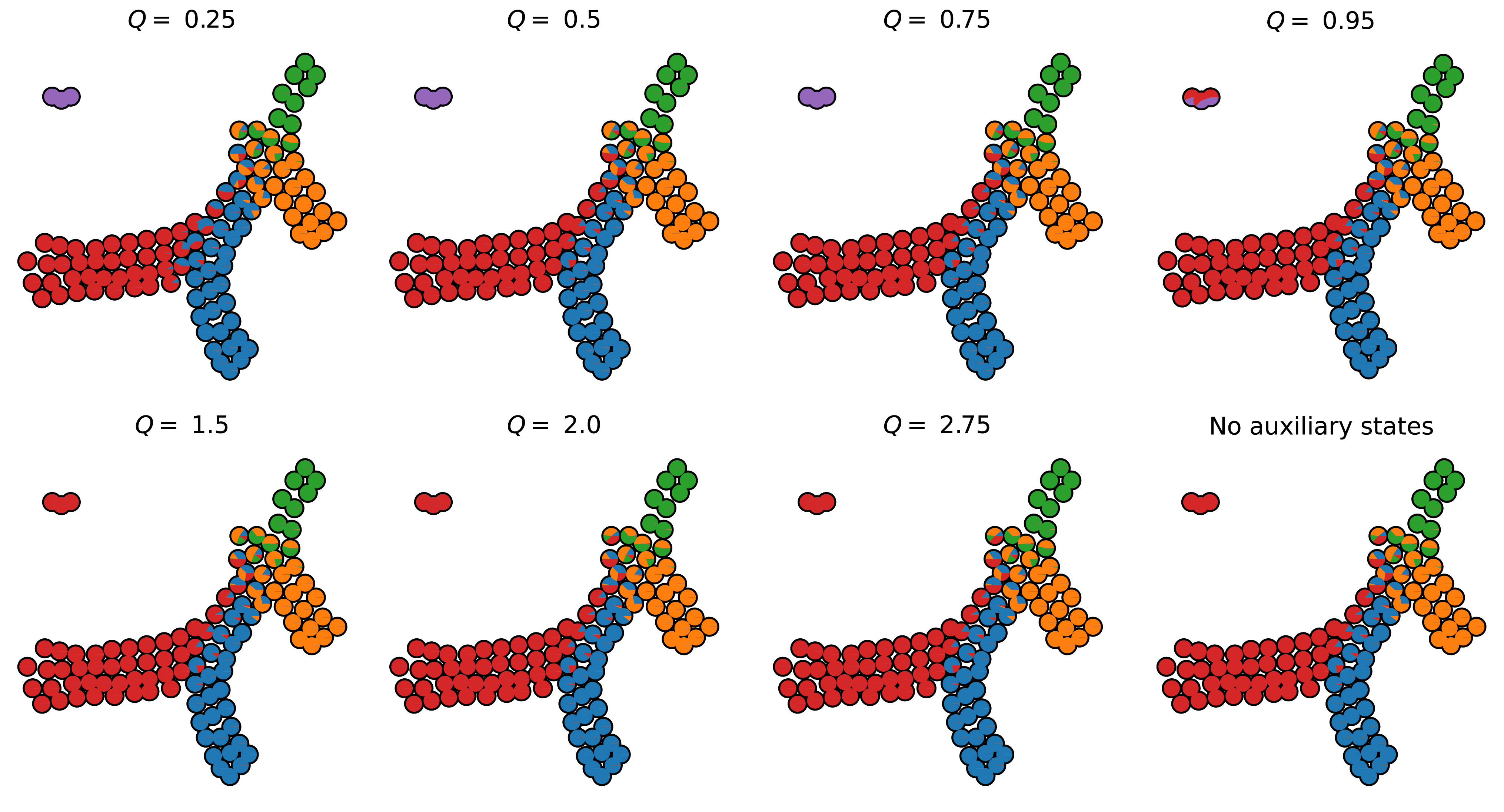} 
    \caption{Effects of varying the fixed cost $Q$ of transitioning to the auxiliary cell states on the estimated cell fate probabilities. For reference, we include the cell fate probabilities when no auxiliary states are used (bottom rightmost plot). Low $Q$-values favors more utilization of the auxiliary states, reflected in the strong commitment of the outliers to terminate in the auxiliary terminal state (purple color) for $Q \in \{0.25,0.5,0.75\}$. Conversely, higher $Q$-values means less involvement of the auxiliary states. For the results in \autoref{fig:outlier-results}b we used $Q = 0.9$.} 
\label{fig:extension_synthetic_data_outliers_different_q_values}
\end{figure}

\begin{figure}[H]
    \centering
    \includegraphics[width=\linewidth]{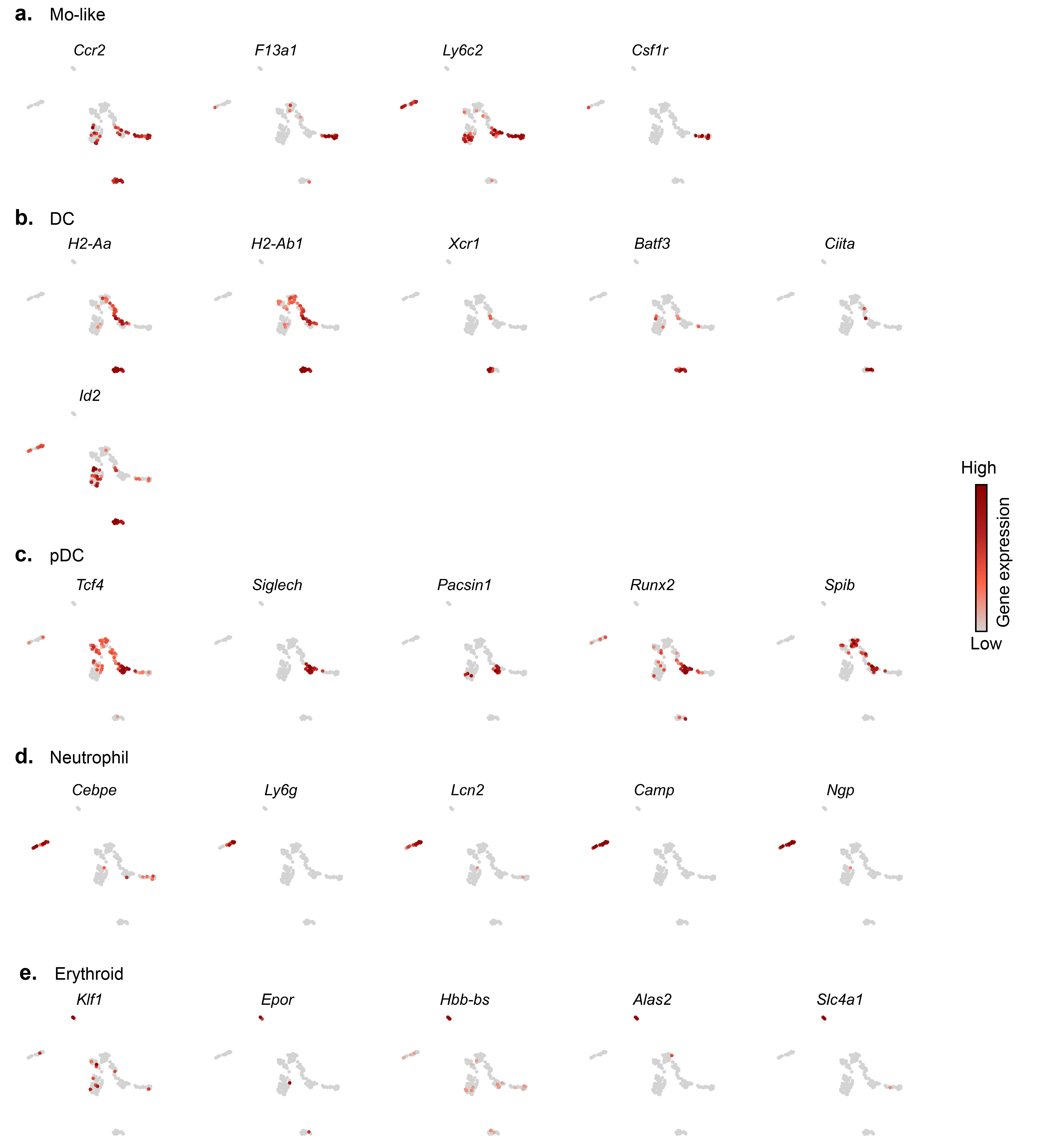} 
   \caption{The expression levels (log($x$+1)-transformed) of lineage-associated gene markers used in annotating the candidate outliers identified by MultistageOT in the data set from Dahlin et al. \cite{dahlin2018single}. UMAP embedding featuring only the candidate outliers. (a) Monocyte-like (Mo-like) markers from Mildner et al. \cite{mildner2017genomic} and Konturek-Ciesla et al. \cite{konturek2023temporal}. (b) DC (dendritic cell) markers from Lukowski et al. \cite{lukowski2021absence}. (c) pDC (plasmocytoid DC) markers from Lukowski et al. \cite{lukowski2021absence}. (d) Neutrophil markers from Konturek-Ciesla et al. \cite{konturek2023temporal} and Grieshaber-Bouyer et al. \cite{grieshaber2021neutrotime}. (e) Erythroid markers from An et al. \cite{an2014global} and Dzierzak \& Philipsen \cite{dzierzak2013erythropoiesis}. } 
\label{fig:dahlin2018_gene_annotations_1}
\end{figure}

\begin{figure}[H]
    \centering
    \includegraphics[width=\linewidth]{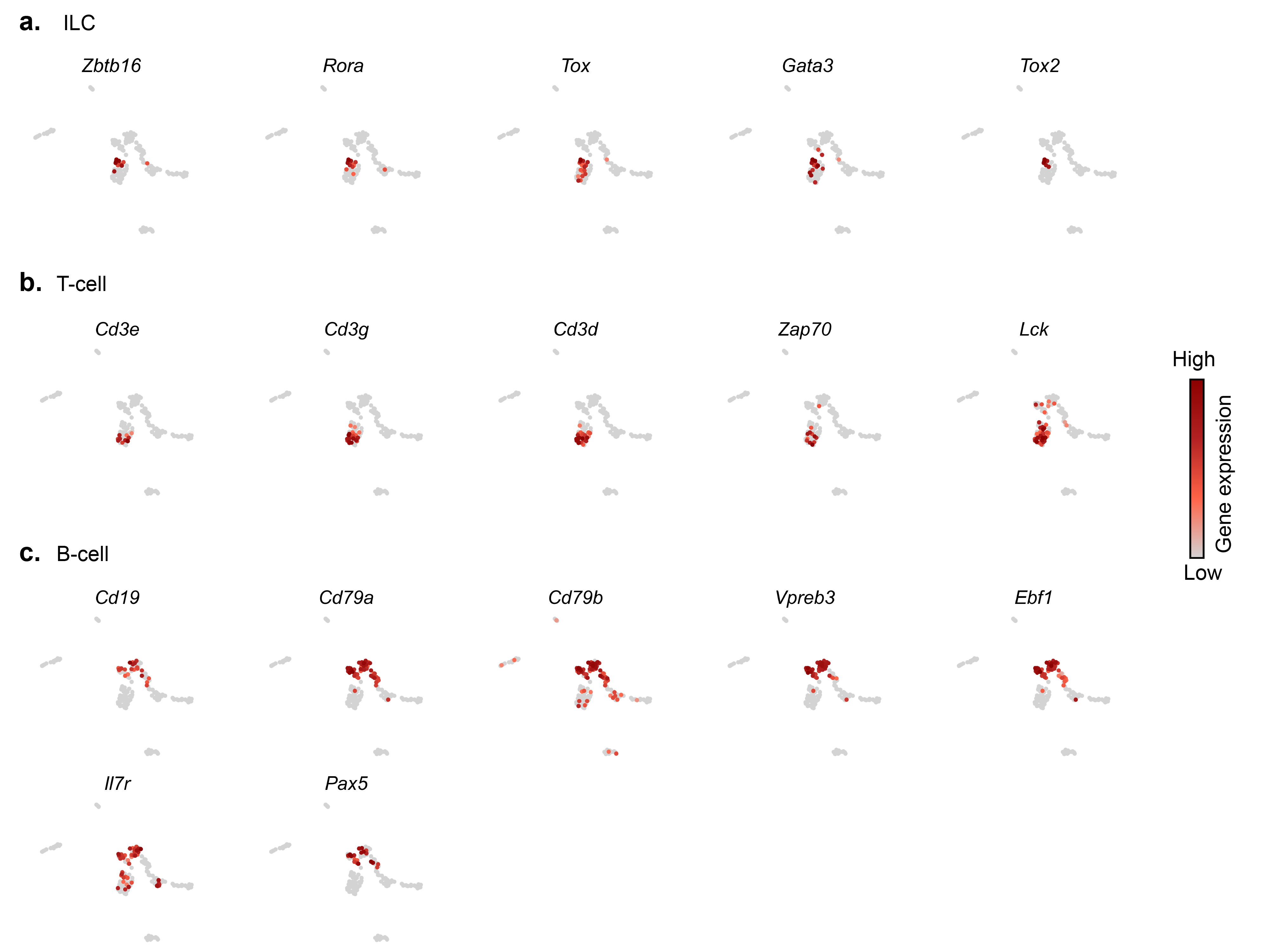} 
   \caption{The expression levels (log($x$+1)-transformed) of lineage-associated gene markers used in annotating the candidate outliers identified by MultistageOT in the data set from Dahlin et al. \cite{dahlin2018single}. UMAP embedding featuring only the candidate outliers. (a) ILC (innate lymphoid cell) markers from Seillet et al. \cite{seillet2016deciphering}. (b) T-cell markers from Rothenberg et al. \cite{rothenberg2008launching} and Rothenberg \cite{rothenberg2014transcriptional}. (c) B-cell markers from Rothenberg \cite{rothenberg2014transcriptional}.} 
\label{fig:dahlin2018_gene_annotations_2}
\end{figure}

\begin{figure}[H]
    \centering
    \includegraphics[width=\linewidth]{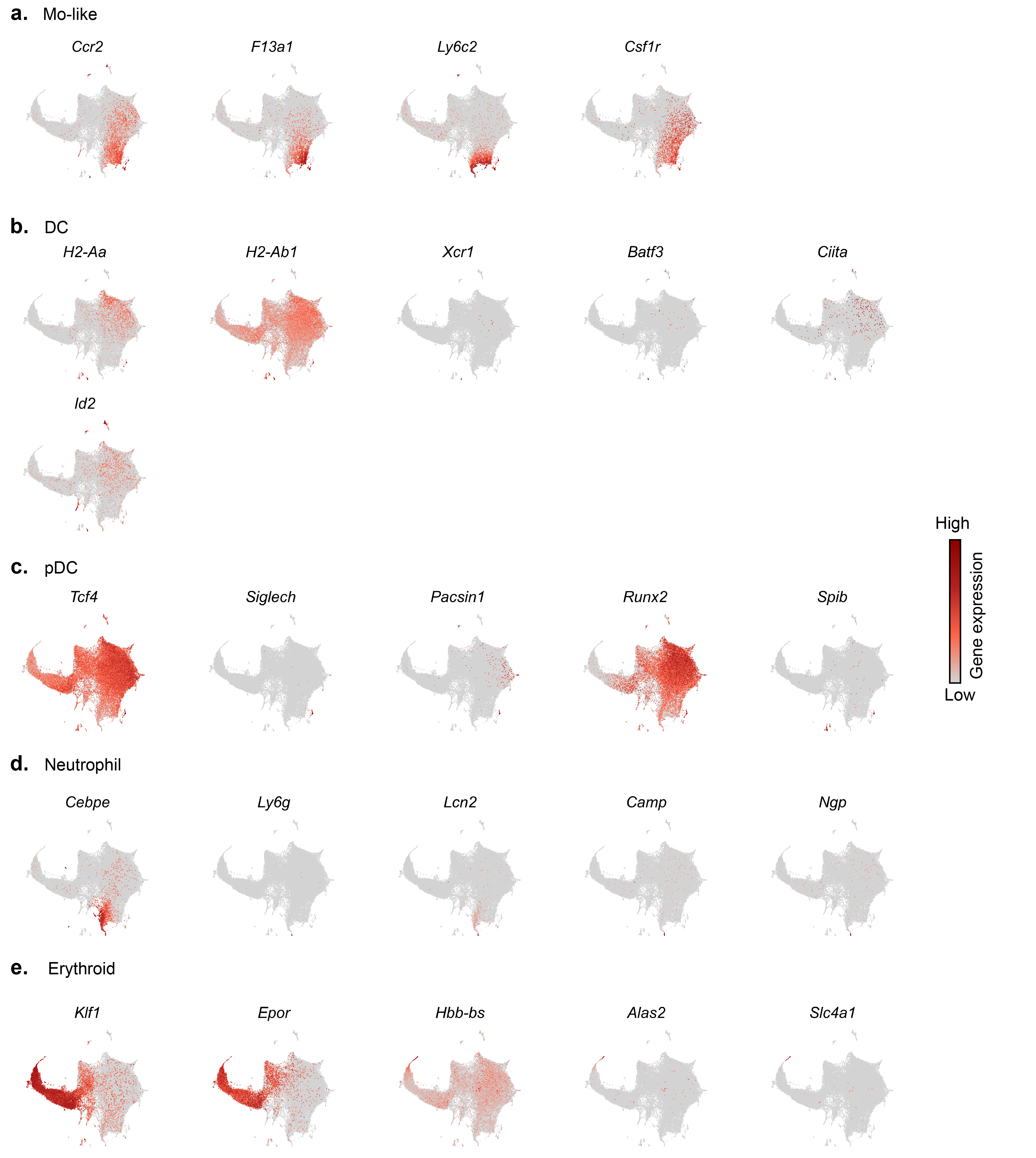}    \caption{The expression levels (log($x$+1)-transformed) of lineage-associated gene markers used in annotating the candidate outliers identified by MultistageOT in the data set from Dahlin et al. \cite{dahlin2018single}. (a) Monocyte-like (Mo-like) markers from Mildner et al. \cite{mildner2017genomic} and Konturek-Ciesla et al. \cite{konturek2023temporal}. (b) DC (dendritic cell) markers from Lukowski et al. \cite{lukowski2021absence}. (c) pDC (plasmocytoid DC) markers from Lukowski et al. \cite{lukowski2021absence}. (d) Neutrophil markers from Konturek-Ciesla et al. \cite{konturek2023temporal} and Grieshaber-Bouyer et al. \cite{grieshaber2021neutrotime}. (e) Erythroid markers from An et al. \cite{an2014global} and Dzierzak \& Philipsen \cite{dzierzak2013erythropoiesis}. } 
\label{fig:dahlin2018_gene_annotations_full_umap_1}
\end{figure}

\begin{figure}[H]
    \centering
    \includegraphics[width=\linewidth]{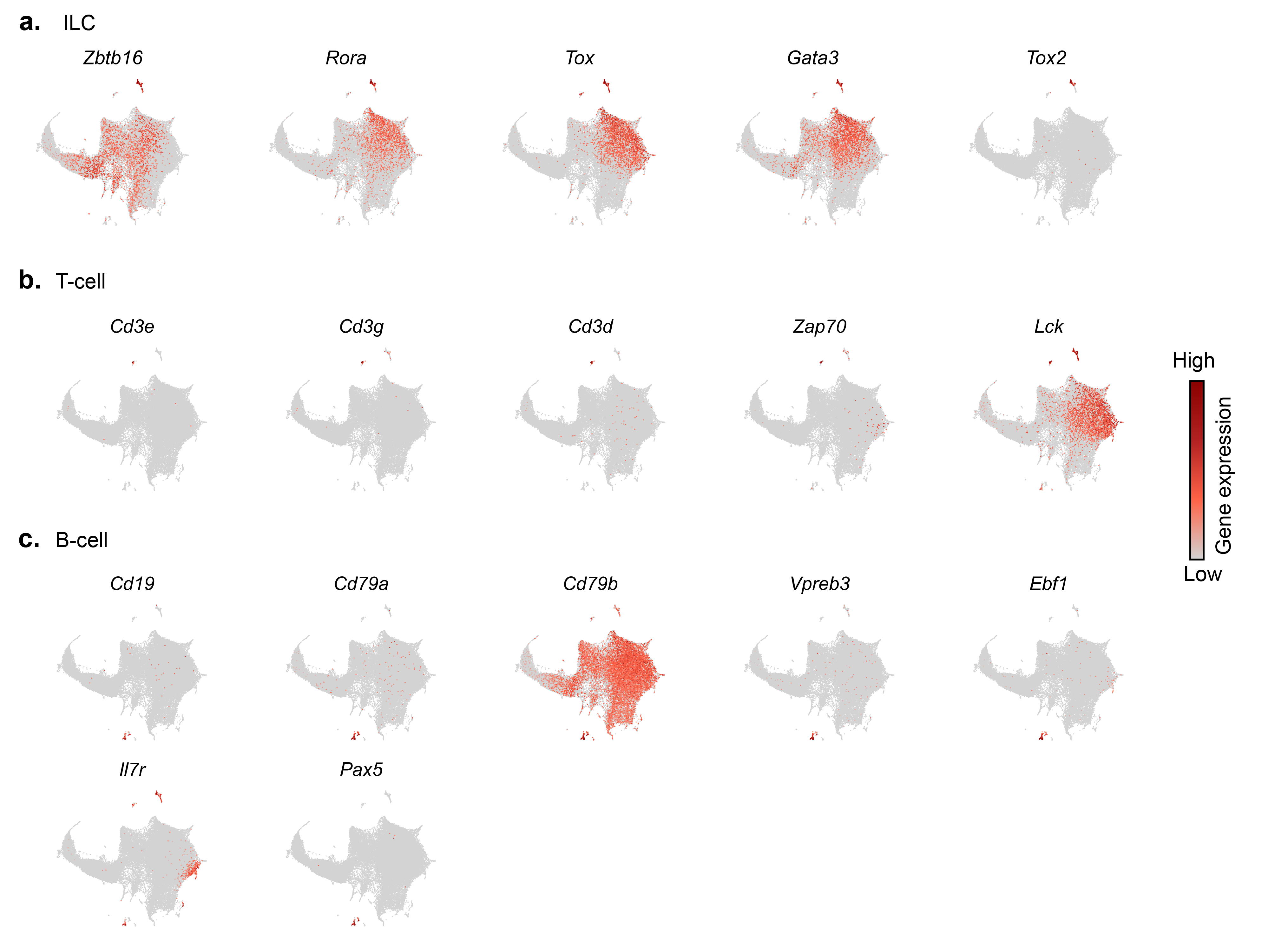} 
   \caption{The expression levels (log($x$+1)-transformed) of lineage-associated gene markers used in annotating the candidate outliers identified by MultistageOT in the data set from Dahlin et al. \cite{dahlin2018single}. (a) ILC (innate lymphoid cell) markers from Seillet et al. \cite{seillet2016deciphering}. (b) T-cell markers from Rothenberg et al. \cite{rothenberg2008launching} and Rothenberg \cite{rothenberg2014transcriptional}. (c) B-cell markers from Rothenberg \cite{rothenberg2014transcriptional}. } 
\label{fig:dahlin2018_gene_annotations_full_umap_2}
\end{figure}

\begin{figure}[H]
    \centering
    \includegraphics[width=\linewidth]{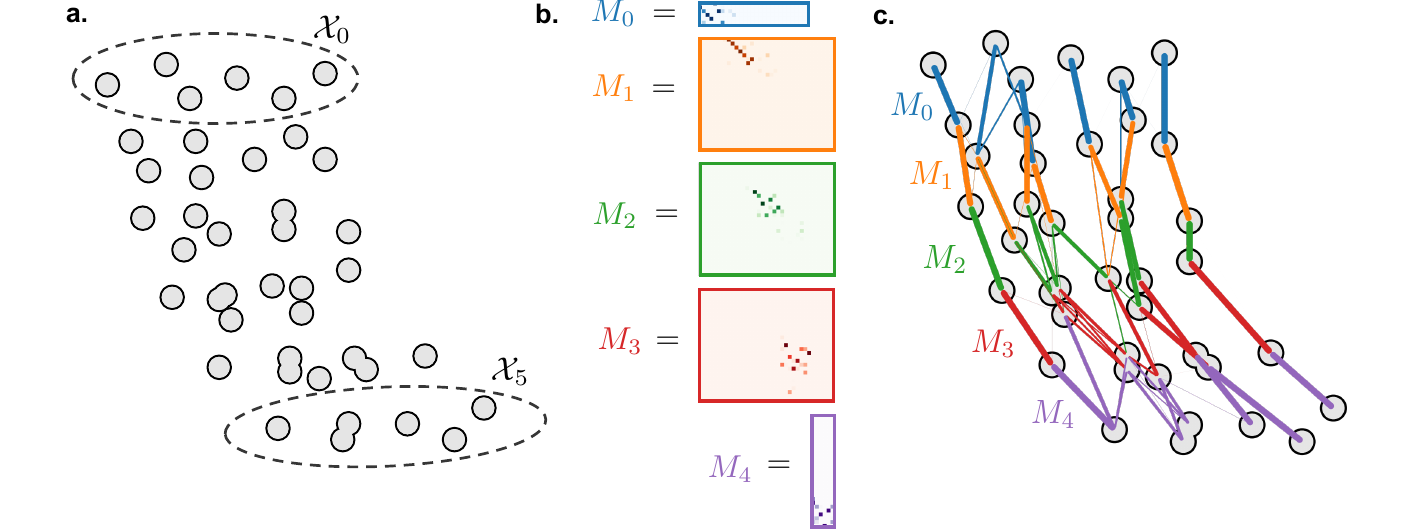}
    \caption{Multistage problem  \eqref{eq:regularizedprimalopt} solved for a two-dimensional toy data set. (a) Initial and terminal states are specified as input to the model. (b) Problem \eqref{eq:regularizedprimalopt} is solved with $\epsilon = 0.005$ and $T=5$. The color intensity of a matrix element $(m_t)_{ij}$ in the transport plan matrices $M_t$, $t = 0,1,2,3,4$ is determined by the amount of mass sent from cell $i$ to cell $j$ in time step $t$. (c) Visualization of the optimal transport plans $M_0 = \tilde{M}_0$, $M_t = \begin{bmatrix}
        \tilde{M}_t & \hat{M}_t
    \end{bmatrix}$, $t=0,1,2,3$,  $M_4 = \hat{M}_{4}$. The width of a colored line connecting two cells is proportional to the amount of mass sent between the cells in the corresponding time step: $t = 0$ (blue);  $t = 1$ (orange);  $t = 2$ (green) ;  $t = 3$ (red);  $t = 4$ (purple). These mass transport plans can be interpreted in terms of transitions likelihoods. }
    \label{fig:visualizationoftransportplans}
\end{figure}

\begin{figure}[H]
    \centering

\begin{tikzpicture}[
    ->,
    >={Stealth[round]},
    node distance=2cm and 4cm,
    every node/.style={circle, draw, minimum size=1cm, inner sep=0pt}
]

\node (A) {$\mathcal{X}_0$};
\node (B) [right of=A] {$\mathcal{X}$};
\node (C) [right of=B] {$\mathcal{X}_F$};
\node (D)[blue] [below of=A] {$\mathcal{A}_0$};
\node (E)[blue] [right of=D] {$\mathcal{A}$};
\node (F)[blue] [right of=E] {$\mathcal{A}_F$};

\draw (A) edge (B); 

\draw (A)[blue, dashed, ->] edge (F);
\draw (B) edge (C);
\draw (B)[blue, dashed, ->] edge[bend left] (F);
\draw (B) edge[loop above] (B);
\draw (D)[blue, dashed, ->] edge (E);
\draw (E)[blue, dashed, ->] edge (F);  
\draw (B)[blue, dashed, ->] edge (E);
\draw (E)[blue, dashed, ->] edge (B);
\draw (D)[blue, dashed, ->] edge[bend right] (B); 
\draw (D)[blue, dashed, ->] edge[bend right=60] (F);
\end{tikzpicture}

    \caption{Graph representation of our multistage optimal transport model of cell differentiation, together with the extension with added auxiliary states, $\mathcal{A}_0$ (initial),  $\mathcal{A}$ (intermediate), $\mathcal{A}_F$ (terminal). The black graph represents the MultistageOT model without the extension (i.e., a collapsed version of the graph in (\hyperref[box:ot]{Box 2}, Fig. 3) and the  blue graph with dashed arrows represent the extension with auxiliary states (\hyperref[sec:supplementary_note]{Supplementary Note}).  }     \label{fig:network_representation_extension}
\end{figure}
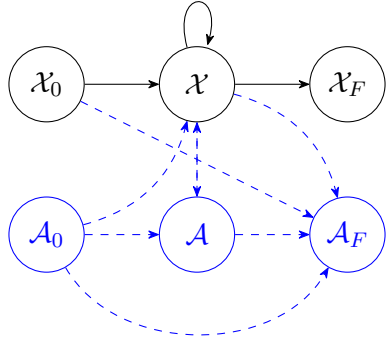
\end{appendix}

\end{bibunit}

\newpage
\pagenumbering{gobble}

\setcounter{section}{0}
\setcounter{page}{0}
\renewcommand{\thesection}{\arabic{section}}
\counterwithin{equation}{section}
\counterwithin{figure}{section}
\counterwithin{table}{section}

\begin{bibunit}
\section*{Supplementary Note}
\label{sec:supplementary_note}

\vspace{1cm}
\section*{\centering \LARGE Mathematical formulation of MultistageOT}
\begin{center}
   \large Magnus Tronstad$^{1*}$, Johan Karlsson$^{2\dag*}$, and Joakim S. Dahlin$^{1\dag*}$
\end{center}  
\begin{center} 
    $^1$Department of Medicine Solna, Karolinska Institutet, and Center for
Molecular Medicine, Karolinska University Hospital, Stockholm,
Sweden.\\ 
    $^2$Department of Mathematics, KTH Royal Institute of Technology,
Stockholm, Sweden. \\
\vspace{0.5cm}
$^\dag$These authors jointly supervised the work. \\
$^*$Corresponding authors. Emails: magnus.tronstad@ki.se, johan.karlsson@math.kth.se, and joakim.dahlin@ki.se
\end{center}

\pagenumbering{arabic}

\section{Notation}
Whenever the exponential function is applied on vectors or matrices, it is assumed that it works component-wise. Also we use the following notation:

\vspace{0.5cm}
\begin{tabular}{p{1.5cm}p{9cm}p{1cm}}
$\odot$ & Component-wise multiplication (Hadamard product). \\
$\oslash$ & Component-wise division.  \\
$\langle\cdot, \cdot\rangle$ & Frobenius inner product.\\
$\delta_x$  & Dirac delta point measure, corresponding to a unit mass located in $x$. Formally: 

\begin{equation*}\delta_x(A) = \begin{cases}
    0, \quad x \notin A\\
    1, \quad x \in A,
\end{cases}
\end{equation*}
for any measurable set $A$.\\
$(x_i)_{i=1}^{n}$ & Ordered set of elements $x_i$, $i= 1,\dots,n$.\\
$\mR_+$ &The set $[0,\infty)$. \\
$\overline\mR_+$ &The nonnegative real line, including infinity $\mR_+\cup \{\infty\}$. \\
$\ones{n}$ & A vector of all-ones of length $n$.
\end{tabular}

\section{Background}

Optimal transport solves the problem of moving one distribution into another in an optimal way \cite{villani2021topics}. This general framework can be traced back to the 18th century \cite{monge1781memoire}, and traditional areas of applications include economics and logistics  \cite{hitchcock1941distribution, kantorovich1942translocation, kantorovich1960mathematical}. Extension to a broader range of applications was for a long time hampered, as the computational cost associated with numerically computing an optimal transport solution placed severe limits on problem size. Notwithstanding, following a computational breakthrough due to Cuturi \cite{cuturi2013sinkhorn}, it is now possible to address and solve problems that were historically too large for standard numerical methods. This is accomplished by adding an entropic regularization term to the classical optimal transport formulation; Sinkhorn iterations \cite{sinkhorn1967concerning} are then leveraged to efficiently obtain a near-optimal solution, even for very large problem sizes. 
Since then, the interest for optimal transport has increased rapidly and it has been applied to a wide range of problems, including image processing \cite{rabin2014adaptive, ferradans2014regularized}, tracking and sensor fusion \cite{elvander2020multi}, ensemble control problems \cite{haasler2021control}, fluid mechanics \cite{benamou2000computational}, the Schrödinger bridge problem \cite{chen2016relation} as well as flow problems \cite{haasler2023scalable, mascherpa2023estimating} and applications to biological systems \cite{sandhu2015graph, farooq2019network}.

Recently, optimal transport theory has successfully been applied in the context of cellular development \cite{schiebinger2019optimal, yang2020, forrow2021lineageot, zhang2021optimal, lamoline2024gene}. Many of these studies have been designed around time-series matching. In this setting, one is faced with multiple snapshots of single-cell sequencing data: Each snapshot comprises a large collection of individual measurements of cell states from a specific time point, and the goal is to match cells from consecutive time points in a biologically meaningful way. The cell states can be represented as point masses in state space, and a cost can be associated with the transportation of a unit mass between any pair of cells in two consecutive time points. Optimal transport then provides a natural mathematical framework for finding  optimal (least costly) assignments of cell states between time points. In 2019, Schiebinger et al. \cite{schiebinger2019optimal} introduced Waddington-OT which uses optimal transport to find the optimal matching of single cells in pairs of consecutive time points. In Forrow et al.  \cite{forrow2021lineageot}, this framework is combined with lineage information to improve the reconstruction of developmental trajectories. Yang et al. \cite{yang2020} extend the Waddington-OT approach by computing the transport costs in a latent space generated by an autoencoder. 

Time-series data is not the only means by which one can study gene expression dynamics of cellular populations. When the developmental process originates in a population of continuously self-replenishing cell states (such as in blood cell development), a single sequencing snapshot can be expected to represent cells from a range of different stages of maturation, from blood stem cells to late stage progenitors. Thus, a second approach to inferring developmental trajectories relies on analyzing a single snapshot of single-cell sequencing data. Zhang et al. \cite{zhang2021optimal}, presented StationaryOT, showing that classical bimarginal discrete optimal transport can be used for such within-snapshot trajectory inference when cell growth rate estimates are available.  

In this work, we propose MultistageOT, a novel modeling framework for trajectory inference within a snapshot. MultistageOT extends bimarginal optimal transport so that transport occurs across multiple (more than two) marginals. This allows cell differentiation in a snapshot to be modeled as a transition process on gene expression over a number of transport steps, each step corresponding to an intermediate differentiation stage. This introduces a temporal axis to our model, which provides a natural way to order the cells in terms of maturity in the absence of time-resolved data. 

This note contains information presented in the main article but expands on it for a complete mathematical exposition. In the remainder of this note, we briefly review some preliminaries on classical discrete optimal transport before presenting a mathematical formulation of our MultistageOT model of cell differentiation. We state the entropy-regularized optimization problem solved in MultistageOT, and then derive a generalized Sinkhorn-Knopp algorithm for efficient computation of the optimal solution. 

\subsection{Optimal transport}
\label{sec:optimaltransport}
MultistageOT is based on an extension of discrete optimal transport. In the classical discrete optimal transport setting, two discrete mass distributions are given:
\[
\sum_{i=1}^{n_1} \mu^{(i)}_1\delta_{x^{(i)}_1}, \quad \sum_{j=1}^{n_2} \mu^{(j)}_2\delta_{x^{(j)}_2}, 
\]
where $\mathcal{X}_1=(x_1^{(i)})_{i=1}^{n_1}$ and $\mathcal{X}_2=(x_2^{(j)})_{j=1}^{n_2}$ correspond to the points of support of the respective distributions, and where $\mu_1^{(i)}$ and $\mu_2^{(j)}$ represent the mass in the points $x_1^{(i)}$ and  $x_2^{(j)}\!,$ respectively. 
A transport plan is a matrix, $M=[m_{ij}]_{i=1, j=1}^{n_1, n_2}$, whose components $m_{ij}$ denote the amount of mass transported from  $x^{(i)}_1 \in \mathcal{X}_1$ to $x^{(j)}_2 \in \mathcal{X}_2$. We say that the transport plan is feasible if the total amounts transported are consistent with the initial and final distributions, i.e., if $M\ones{n_2}=\mu_1$ and $M^T\ones{n_1}=\mu_2$, where $\mu_1 = [\mu_1^{(i)}]_{i=1}^{n_1}$ and $\mu_2 = [\mu_2^{(j)}]_{j=1}^{n_2}$ are the vectors representing the mass distributions.
Next, let $c_{ij}=c\left(x^{(i)}_1, x^{(j)}_2\right)$ denote the cost of moving a unit of mass from $x^{(i)}_1$ to $ x^{(j)}_2$\!, and let $C=[c_{ij}]_{i=1, j=1}^{n_1, n_2}$ be the corresponding cost matrix.  
The optimal transport problem, then, is to find a feasible transport plan that moves the mass in $\mu_1$ to $\mu_2$ with minimal cost, i.e., 
\begin{subequations}
\label{eq:standardOT}
\begin{align}
  \mathcal{T}(\mu_1, \mu_2) := \underset{M \in \mathbb{R}_+^{n_1\times n_2}}{\text{minimize}}   \quad &\langle C, M  \rangle \label{eq:standardOTobjective}  \\
   \mbox{subject to} \quad  &M\ones{n_2} \;\;= \mu_1 \\
   &M^T \ones{n_1} = \mu_2,
   \end{align}
\end{subequations}
where 
$\langle C, M  \rangle = \sum_{i=1}^{n_1}\sum_{j=1}^{n_2}  c_{ij}m_{ij}$. This problem is referred to as the Kantorovich formulation \cite{villani2021topics}, and is a linear programming problem. For certain costs, this problem can be used to define a metric space. For example, when $c(x_1^{(i)},x_2^{(j)})=\|x_1^{(i)}-x_2^{(j)}\|^2$, the Wasserstein-2 metric between two distributions with equal mass is defined as the square root of the optimal transport cost \eqref{eq:standardOT}.

For small enough problems, the formulation \eqref{eq:standardOT} can be solved using standard methods for linear programs such as the Simplex algorithm or interior-point methods. However, as the number of points in each support grows large, these methods become less practical;  to efficiently address large scale problems, one can introduce an entropic regularization term in the objective as proposed by Cuturi \cite{cuturi2013sinkhorn}. The entropy-regularized optimal transport problem is then formulated as 
\begin{subequations}
 \label{eq:standardOTregularized}
\begin{align}
  \mathcal{T}_\epsilon(\mu_1, \mu_2) := \underset{M \in \mathbb{R}_+^{n_1\times n_2}}{\text{minimize}}  \quad &\langle C, M  \rangle 
 + \epsilon D(M|P) \label{eq:standardOTregularizedobjective}  \\
   \text{subject to} \,\, \quad  &M\ones{n_2} \;\;= \mu_1 \\
   &M^T \ones{n_1} = \mu_2,
   \end{align}
\end{subequations}
where the scalar $\epsilon > 0$ is a parameter determining the weight given to the entropy term
\begin{equation}
        D(M|P) = \sum_{i=1}^{n_1}\sum_{j=1}^{n_2}\bigg( m_{ij} \log \frac{m_{ij}}{p_{ij}} - m_{ij} + p_{ij} \bigg)
    \end{equation}
and where the matrix $P = [p_{ij}]_{i=1,j=1}^{n_1,n_2}$ can represent a prior distribution \cite{peyre2019computational}. In general, $P$ can be any matrix with strictly positive elements, but common choices include $P = \ones{n_1}\ones{n_2}^T$ \cite{cuturi2013sinkhorn} or $P = \mu_1 \mu_2^T$. It can be shown \cite{cuturi2013sinkhorn} that the solution to \eqref{eq:standardOTregularized} is in the form
\begin{align}
   M_\epsilon^* = K \odot (u_1 u_2^T)
\end{align}
where $K = \exp(-C/\epsilon)\odot P$, and the vectors $u_1\in \mR^{n_1} $, $u_2\in \mR^{n_2}$ can be obtained through Sinkhorn iterations:
\begin{align}
    u_1 \leftarrow \mu_1 \oslash (K u_2), \quad u_2 \leftarrow \mu_2 \oslash (K^T u_1).
\end{align}
Note that these iterations can be performed efficiently since the bottleneck is matrix-vector multiplication. It can also be shown that if the costs are finite the algorithm  converges linearly \cite{franklin1989scaling}.
There are several ways to interpret Sinkhorn's method: scaling via diagonal matrix multiplication \cite{sinkhorn1967concerning}, iterative Bregman projections  \cite{benamou2015iterative}, Dykstra’s algorithm \cite{benamou2015iterative}, or dual coordinate ascent \cite{karlsson2017generalized}.

\section{Multistage optimal transport}
We model a snapshot of single-cell RNA sequencing (scRNA-seq) data as a collection of data points generated by multiple trajectories from a continuous time-invariant (unknown) dynamical system with state variables, $x_k(t) \in \mathbb{R}$, $k = 1,\dots,m$,  corresponding to the gene expression level of gene $k$ at time $t \in \mathbb{R}_+$ (note that $t$ does not correspond to elapsed real time; rather, it measures progression through the differentiation process). Under this framework, the transcriptional state of a cell is represented by the state vector $x(t) = (x_1(t), x_2(t), \dots, x_m(t)) \in \mathbb{R}^m$ which, as it develops, traces a trajectory in $m$ dimensional gene expression space (one dimension per gene). 

Now assume we are given a snapshot $\mathcal{D}=\{x^{(i)}\,|\, i=1,\ldots, N\}$, corresponding to a collection of cell states, $x^{(i)} = (x_1^{(i)}, x_2^{(i)}, \dots, x_m^{(i)})\in \mR^m$, representing the measured RNA expression levels of $m$ different genes, for each cell $i=1,\ldots, N$. Each cell state is assumed to have been generated by sampling a state trajectory from the underlying dynamical system after some time $t$. We thus assume that  $\mathcal{D}$ contains cells from a range of intermediate stages of differentiation. Let the first $n_0$ cells in $\mathcal{D}$ correspond to the most immature cell states; we refer to these as the initial states, and denote them by $\mathcal{X}_0$. Let the last $n_F$ cells correspond to the most mature cell states; we refer to these as the terminal states, and denote them by $\mathcal{X}_F$. The remaining cells are referred to as intermediate states and are denoted $\mathcal{X}$ (see \hyperref[box:msot]{Box 2}, Fig. 2 in main article).
 These three ordered sets are thus given by $\mathcal{X}_0=(x^{(i)}\,|\, i=1,\ldots, n_0)$, $\mathcal{X}_F=(x^{(i)}\,|\, i=N-n_F+1,\ldots, N)$, and
$\mathcal{X}=(x^{(i)}\,|\, i=n_0+1,\ldots, N-n_F)$. 
We assume that $\mathcal{X}_0$ and $\mathcal{X}_F$  are known, and one main problem that needs to be solved is to determine how the cells in $\cX$ should be ordered in terms of temporal progression through the differentiation process.

With MultistageOT, we model the cell differentiation process as a process on gene expression space, in which the cells corresponding to the initial states in $\mathcal{X}_0$ are assumed to transition into the terminal states in $\mathcal{X}_F$ over at most $T$ discrete time steps, using all the intermediate states as possible transit hubs. With each state transition we associate a cost, 
modeling the likelihood of the transition, and we seek the transitions that minimize the overall cost. As we will see, this can be formulated using optimal mass transport. 

To that end, we equip each cell with a notion of ``mass'', that can be transported to other cells, and define the cost\footnote{While the norm used in the cost in  \eqref{eq:sqeuclideancost} is the squared Euclidean norm, it should be noted that this framework in principle can be used with any type of norm or generalized notion of a transport cost between states.} of transporting a unit of mass between any two cells with states  $x^{(i)}, x^{(j)} \in \mathbb{R}^m$ as
\begin{align}
\label{eq:sqeuclideancost}
    c(x^{(i)}, x^{(j)}) = \big \lVert x^{(i)} - x^{(j)} \big \rVert_2^2,
\end{align}
i.e., the squared Euclidean distance between $x^{(i)}$ and $x^{(j)}$. Moreover, to allow for more concise notation, with two ordered sets, $\mathcal{X}_1 = (x_1^{(i)}\,|\, i=1,\ldots, n_1)$ and $\mathcal{X}_2 = (x_2^{(j)}\,|\, j=1,\ldots, n_2)$,  we associate a cost matrix $c(\mathcal{X}_1,\mathcal{X}_2) \in \overline\mR_+^{n_1\times n_2}$, with elements
\begin{equation} 
\label{eq:matrixcost}
    c(\mathcal{X}_1,\mathcal{X}_2) =\left[c(x_1^{(i)}, x_2^{(j)})\right]_{i=1,j=1}^{n_1,n_2} = \left[ \big \lVert x^{(i)} - x^{(j)} \big \rVert_2^2\right]_{i=1,j=1}^{n_1,n_2}.
\end{equation}

To reflect biology, we assume that the mass transportation starts in the initial states, $\mathcal{X}_0$, and terminates in the terminal states, $\mathcal{X}_F$. In each time step, cells can transport mass to other cells in either $\mathcal{X}$ or $\mathcal{X}_F$. Any mass transported to $\mathcal{X}$ remains within the system and can be further transported in subsequent time steps, whereas mass sent to $\mathcal{X}_F$ exits the system---representing a terminal differentiation step. Let $\tilde{\mu}_0^{(i)}$, for $i=1,\dots,n_0$, denote the portion of mass that is transported in $t = 0$ from cell $x^{(i)} \in \mathcal{X}_0$ to cells in $\mathcal{X}$, and let $\hat{\mu}_t^{(i)}$ denote the corresponding portion that is transported to $\mathcal{X}_F$.  For notational brevity, let these two mass distributions be represented by the vectors
\begin{align}
    \tilde{\mu}_0 &= \left[\tilde{\mu}_0^{(i)}\right]_{i=1}^{n_0}, \quad (\text{mass to $\mathcal{X}$, \textbf{remains} in system})\\
    \hat{\mu}_0 &= \left[\hat{\mu}_0^{(i)}\right]_{i=1}^{n_0} \quad \enspace (\text{mass to $\mathcal{X}_F$, \textbf{exits} system}).
\end{align}
Similarly, in subsequent time steps $t = 1,\dots,T-1$, cells in $\mathcal{X}$ must transport mass received in the previous time step, $t - 1$, to cells either in $\mathcal{X}$ or $\mathcal{X}_F$. As before, whatever mass is sent to $\mathcal{X}_F$ will exit the system. Thus, for $i=1,\dots, N-n_0-n_F$, we let $\tilde{\mu}_t^{(i)}$ denote the portion of mass transported in time step $t$ from cell $x^{(n_0+i)} \in \mathcal{X}$ to cells in $\mathcal{X}$ and let $\hat{\mu}_t^{(i)}$ denote the corresponding portion transported to $\mathcal{X}_F$. We define, for each $t = 1,\dots,T-1$, the corresponding vectors

\begin{align}
\tilde{\mu}_t &= \left[\tilde{\mu}_t^{(i)}\right]_{i=1}^{N-n_0-n_F}, \quad (\text{mass to $\mathcal{X}$, \textbf{remains} in system}) \\
\hat{\mu}_t &= \left[\hat{\mu}_t^{(i)}\right]_{i=1}^{N-n_0-n_F} \quad \enspace(\text{mass to $\mathcal{X}_F$, \textbf{exits} system}) .
\end{align}

Finally, let $\tilde{\nu}_t = [\tilde{\nu}_t^{(i)}]_{i=1,\dots,N-n_0-n_F}$, for $t=0,1,\dots, T-1$, be the vector whose element $\tilde{\nu}_t^{(i)}$ denotes the mass received by intermediate states $x^{(n_0+i)} \in \mathcal{X}$ ($i=1,\dots,N-n_0-n_F$)  in time step $t$, and let $\hat{\nu}_t = [\hat{\nu}_t^{(i)}]_{i=1,\dots,n_F}$, for $t=0,1,\dots, T-1$, be the vector whose element $\hat{\nu}_t^{(i)}$ denotes the mass received by terminal state $x^{N-n_F+i}$ ($i=1,\dots,n_F$) in time step $t$. A schematic overview of this sequential mass transport process is shown in (\hyperref[box:ot]{Box 2}, Fig. 3 in main article).

In this work, we take the global minimum cost transport plans---describing how mass should most efficiently be sent through all the cell states, over all time points---as a model for possible differentiation trajectories. To achieve this, we constrain the transport plans in the following way: 
\begin{enumerate}
    \item The total mass transported from a cell in time step $t$ must be equal to that which was received by that same cell in the previous time step, i.e.,
\begin{align}
    \tilde{\mu}_t + \hat{\mu}_t  &= \tilde{\nu}_{t-1}, \quad \text{for}\enspace t = 0,1,\dots,T-2, \\
    \hat{\mu}_{T-1} &= \tilde{\nu}_{T-2}.
\end{align}
Note that in the last step of transport, the intermediate states can only transport mass to cells in $\mathcal{X}_F$.
\item Every initial state should transport at least a unit amount of mass, i.e., we impose that
\begin{equation}
    \tilde{\mu}_0 + \hat{\mu}_0 \geq \ones{n_0}.
\end{equation}
\item Every intermediate state should in total, over all time steps $t=1,\dots,T-1$, transport at least a unit amount of mass. We formulate this as
\begin{equation}
   \sum_{t=1}^{T-2} \tilde{\mu}_{t} + \sum_{t=1}^{T-1}\hat{\mu}_{t}   \geq \ones{N-n_0-n_F}.
\end{equation}
\item Every terminal state should, over all time steps $t=1,\dots,T-1$, receive at least a unit amount of mass, i.e., we impose that
\begin{equation}
    \sum_{t=0}^{T-1}\hat{\nu}_t   \geq \ones{n_F}.
\end{equation}
\end{enumerate} 
Now, let 
\begin{align}
    \tilde{M}_0 &= \left[(\tilde{M}_0)_{ij} \right]_{i=1,j=1}^{n_0,N-n_0-n_F}
\end{align}
be the transport plan matrix whose components $(\tilde{M}_0)_{ij}$, for $i=1,\dots,n_0$, $j=1,\dots,N-n_0-n_F$, denote the mass sent from $x^{(i)} \in \mathcal{X}_0$ to $x^{(n_0+j)} \in \mathcal{X}$ in time step $t = 0$, and let 
\begin{align}
    \hat{M}_0 &= \left[(\hat{M}_0)_{ij} \right]_{i=1,j=1}^{n_0,n_F}
\end{align}
be the transport plan matrix whose components $(\hat{M}_0)_{ij}$, for $i=1,\dots,n_0$, $j=1,\dots,n_F$, denote the amount of mass sent from $x^{(i)} \in \mathcal{X}_0$ to $x^{(N-n_F+j)} \in \mathcal{X}_F$ in time step $t = 0$. Note that the following relations hold
\begin{equation}
\label{eq:M0-to-mu0}
    \tilde{\mu}_0 = \tilde{M}_0 \ones{n_0}, \quad     \hat{\mu}_0 = \hat{M}_0 \ones{n_F}
\end{equation}
Analogously, we define the transport matrices
\begin{align}
        \tilde{M}_t &= \left[(\tilde{M}_t)_{ij} \right]_{i=1,j=1}^{N-n_0-n_F,N-n_0-n_F}, \quad \text{for $t=1,\dots,T-2$}\\
            \hat{M}_t &= \left[(\hat{M}_t)_{ij} \right]_{i=1,j=1}^{N-n_0-n_F,n_F} \quad \text{for $t=1,\dots,T-1$}
\end{align} 
such that
\begin{align}
\label{eq:Mt-to-mut-tilde}
    &\tilde{\mu}_t = \tilde{M}_t \ones{N-n_0-n_F}, \quad \text{for $t=1,\dots,T-2$} \\ 
    \label{eq:Mt-to-mut-hat} &\hat{\mu}_t = \hat{M}_t \ones{n_F} \quad \quad\quad\quad \,\text{for $t=1,\dots,T-1$}.
\end{align}
For each transport plan matrix, for $t = 0,1,\dots, T-1$, we associate corresponding cost matrices $\tilde{C}_t, \hat{C}_t$ of matching sizes (i.e., $\tilde{C}_t$ corresponds to $\tilde{M}_t$ and  $\hat{C}_t$  to $\hat{M}_t$), so that $(\tilde{C}_t)_{ij}$ and $(\hat{C}_t)_{ij}$ encodes the cost associated with mass transports $(\tilde{M}_t)_{ij}$ and $(\hat{M}_t)_{ij}$ respectively. 

\subsection{Optimization problem}
Recall that, in \eqref{eq:standardOT}, $\mathcal{T}(\mu_1,\mu_2)$ is defined as the optimal objective function value in the optimal mass transport problem of moving $\mu_1$ to $\mu_2$ in a minimum cost way. Using this formalism, the multistage optimal mass transport problem over all time steps can be formulated as
\begin{equation*}
 (P) \quad \begin{aligned}
  \label{eq:primalopttransport}
  \underset{\underset{t=0,\dots,T-1 }{\tilde{\mu}_t,\hat{\mu}_t\hat{\nu}_t}}{\text{minimize}} \quad  &\sum_{t=0}^{T-2}\tilde{\mathcal{T}}_t (\tilde{\mu}_t,\tilde{\mu}_{t+1}+\hat{\mu}_{t+1}) + \sum_{t=0}^{T-1}\hat{\mathcal{T}}_t (\hat{\mu}_t,\hat{\nu}_{t})  \\
	\text{subject to}  \quad   &\tilde{\mu}_{0} + \hat{\mu}_{0}  \geq \ones{n_0} \\
&\sum_{t=1}^{T-2} \tilde{\mu}_{t} + \sum_{t=1}^{T-1}\hat{\mu}_{t}    \geq \ones{N - n_0 - n_F} \\
 &\sum_{t=0}^{T-1} \hat{\nu}_t   \geq \ones{n_F}.
	\end{aligned}
\end{equation*}
A more explicit formulation can be obtained by expressing (\hyperref[eq:primalopttransport]{$P$})  directly in terms of the transport plan matrices $\tilde{M}_t$ for $t = 0,1,\dots, T-2$ and $\hat{M}_t$ for $t = 0,1,\dots, T-1$. For brevity, we define $n = N-n_0-n_F$, and let $n_t = n$ for $t = 1,\dots, T-1$ Then,

\begin{subequations}
  \label{eq:primalopttransportexplicit}
  \begin{align}
  \underset{\underset{t=0,\dots,T-1 }{\tilde{M}_t, \hat{M}_t\geq 0} }{\text{minimize}} \quad  &\sum_{t=0}^{T-2}  \langle \tilde{C}_t, \tilde{M}_t \rangle    + \sum_{t=0}^{T-1} \langle \hat{C}_t, \hat{M}_t \rangle  \label{eq:primalopttransportexplicit_a} \\
	\text{subject to} \quad &\tilde{M}_0 \ones{n_1} + \hat{M}_0\ones{n_F} \geq \ones{n_0} \label{eq:primalopttransportexplicit_b}\\
 &\tilde{M}_t\ones{n_{t+1}} + \hat{M}_t\ones{n_F} = \tilde{M}_{t-1}^T\ones{n_{t-1}} \quad \text{for $t=1,\dots,T-2$}\label{eq:primalopttransportexplicit_c}\\
&\hat{M}_{T-1}\ones{n_F} = \tilde{M}_{T-2}^T\ones{n_{T-2}} \label{eq:primalopttransportexplicit_d} \\
&\sum_{t=1}^{T-2}  \tilde{M}_t\ones{n_{t+1}} + \sum_{t=1}^{T-1}\hat{M}_t\ones{n_F}     \geq \ones{n} \label{eq:primalopttransportexplicit_e}\\
 &\sum_{t=0}^{T-1} \hat{M}_t^T \ones{n_t}   \geq \ones{n_F}.\label{eq:primalopttransportexplicit_f}
	\end{align}
\end{subequations} 
Note that \eqref{eq:primalopttransportexplicit} is a linear program. However, due to the sizes of modern scRNA-seq data sets, it is not practical to solve the problem \eqref{eq:primalopttransportexplicit} with standard methods for linear programs. Moreover, due to the large amount of noise inherent in the data, the transport maps in the solution to \eqref{eq:primalopttransportexplicit}, lying on an edge of the feasible polytope, are likely overly sparse (i.e., mapping one cell to few other cells). To resolve these issues, we instead solve an approximate optimal transport problem in which a regularizing entropy term is added to the objective function (see \hyperref[sec:optimaltransport]{\autoref{sec:optimaltransport}} for details). This regularization not only allows us to derive an efficient iterative algorithm for solving the problem numerically, but it also promotes more diffuse transport maps, reducing the variance in the estimated cell-cell couplings. The regularized problem solved in MultistageOT is

\begin{subequations}
  \label{eq:regularizedprimalopt}
  \begin{align}
  \underset{\underset{t=0,\dots,T-1 }{\tilde{M}_t, \hat{M}_t}}{\text{minimize}} \quad  &\sum_{t=0}^{T-2}\left( \langle \tilde{C}_t, \tilde{M}_t \rangle  + \epsilon D(\tilde{M}_t|\tilde{P}_t) \right) +\sum_{t=0}^{T-1} \left(\langle \hat{C}_t, \hat{M}_t \rangle + \epsilon D(\hat{M}_t|\hat{P}_t) \right)  \label{eq:regularizedprimalopt_a} \\
	\text{subject to} \quad &\tilde{M}_0 \ones{n_1} + \hat{M}_0\ones{n_F} \geq \ones{n_0} \label{eq:regularizedprimalopt_b}\\
 &\tilde{M}_t\ones{n_{t+1}} + \hat{M}_t\ones{n_F} = \tilde{M}_{t-1}^T\ones{n_{t-1}}, \quad \text{for $t=1,\dots,T-2$}\label{eq:regularizedprimalopt_c}\\
&\hat{M}_{T-1}\ones{n_F} = \tilde{M}_{T-2}^T\ones{n_{T-2}}, \label{eq:regularizedprimalopt_d} \\
&\sum_{t=1}^{T-2}  \tilde{M}_t\ones{n_{t+1}} + \sum_{t=1}^{T-1}\hat{M}_t\ones{n_F}     \geq \ones{n} \label{eq:regularizedprimalopt_e}\\
 &\sum_{t=0}^{T-1} \hat{M}_t^T \ones{n_t}   \geq \ones{n_F}, \label{eq:regularizedprimalopt_f}
	\end{align}
\end{subequations} 
 where $\tilde{P}_t$ and $\hat{P}_t$ are ``prior'' transport plan matrices corresponding to $\tilde{M}_t$ and $\hat{M}_t$ respectively (see \ref{sec:optimaltransport}). Intuitively, the entropy-regularization penalizes deviations from these prior matrices. In case $\tilde{P}_t$ and $\hat{P}_t$  are simply taken to be matrices of ones, it will favor less sparse optimal transport plans in general. Thus, the solution depends on the choice of the regularization parameter, $\epsilon$ (smaller values of $\epsilon$ promote sparser transport plans, whereas larger values promote denser transport plans). In \hyperref[fig:visualizationoftransportplans]{\autoref{fig:visualizationoftransportplans}}, MultistageOT is applied to a 2D toy data set to visualize the optimal solution to \eqref{eq:regularizedprimalopt}.

The marginals $\tilde{\mu}_t$, for $t=0,1\dots,T-2$, and  $\hat{\mu}_t$,  for $t=0,1\dots,T-1$, are retrieved from the optimal transport plans via \eqref{eq:M0-to-mu0}, \eqref{eq:Mt-to-mut-tilde} and \eqref{eq:Mt-to-mut-hat}. The total mass sent from every intermediate state is given by
\begin{align}
\label{eq:full_marginal_distributions}
    \mu_t := \tilde{\mu}_t + \hat{\mu}_t \quad \text{for $t=1,\dots,T-2$},
\end{align}
and $\mu_{T-1}:= \hat{\mu}_{T-1}$. These intermediate marginals can be scaled so that they represent, for each intermediate state, the probability of that state belonging to a particular transport stage $t$. This allows us to compute a mean transport stage for each state and thereby order cells in terms of temporal progression through the differentiation process (see \hyperref[fig:synthetic-data-results]{\autoref{fig:synthetic-data-results}}c-e, and \hyperref[sec:methods]{Methods}, \autoref{sec:downstream} in the main article for more details on downstream analyses based on the optimal transport solution). 

\section{Algorithm}
The iterative algorithm for solving the entropy-regularized MultistageOT problem  \eqref{eq:regularizedprimalopt} presented in this work can be viewed as a form of generalized Sinkhorn-Knopp \cite{sinkhorn1967concerning} iterations, and we here derive it using Lagrangian duality. The iterations can be interpreted as performing block-coordinate ascent on the dual to \eqref{eq:regularizedprimalopt}.
\subsection{Duality}
To allow for more succinct notation, we let 
\begin{equation}
    \mathbf{M} = (\tilde{M}_0, \hat{M}_0, \tilde{M}_1, \hat{M}_1\dots, \tilde{M}_{T-2}, \hat{M}_{T-2}, \hat{M}_{T-1})
\end{equation}
denote the tuple of all transport plans, and we let $\mathbf{\Lambda} = (\rho,\lambda_0,\lambda_1,\dots, \lambda_T)$ denote the tuple of all the Lagrange multipliers (i.e., the primal and dual variables respectively).

The Lagrange function corresponding to \eqref{eq:regularizedprimalopt} is then
\begin{align*}
    \mathcal{L}(\mathbf{M},\mathbf{\Lambda}) &= \sum_{t=0}^{T-2}\left( \langle \tilde{C}_t, \tilde{M}_t \rangle  + \epsilon D(\tilde{M}_t|\tilde{P}_t) \right) +\sum_{t=0}^{T-1} \left(\langle \hat{C}_t, \hat{M}_t \rangle + \epsilon D(\hat{M}_t|\hat{P}_t) \right)  \\ &+ \rho^T  \left(\ones{n} -\left( \sum_{t=1}^{T-2}  \tilde{M}_t\ones{n_{t+1}} + \sum_{t=1}^{T-1}\hat{M}_t\ones{n_F} \right)  \right) 
    + \lambda_0^T\left( \mathbf{1}_{n_0} - \tilde{M}_0 \ones{n_1} - \hat{M}_0\ones{n_F}  \right) \\
    &+ \sum_{t=1}^{T-2} \lambda_t^T \left(  \tilde{M}_{t-1}^T\ones{n_{t-1}} - \tilde{M}_t\ones{n_{t+1}} - \hat{M}_t\ones{n_F}\right)  \\
    &+  \lambda_{T-1}^T \left(\tilde{M}_{T-2}^T\ones{n_{T-1}} - \hat{M}_{T-1}\ones{n_F}\right) + \lambda_T^T \left( 
 \left(\sum_{t=0}^{T-1} \hat{M}^T_t \ones{n_t} \right) - \ones{n_F}  \right)
\end{align*}
and the Lagrange dual problem is
\begin{align}
\label{eq:dualproblem}
   \max_{\mathbf{\Lambda} | \rho, \lambda_0, \lambda_T \geq 0} \enspace \inf_{\mathbf{M}\geq 0} \mathcal{L}(\mathbf{M},\mathbf{\Lambda}).
\end{align}
The key to deriving the algorithm presented in this work lies in leveraging the block-like structure in the problem induced by the different transport stages. In the following proposition, each transport plan is expressed in terms of the dual variables, and the dual problem is expressed on a form which emphasizes this block-structure. 

\vspace{0.5cm}

\begin{proposition}
\label{sec:factorizationproposition}  Assume there exists a feasible solution to \eqref{eq:regularizedprimalopt} with $(\tilde{M}_t)_{ij} > 0$ whenever $(\tilde{C}_t)_{ij} < \infty$ and $(\hat{M}_t)_{ij} > 0$ whenever $(\hat{C}_t)_{ij} < \infty$. Then, the optimal transport plans in \eqref{eq:regularizedprimalopt} may be factorized as
\begin{equation*}
   \boxed{ \begin{aligned}
   \tilde{M}_t &= \tilde{K}_t \odot \tilde{U}_t,  \quad \text{for} \enspace t = 0,1,\dots,T-2, \\
      \hat{M}_t &= \hat{K}_t \odot \hat{U}_t, \quad \text{for} \enspace t = 0,1,\dots,T-1,
      \end{aligned} }
\end{equation*}
 where 
\begin{align*}
&\tilde{K}_t = \tilde{P}_t \odot \exp(-\tilde{C}_t/\epsilon), \quad \text{for} \enspace t = 0,1,\dots,T-2,
\\ 
&\hat{K}_t = \hat{P}_t \odot \exp(-\hat{C}_t/\epsilon),  \quad \text{for} \enspace t = 0,1,\dots,T-1, \\
&\tilde{U}_t =  (u_t \odot s_t) v_{t+1}^T, \quad \text{for} \enspace t = 0,1,\dots,T-2,\\
&\hat{U}_t =  (u_t \odot s_t)v_T^T, \quad \text{for} \enspace t = 0,1,\dots,T-1,
\end{align*}
and where
\begin{align*}
    &u_t = \exp(\lambda_t / \epsilon) \enspace \, \quad  t = 0,1\dots, T \\
&v_t =
\exp(- \lambda_t / \epsilon)\quad t = 0,1,\dots, T
\\ 
    &s_t = \begin{cases}
         \ones{n_0} \quad &t = 0\\
        \exp(\rho/\epsilon) &t= 1,\dots, T-1.
    \end{cases}
\end{align*}
Moreover, strong-duality holds, and the dual function of \eqref{eq:regularizedprimalopt} may, up to a constant, be expressed as
\begin{align*}
\varphi(\mathbf{\Lambda})  = &-\epsilon \big(e^{\lambda_0/\epsilon}\big)^T\left(\tilde{K}_0 e^{-\lambda_1/\epsilon} + \hat{K}_0 e^{-\lambda_T/\epsilon} \right) \\
&-\epsilon \sum_{t=1}^{T-2} \big(e^{\lambda_t/\epsilon} \odot e^{ \rho/\epsilon}\big)^T
 \big(\tilde{K}_t e^{ - \lambda_{t+1}/\epsilon}
+ \hat{K}_t e^{-\lambda_T/\epsilon}\big) \\ &-\epsilon  \big(e^{\lambda_{T-1}/\epsilon} \odot e^{ \rho/\epsilon}\big)^T  \left( \hat{K}_{T-1} e^{-\lambda_T/\epsilon} \right)  \\
&+ \ones{n_0}^T \lambda_0 + \ones{n_F}^T\lambda_T + \ones{n}^T \rho.
\end{align*}
\end{proposition}
\begin{proof} Assume first that $(\tilde{C}_t)_{ij}<\infty$, $(\hat{C}_t)_{ij}<\infty$ for all $i, j$ and $t$, in which case there always exists a strictly positive feasible solution to \eqref{eq:regularizedprimalopt} (e.g., sending mass uniformly to all other cells in each time point). Since the objective \eqref{eq:regularizedprimalopt_a} is strictly convex, if a minimizer $\mathbf{M}^*$ exists, it is unique. Moreover, since the entropy-regularization terms have infinite, negative, slope at $(\tilde{M}_t)_{ij} = 0$, $(\hat{M}_t)_{ij} = 0$, for $i=1,\dots,n_t$, $j=1,\dots,{n_{t+1}}$, and $t = 0,\dots, T-1$, it follows that all elements in the optimal solution will be strictly positive. Therefore, since the Lagrange function is continuously differentiable with respect to $\mathbf{M}$, we know that any minimizer $\mathbf{M}^*$ must satisfy
\begin{equation*}
   \nabla_\mathbf{M} \mathcal{L} ( \mathbf{M}^*) = 0.   
\end{equation*}
Differentiating $\mathcal{L}$ with respect to the variables $(\tilde{M}_0)_{ij}$ for $i = 1,\dots, n_0, j = 1,\dots,n_1$ and setting the derivative to zero yields
\begin{align*}
 (\tilde{M}_0)_{ij} = (\tilde{P}_0)_{ij}\exp \left( -(\tilde{C}_0)_{ij}/\epsilon + (\lambda_0)_i/\epsilon -(\lambda_1)_j/\epsilon \right), 
\end{align*}
or, equivalently, 
\begin{align*}
    \tilde{M}_0 &= \text{diag}(u_0) \tilde{P}_0 \odot\exp(-\tilde{C}_0/\epsilon) \text{diag}(v_1) =  \tilde{P}_0 \odot \exp(-\tilde{C}_0/\epsilon)  \odot \left(u_0 v_1^T\right) = \tilde{K}_0 \odot \left(u_0 v_1^T\right) \\
    &= \{ s_0 = \ones{n_0}\} =  \tilde{K}_0 \odot \left((u_0 \odot s_0)v_1^T\right).
\end{align*}
Doing the same for $(\hat{M}_0)_{ij}$, for $i = 1,\dots,n_0$, $j=1,\dots,n_F$, yields
\begin{equation*}
     (\hat{M}_0)_{ij} = (\hat{P}_0)_{ij} \exp\left( -(\hat{C}_0)_{ij}/\epsilon + (\lambda_0)_i/\epsilon - (\lambda_T)_j/\epsilon \right),
\end{equation*}
or,  
\begin{equation*}
    \hat{M}_0 = \hat{K}_0\left(u_0 v_T^T \right) = \{ s_0 = \ones{n_0}\} = \hat{K}_0\left( (u_0 \odot s_0) v_T^T \right).
\end{equation*} 
Similarly, for $t = 1,\dots, T-2$ we obtain
\begin{equation*}
    (\tilde{M}_t)_{ij} = 
        (\tilde{P}_t)_{ij}\exp \left( -(\tilde{C}_t)_{ij}/\epsilon + \rho_i/\epsilon + (\lambda_t)_i/\epsilon - (\lambda_{t+1})_j/\epsilon  \right),
\end{equation*}
 for $i = 1,\dots, n, j = 1,\dots,n$, and 
 \begin{equation*}
    (\hat{M}_t)_{ij} = 
        (\hat{P}_t)_{ij}\exp \left( -(\hat{C}_t)_{ij}/\epsilon + \rho_i/\epsilon + (\lambda_t)_i/\epsilon - (\lambda_{T})_j/\epsilon  \right),
 \end{equation*}
  for $i = 1,\dots, n, j = 1,\dots,n_F$. Hence, for $t = 1,\dots, T-2$, we have that
\begin{align*}
    \tilde{M}_t  &=   \tilde{K}_t \odot \left( (u_t \odot s_t) v_{t+1}^T \right), \\
    \hat{M}_t  &= \hat{K}_t \odot \left( (u_t \odot s_t) v_{T}^T \right).
\end{align*}
Finally, for $t=T-1$, we obtain
\begin{equation*}
    (\hat{M}_{T-1})_{ij} = (\hat{P}_{T-1})_{ij} \exp \left( -(\hat{C}_{T-1})_{ij}/\epsilon + \rho_i/\epsilon + (\lambda_{T-1})_i/\epsilon - (\lambda_{T})_i/\epsilon \right),
\end{equation*} 
for $i = 1,\dots, n, j = 1,\dots,n_F$, which implies
\begin{equation*}
    \hat{M}_{T-1} =  \hat{K}_{T-1}\odot  \left((u_{T-1}\odot s_{T-1})v_T^T \right).
\end{equation*}
It follows that in the optimal solution,  the entropy-regularization for $\tilde{M}_0$ reduces to
\begin{align*}
\epsilon D(\tilde{M}_0|\tilde{P}_0) =   &\epsilon \sum_{i=1}^{n_0} \sum_{j=1}^{n} (\tilde{M}_0)_{ij} \log(\tilde{M}_0)_{ij} - (\tilde{M}_0)_{ij} \log(\tilde{P}_0)_{ij}    -  (\tilde{M}_0)_{ij} +  (\tilde{P}_0)_{ij}   \\
= &-\sum_{i=1}^{n_0} \sum_{j=1}^{n}  (\tilde{M}_0)_{ij}(\tilde{C}_0)_{ij}  + \sum_{i=1}^{n_0}  \sum_{j=1}^{n}(\tilde{M}_0)_{ij}(\lambda_0)_i  - \sum_{i=1}^{n_0}  \sum_{j=1}^{n}(\tilde{M}_0)_{ij}(\lambda_1)_j \\
&- \epsilon \sum_{i=1}^{n_0}  \sum_{j=1}^{n} (\tilde{M}_0)_{ij} + \epsilon \sum_{i=1}^{n_0}  \sum_{j=1}^{n} (\tilde{P}_0)_{ij}  \\
= &-\langle \tilde{C}_0, \tilde{M}_0 \rangle + \lambda_0^T \tilde{M}_0\ones{n_1} - \lambda_1^T \tilde{M}_0^T \ones{n_0} -  \epsilon\ones{n_0}^T \tilde{M}_0 \ones{n_1}  +  \epsilon\ones{n_0}^T \tilde{P}_0 \ones{n_1}.
\end{align*}
Similarly, we get that
\begin{equation*}
    \epsilon D(\hat{M}_0|\hat{P}_0) =  -\langle \hat{C}_0, \hat{M}_0 \rangle + \lambda_0^T \hat{M}_0\ones{n_F} - \lambda_1^T \hat{M}_0^T \ones{n_0} -  \epsilon\ones{n_0}^T \hat{M}_0 \ones{n_F}  +  \epsilon\ones{n_0}^T \hat{P}_0 \ones{n_F}
\end{equation*}
For $t = 1,\dots,T-2$, we have
\begin{align*}
    \epsilon D(\tilde{M}_t|\tilde{P}_t) =   &-\langle \tilde{C}_t, \tilde{M}_t \rangle + \rho^T \tilde{M}_t\ones{n_{t+1}} + \lambda_t^T \tilde{M}_t\ones{n_{t+1}} - \lambda_{t+1}^T \tilde{M}_t^T \ones{n_t} \\
    &-  \epsilon\ones{n_t}^T \tilde{M}_t \ones{n_{t+1}}  + \epsilon\ones{n_t}^T \tilde{P}_t \ones{n_{t+1}},\\
    \epsilon D(\hat{M}_t|\hat{P}_t) =   &-\langle \hat{C}_t, \hat{M}_t \rangle + \rho^T \hat{M}_t\ones{n_{T}} + \lambda_t^T \hat{M}_t\ones{n_{T}} - \lambda_{T}^T \hat{M}_t^T \ones{n_t} \\
    &-  \epsilon\ones{n_t}^T \hat{M}_t \ones{n_{T}}  + \epsilon\ones{n_t}^T \hat{P}_t \ones{n_{t+1}},
\end{align*}
and, for $t = T-1$,
\begin{align*}
    \epsilon D(\tilde{M}_{T-1}|\tilde{P}_{T-1}) =   &-\langle \tilde{C}_{T-1}, \tilde{M}_{T-1} \rangle + \rho^T \tilde{M}_{T-1}\ones{n} + \lambda_{T-1}^T \tilde{M}_{T-1}\ones{n}  \\
    &-  \epsilon\ones{n}^T \tilde{M}_{T-1} \ones{n_{T}}  + \epsilon\ones{n}^T \tilde{P}_{T-1} \ones{n},  \\
    \epsilon D(\hat{M}_{T-1}|\hat{P}_{T-1}) =   &-\langle \hat{C}_{T-1}, \hat{M}_{T-1} \rangle + \rho^T \hat{M}_{T-1}\ones{n} + \lambda_{T-1}^T \hat{M}_{T-1}\ones{n} - \lambda_{T}^T \hat{M}_{T-1}^T \ones{n} \\
    &-  \epsilon\ones{n}^T \hat{M}_{T-1} \ones{n_{T}} + \epsilon\ones{n}^T \hat{P}_{T-1} \ones{n}.
\end{align*}
Thus, when we substitute the optimal transport plans (in terms of the dual variables) into the Lagrange function, many of the terms cancel, and we obtain the following dual function (up to an additive constant):
\begin{align*}
\varphi(\mathbf{\Lambda}) = -\epsilon \left( \sum_{t=0}^{T-2}  \langle\tilde{K}_t,\tilde{U}_t\rangle  -  \sum_{t=0}^{T-1}  \langle \hat{K}_t,\hat{U}_t\rangle \right)  + \ones{n_0}^T \lambda_0 + \ones{n_F}^T\lambda_T + \ones{n}^T \rho.
\end{align*}
Alternatively, it can be expressed explicitly in terms of the dual variables $\mathbf{\Lambda}$,
\begin{align*}
\varphi(\mathbf{\Lambda})  = &-\epsilon \big(e^{\lambda_0/\epsilon}\big)^T\left(\tilde{K}_0 e^{-\lambda_1/\epsilon} + \hat{K}_0 e^{-\lambda_T/\epsilon} \right) \\
&-\epsilon \sum_{t=1}^{T-2} \big(e^{\lambda_t/\epsilon} \odot e^{ \rho/\epsilon}\big)^T
 \big(\tilde{K}_t e^{ - \lambda_{t+1}/\epsilon}
+ \hat{K}_t e^{-\lambda_T/\epsilon}\big) \\ &-\epsilon  \big(e^{\lambda_{T-1}/\epsilon} \odot e^{ \rho/\epsilon}\big)^T  \left( \hat{K}_{T-1} e^{-\lambda_T/\epsilon} \right)  \\
&+ \ones{n_0}^T \lambda_0 + \ones{n_F}^T\lambda_T + \ones{n}^T \rho.
\end{align*}
Since there exists a strictly positive feasible solution, we know there exists a point in the relative interior of the feasible set (Slater's condition). Thus, by Slater's theorem, strong duality holds (see \cite{boyd2004convex} Section~5.3.2. for a proof). 

In the case of infinite cost elements $(\tilde{C}_t)_{ij} =  \infty$ or $(\hat{C}_t)_{ij} = \infty$ for some $i,j$ and $t$, the corresponding elements in $\tilde{M}_t$ or $\hat{M}_t$ need to be 0 for the transport to be feasible, and such transports may thus be removed from the formulation without changing the problem. If there is a feasible solution with all remaining variables strictly positive, then Slater's condition still holds.

 \end{proof}
\begin{remark} For certain choices $(\tilde{C}_t)_{ij} =  \infty$ or $(\hat{C}_t)_{ij} = \infty$,  the optimal solution might require setting other transport elements to 0 even though the corresponding costs are finite, meaning Slater's condition is not fulfilled; this may for example happen if an entire row or column is filled with only infinite cost elements. This situation corresponds to infinite dual variables and is avoided by the assumption that $(\tilde{M}_t)_{ij} > 0$ whenever $(\tilde{C}_t)_{ij} < \infty$, and $(\hat{M}_t)_{ij} > 0$ whenever $(\hat{C}_t)_{ij} < \infty$ in \hyperref[sec:factorizationproposition]{Proposition~\ref{sec:factorizationproposition}}.
\end{remark}

\subsection{Computation of optimal transport plans via block-coordinate ascent on the dual function}
We solve the dual problem derived above in a block-wise fashion: each ``block 
 update'' corresponds to optimizing the dual function with respect to a particular dual variable vector in each step, while keeping the others fixed. This procedure then directly leads to \hyperref[alg:baseblockascent]{Algorithm \ref{alg:baseblockascent}}. From the optimal dual variables, the optimal transport plans are readily obtained using the results in \hyperref[sec:factorizationproposition]{Proposition \ref{sec:factorizationproposition}}.

Taking the gradient with respect to $\rho$  yields
\begin{equation*}
        \nabla_{\rho} \varphi = \ones{n} -  e^{\rho/\epsilon} \odot \Bigg[ e^{\lambda_{T-1}/\epsilon} \odot \left( \hat{K}_{T-1} e^{-\lambda_T/\epsilon} \right) + \sum_{t=1}^{T-2}    e^{\lambda_t/\epsilon} \odot  \big(\tilde{K}_te^{-\lambda_{t+1}/\epsilon} +    \hat{K}_t e^{-\lambda_T/\epsilon}\big)\Bigg],
\end{equation*}
or, defining $s= \exp(\rho/\epsilon)$, we can express this in terms of the utility variables:
\begin{equation*}
         \nabla_{\rho} \varphi = \ones{n} -  s \odot  \left[  u_{T-1}\odot \left(  \hat{K}_{T-1} v_T \right) 
 + \sum_{t=1}^{T-2} u_t \odot  \big(\tilde{K}_t v_{t+1} +    \hat{K}_t v_T\big)\right].
\end{equation*}
Setting this to zero and solving for $s$ leads to
\begin{equation*}
    s = \ones{n} \oslash \left[ u_{T-1}\odot \left(  \hat{K}_{T-1} v_T \right) + \sum_{t=1}^{T-2} u_t \odot  \big(\tilde{K}_t v_{t+1} +    \hat{K}_t v_T\big)\right].
\end{equation*}
Note however that we also have the constraints $\rho \geq 0$. To ensure that this constraint is not violated, the update rule is modified and becomes:
\begin{align}
    s = \max\left\{\ones{n}, \ones{n} \oslash \left[  u_{T-1}\odot \left(  \hat{K}_{T-1} v_T \right) + \sum_{t=1}^{T-2} u_t \odot  \big(\tilde{K}_t v_{t+1} +    \hat{K}_t u_T\big)\right] \right\}.
\end{align}
Taking the gradient with respect to $\lambda_0$  yields
\begin{equation*}
    \nabla_{\lambda_0} \varphi = \ones{n_0} -\big(\tilde{K}_0 e^{-\lambda_1/\epsilon} + \hat{K}_0 e^{\lambda_T/\epsilon}\big) \odot e^{\lambda_0/\epsilon} = \ones{n_0} -(\tilde{K}_0v_1 + \hat{K}_0v_T)\odot u_0.
\end{equation*}
Again, imposing that this should be zero (and ensuring $\lambda_0 \geq 0)$ we may solve for $u_0$ to obtain:
\begin{align}
     u_0 = \max\big\{\ones{n_0}, \ones{n_0}\oslash (\tilde{K}_0v_1 + \hat{K}_0 v_T)\big\}.
\end{align}
Similarly, differentiating with respect to $\lambda_1$, one obtains
\begin{align*}
    \nabla_{\lambda_1} \varphi &= \big(\tilde{K}_0^T e^{\lambda_0/\epsilon}\big) \odot e^{-\lambda_1/\epsilon} - \big(\tilde{K}_1e^{-\lambda_{t+1}/\epsilon}  + \hat{K}_1e^{\lambda_T/\epsilon}\big)\big(e^{\lambda_t/\epsilon} \odot e^{\rho/\epsilon}\big) \\
    &= \big(\tilde{K}_0^T u_0\big) \oslash u_1 - \big(\tilde{K}_1v_2  + \hat{K}_1v_T\big)\big(u_1 \odot s\big) = 0,
\end{align*}
and, solving for $u_1$,
\begin{align}
    u_1 = \left\{\left(\tilde{K}_0^T u_0\right) \oslash \left[ s \odot \left(\tilde{K}_1v_2  + \hat{K}_1v_T\right) \right]\right\}^{1/2}.
\end{align}
Repeating this procedure for $t=2,\dots, T-2$, one obtains
\begin{align}
    u_t  =  \left\{ \left( \tilde{K}_{t-1}^T \left(u_{t-1} \odot  s \right)\right) \oslash \left[s  \odot  \left(\tilde{K}_t v_{t+1}+  \hat{K}_t v_T \right) \right ]\right\}^{1/2},
\end{align}
and for the final two blocks, 
\begin{align}
    u_{T-1}  &=  \left\{ \left( \tilde{K}_{T-2}^T \left(u_{T-2} \odot  s\right) \right) \oslash \left(s  \odot \left( \hat{K}_{T-1} v_T  \right) \right) \right\}^{1/2} \\
    u_T &= \max \left\{\ones{n_F},  \ones{n_F} \oslash \left[ \hat{K}_0^T u_0 + \sum_{t=1}^{T-1} \hat{K}_t^T (u_t \odot s) \right]\right \}.
\end{align}
The algorithm is summarized in \hyperref[alg:baseblockascent]{Algorithm \ref{alg:baseblockascent}}.

In practice, we define the following cost matrices:
\begin{align}
\label{eq:costdefinitions}
    \tilde{C}_0 &:= c(\mathcal{X}_0, \mathcal{X})\\
    \hat{C}_0 &:= \left[\infty \right]_{i=1, j=1}^{n_0, n_F}\\
        \tilde{C}_t &:= c(\mathcal{X}, \mathcal{X}) + \eta I, \enspace \quad \text{for $t = 1,\dots,T-2$}\\
    \hat{C}_t &:= c(\mathcal{X}, \mathcal{X}_F), \quad \text{for $t = 1,\dots,T-1$},
\end{align}
where the cost function $c(\cdot, \cdot)$ is defined in \eqref{eq:matrixcost}, corresponding to the squared Euclidean distances between cell states, $\eta$ is a large number which penalizes self-transports and $I$ is the $n$ by $n$ identity matrix. We effectively forbid such transport by setting the diagonal to $\texttt{numpy.inf}$ in our Python implementation. Note that our choice of $\hat{C}_0$ corresponds to an assumption that no initial states takes a terminal differentiation step in the first stage of transport.

\subsection{Reduction of entropy-regularization through a proximal point scheme}
\label{sec:proximalpoint}
In our derivation of \hyperref[alg:baseblockascent]{Algorithm \ref{alg:baseblockascent}}, the regularization parameter $\epsilon$ was taken as a given constant. However, the solution indeed depends on $\epsilon$, and  setting $\epsilon = 0$ reverts the regularized problem back into the original problem $(P)$. For this reason, it may be tempting to choose $\epsilon$ very small. However, the block-coordinate update scheme derived above is not defined for $\epsilon = 0$, as this would cause all elements in the $K_t$-matrices to become 0, yielding infeasible transport plans for any finite choice of the dual variables; and even for strictly positive values, choosing $\epsilon$ too small leads to numerical instability if enough elements in the $K_t$-matrices become too small to be represented as non-zero values. For small $\epsilon > 0$, this issue can in principle be solved by performing all computations on the log-transformed utility variables instead \cite{schmitzer2019}. However, this greatly increases the computational complexity of the algorithm   \cite{pmlr-v115-xie20b}.

\begin{algorithm}[H]
	\caption{Block-coordinate ascent algorithm for solving \hyperref[eq:regularizedprimalopt]{(\ref{eq:regularizedprimalopt})}  } \label{alg:baseblockascent}
	\textbf{Given:}  A data set $\mathcal{D} = \{ x^{(i)} |i = 1,\dots, N\}$ of the cellular states (e.g., gene expression levels) of $N$ cells partitioned into three disjoint ordered subsets $\mathcal{X}_0$ (initial states), $\mathcal{X}$ (intermediate states), $\mathcal{X}_F$ (terminal states). \\
 \textbf{Choose:} parameter values for the number of time steps, $T$, and entropy-regularization, $\epsilon \in (0,\infty)$. \\
	\textbf{Pre-compute:} matrices $\tilde{K}_t$ for $t = 0,1,2,\dots, T-2$, and $\hat{K}_t$ for $t = 1,2,\dots, T-1$ (Note: we set all elements of $\hat{K}_0 = [ (\hat{K}_0)_{ij}]_{i=1,j=1}^{n_0, n_F}$ to zero in our implementation, corresponding to infinite costs, $(\hat{C}_0)_{ij} = \infty$, $i=1,\dots,n_0$, $j=1,\dots,n_F$). \\
	\textbf{Initialize}  $u_t = \ones{n_t}$, for $t = 0,1,\dots, T$, and  $s= \ones{n}$. 
	\begin{algorithmic}
		\While{not converged}  
		\For{$t = 0,1, \dots, T$}
		\State $v_t \gets \ones{n_t} \oslash u_t$
		\EndFor 
		\State $s \gets  \max \left\{\ones{n}, \ones{n} \oslash \left [  u_{T-1}  \odot \left(\hat{K}_{T-1} v_T\right) + \sum_{t=1}^{T-2} u_t \odot \left(\tilde{K}_t v_{t+1}+ \hat{K}_t v_{T}\right)  \right] \right\}$
		\State $u_0 \gets  \max \left\{\ones{n_0}, \ones{n_0} \oslash \left(\tilde{K}_0 v_1 + \hat{K}_0 v_T\right) \right\}$
		\State $u_1 \gets   \left\{ \left(\tilde{K}_0^T u_0 \right) \oslash \left[s  \odot  \left(\tilde{K}_1 v_2 +  \hat{K}_1v_T\right) \right]\right\}^{1/2}$
		\For{$t = 2, \dots, T-2$}
		\State $u_t \gets  \left\{ \left( \tilde{K}_{t-1}^T \left(u_{t-1} \odot  s \right) \right) \oslash \left[s  \odot  \left(\tilde{K}_t v_{t+1}+  \hat{K}_t v_T \right) \right]\right\}^{1/2}$ 
		\EndFor
		\State $u_{T-1} \gets  \left\{ \left( \tilde{K}_{T-2}^T \left(u_{T-2} \odot  s \right) \right) \oslash \left[s  \odot \left(\hat{K}_{T-1} v_T \right) \right]\right\}^{1/2}$
		\State $u_T \gets  \max \left\{\ones{n_F},  \ones{n_F} \oslash \left[ \hat{K}_0^T u_0 + \sum_{t=1}^{T-1} \hat{K}_t^T \left(u_t \odot s \right) \right]\right \}  $ 
		\EndWhile  
		\State{$\tilde{M}_0 \gets   \tilde{K}_0 \odot \left(u_0 v_1^T \right) $}
             \State{$\hat{M}_0 \gets \hat{K}_0 \odot \left(u_0 v_T^T \right) $}
		\For{$t = 1,\dots,T-2$}
		\State{$\tilde{M}_t \gets   \tilde{K}_t \left((u_t \odot s)  v_{t+1}^T \right)$ }
          \State{$\hat{M}_t \gets   \hat{K}_t \left((u_t \odot s)  v_T^T \right)$ }
		\EndFor 
		\State $\hat{M}_{T-1} \gets   \hat{K}_{T-1} \left((u_t \odot s)  v_T^T \right)$
	\end{algorithmic}
\end{algorithm}

To avoid having to choose $\epsilon$ too large (and thus straying too far from optimality in the original problem), we nested the block-coordinate ascent method in an outer proximal point scheme, based on the approach introduced in Xie et al.  \cite{pmlr-v115-xie20b}. The proximal point scheme works by taking a previously computed solution as a prior in the entropy term. The solution to this new problem can be shown to correspond to a different parameter $\epsilon'$ which is strictly smaller than the parameter used for the original problem.  To see this, consider $\tilde{M}_t^p$ as prior optimal transport plan in iteration $p$ of the proximal point scheme, obtained after solving \eqref{eq:regularizedprimalopt} with  prior transport plan matrix $\tilde{M}_t^{p-1}$. Then, from \hyperref[sec:factorizationproposition]{Proposition \ref{sec:factorizationproposition}}, we have that
\begin{align}
    \tilde{M}_t^p = \tilde{K}^p_t \odot \tilde{U}_t^p = \tilde{M}^{p-1}_t \odot \exp\left( -\tilde{C}_t/\epsilon  \right)\odot \tilde{U}_t^p .
\end{align}
In proximal point iteration $p+1$, we solve \eqref{eq:regularizedprimalopt} again using entropy-regularization $D\left(\tilde{M}^{p+1}_t | \tilde{M}_t^p\right)$. The optimal transport plan can be expressed as
\begin{align}
    \tilde{M}_t^{p+1} &=  \tilde{M}_t^p \odot \exp\left(  -\tilde{C}_t/\epsilon  \right) \odot \tilde{U}_t^{p+1}= \\
    &= \left(\tilde{M}^{p-1}_t \odot \exp\left( -\tilde{C}_t/\epsilon  \right)\odot \tilde{U}_t^p \right) \odot ( -\tilde{C}_t / \epsilon) \odot \tilde{U}_t^{p+1} \\
 &=\tilde{M}^{p-1}_t \odot \exp\left (-\frac{\tilde{C}_t}{\frac{\epsilon}{2}} \right) \odot \tilde{U}_t^{p} \odot \tilde{U}_t^{p+1}.
\end{align}
Hence, by induction, we find that the transport plans obtained in this way in proximal point iteration $p$ can be expressed
\begin{align}
    \tilde{M}_t^{p} = \tilde{P}_t^0 \odot \exp \left ( -\frac{\tilde{C}_t}{\frac{\epsilon}{p+1}} \right)  \odot  \mathcal{U}_t^p,
\end{align}
where $\tilde{\mathcal{U}}_t := \tilde{U}_t^{p} \odot \tilde{U}_t^{p-1} \odot \tilde{U}_t^{p-2} \odot \cdots \odot \tilde{U}_t^0$, and $ \tilde{P}_t^0 $ is an initial prior regularization matrix (not necessarily a transport 
plan). This corresponds to regularization parameter $\epsilon / (p+1)$, which decreases as the number of outer iterations $p$ increases.  More generally, if, in iteration $p$ of the proximal point scheme, \eqref{eq:regularizedprimalopt} is solved using  regularization parameter $\epsilon_{p}$ and if the prior transport plan had been obtained by solving \eqref{eq:regularizedprimalopt} with regularization parameter $\epsilon_{p-1}$, then the new optimal transport plans will correspond to regularization parameter
\begin{align}
    \epsilon' =  \frac{1}{\frac{1}{\epsilon_{p-1}}+ \frac{1}{\epsilon_{p}}}.
\end{align}
By induction, the effective regularization parameter after $p$ iterations is
\begin{align}
\label{eq:sequence_of_epsilon}
  \epsilon' = \frac{1}{\frac{1}{\epsilon_p} +\frac{1}{\epsilon_{p-1}} + \cdots \frac{1}{\epsilon_{0}} },
\end{align}
with $\epsilon_{0}$  corresponding to the first transport plan prior obtained by solving \eqref{eq:regularizedprimalopt} with some initial prior regularization matrix. 

In our implementation, we initiate the proximal point scheme with the  ``uninformative'' regularization priors $\tilde{P}_t = \ones{n_t}\ones{n_{t+1}}^T$, $t = 0,1,\dots,T-2$ and $\hat{P}_t = \ones{n_t}\ones{n_F}^T$, $t = 0,1,\dots,T-1$. In all subsequent iterations, the optimal transport plans obtained in the previous iteration are taken as regularization priors---effectively decreasing the effects of regularization on the optimal transport plans as shown above. If in any outer iteration the solutions become numerically unstable, we increase the value of the regularization parameter for that iteration. This generates a sequence of strictly decreasing effective regularization parameter values, as in \eqref{eq:sequence_of_epsilon}, and the proximal point scheme terminates when the desired $\epsilon$-threshold is reached. In theory, this method can be used to reduce the effects of regularization to zero, yielding an ``exact'' solution. However, in practice, we found the value $\epsilon_p$ in each iteration may need to grow so large that this convergence would be realized very slowly. Notwithstanding, it allowed us to reduce the regularization parameter with roughly a factor of 5 compared to the original value (see \autoref{tab:regularization-parameter-values}).

\section{Model extension with auxiliary cell states}
We noticed that when a snapshot contains cells for which there are no nearby cell states, the cost of sending mass to the nearest neighbour may be so large that it leads to numerical instability in the Sinkhorn iterations. From a biological perspective, then, it may be invalid to assume that such cells should be allowed to form connections to other cells in the data set. To simultaneously address both of these issues, we extended the standard MultistageOT model by adding three auxiliary states: an auxiliary initial (denoted $\mathcal{A}_0$), intermediate (denoted $\mathcal{A}$) and terminal  (denoted $\mathcal{A}_F$) state respectively (\autoref{fig:network_representation_extension}). Cells in the snapshot $\mathcal{D}$ may form connections with these auxiliary cell states, if it is too costly to form connections with other cells in $\mathcal{D}$. We thus introduce a fixed cost, denoted $Q$, at which cells can transport mass to any of these auxiliary states.

More precisely, we add an auxiliary state to each of the groups $\mathcal{X}_0, \mathcal{X}$ and $\mathcal{X}_F$, and modify the indices accordingly, i.e.,  
\begin{align}
    &\mathcal{X}_0 \rightarrow (x^{(i)}\,|\, i=1,\ldots, n_0 +1) \\
   & \mathcal{X} \rightarrow (x^{(i)}\,|\, i= n_0 + 2,\ldots, N - n_F + 1) \\
   & \mathcal{X}_F \rightarrow (x^{(i)}\,|\, i=N - n_F + 2,\ldots, N + 3),
\end{align}
so that $x^{(n_0+1)}$ corresponds to the auxiliary initial state,  $x^{(N - n_F + 1)}$ corresponds to the auxiliary intermediate state, and $x^{(N + 3)}$ corresponds to the auxiliary terminal state. We add the following costs to these auxiliary states. For $t= 0$,
 
\begin{align}
    (\tilde{C}_0)_{n_0 + 1},j &= Q \quad \text{for } j=1,\dots,N-n_F + 1\\
        (\hat{C}_0)_{n_0 + 1},j &= \begin{cases}
            Q \quad &\text{for } j = n_F + 1\\
    \infty &\text{otherwise}.
        \end{cases}
\end{align}
For $t = 1,\dots,T-2$,
\begin{align}
        (\tilde{C}_t)_{N-n_0-n_F+1,j} &=  \begin{cases}
            Q &\text{for } j = 1,\dots, N-n_0-n_F\\
            \infty  &\text{for } j = N-n_0-n_F+1
        \end{cases}\\
           (\hat{C}_t)_{N-n_0-n_F+1,j} &=  \begin{cases}
            Q &\text{for } j =  n_F + 1\\
            \infty  &\text{otherwise}
        \end{cases}\\
        (\tilde{C}_t)_{i,N-n_0-n_F+1} &=  
            Q \quad \text{for } i = 1,\dots, N-n_0-n_F \\
           (\hat{C}_t)_{i, n_F+1} &= Q \quad \text{for } i = 1,\dots, N-n_0-n_F+1,
\end{align}
and, for $t=T-1$,
\begin{align}
               (\hat{C}_{T-1})_{N-n_0-n_F+1,j} &=  \begin{cases}
            Q &\text{for } j =  n_F + 1\\
            \infty  &\text{otherwise}
        \end{cases}\\
           (\hat{C}_{T-1})_{i, n_F+1} &= Q \quad \text{for } i = 1,\dots, N-n_0-n_F+1.
\end{align}

Expanding the utility variables analogously, we add the following update rules in each iteration:

\begin{align}
    s_{N - n_F - n_0 + 1} &\gets 1 \\
    (u_0)_{n_0 + 1} &\gets1 \\
    (u_T)_{n_F + 1} &\gets1,
\end{align}
with all other update rules remaining the same. These updates correspond to a lack of transport constraints on the auxiliary states (unlike their counterparts in the original formulation, the auxiliary initial and intermediate states are not constrained to send a certain amount of mass, and the auxiliary terminal state is not constrained to receive a certain amount of mass).

We first validated this extended model on the 2D synthetic data set, as presented in \autoref{fig:synthetic-data-results}, but with an added island of three outlier cell states (\autoref{fig:outlier-results}a). When computing cell fate probabilities, our extended model predicted that all three outliers would end up being absorbed in the ``unknown fate'', represented by the auxiliary terminal fate (\autoref{fig:outlier-results}b). To gauge how the addition of the auxiliary states affected the rest of the solution, we compared the cell fate predictions made on the remaining, non-outlier cells, with the original predictions (\autoref{fig:synthetic-data-results}i) using the total variation (TV) as a measure of the deviations. In doing so, we found that the change in the predicted cell fate probabilities corresponded to a TV  $\leq 3.3\cdot 10 ^{-5}$  in any cell (\autoref{fig:extension_synthetic_data}e), indicating little to no change. We then applied the extended model when no outliers were present, i.e., on the exact data used in \autoref{fig:synthetic-data-results}. In this case, comparing the model with and without auxiliary states, we found the maximum total variation in any cell to be  $< 10^{-7}$, affirming that with no outliers present, the extended model has the desired behavior of simply ignoring the auxiliary states. When applying the extended model on a partition of the barcoded scRNA-seq data set in Weinreb et al. \cite{weinreb2020lineage}, we found that the change in predicted cell fate probabilities for the day 2 cells, as a result of adding auxiliary states, was small for large enough choices of $Q$ (see \autoref{tab:q-values-weinreb}).

\end{bibunit}

\end{document}